\documentclass[12pt,a4paper,twoside,openright,pdftex]{book}
\usepackage{amsmath, amsfonts, amssymb, amsthm, bbm}
\usepackage{ifpdf}
\usepackage{multirow}
\usepackage{dsfont}
\usepackage{enumitem}
\usepackage{bm}
\usepackage{mathtools}
\ifpdf
\usepackage{lettrine}
\usepackage{pgothic}
\usepackage{pdfpages}
\fi
\usepackage{tcolorbox}
\usepackage{physics}
\usepackage[caption=false]{subfig}
\usepackage{listings}
\usepackage[page,toc,titletoc,title]{appendix}
\usepackage{makecell}
 \usepackage{tabularx}
\usepackage{lipsum}
\usepackage{ifthen}
\usepackage[utf8]{inputenc}
\usepackage[T1]{fontenc} 
\usepackage{lmodern}          
\usepackage[nopatch=footnote,final]{microtype}  
\usepackage{booktabs}
\usepackage{tikz} 
\usepackage{lastpage}
\usepackage{hyperref}
\hypersetup{pdftex,
    colorlinks=true,    
    allcolors = black,  
    citecolor=myblue,
    linkcolor=myblue,
    urlcolor=myblue,
    linktoc=page,
    unicode,
    bookmarksopen=true,
    bookmarksopenlevel=0,
    bookmarksnumbered=true}

\usepackage[a4paper,twoside,inner=2.5cm,outer=2.5cm,bottom=3.5cm]{geometry}
\usepackage{fancyhdr}    
\usepackage{graphicx}
\ifpdf
\usepackage{xfp}
\usepackage{eso-pic}
\usepackage{emptypage}        
\fi
\usepackage{setspace}         
\usepackage[margin=0.75cm, font={small,stretch=0.80}]{caption}               
\usepackage{url}        
\usepackage{xurl}    
\usepackage{bookmark}         
\usepackage{xparse}
\usepackage{xspace}
\usepackage[ruled,vlined,linesnumbered]{algorithm2e}

\newcommand{\DropCap}[3]{%
  \ifpdf
  \lettrine[lines=3,lraise=0, nindent=0em, slope=0pt, findent={#3}em]{\pgothfamily #1}{#2}
  \else
  #1#2
  \fi
}
\usetikzlibrary{backgrounds}
\usetikzlibrary{calc}

\makeatletter
\def\thickhrulefill{\leavevmode \leaders \hrule height 2ex \hfill \kern \z@}
\def\@makechapterhead#1{\vspace*{8\p@}{\parindent \z@ \centering \reset@font
        \thickhrulefill\quad
         {\fontfamily{cmr}\selectfont\scshape\large
           \@chapapp\ \thechapter}
        \quad \thickhrulefill
        \par\nobreak
        \vspace*{10\p@}\interlinepenalty\@M
        \hrule
        \vspace*{20\p@}\Huge \bfseries #1\par\nobreak
        \par
        \vspace*{20\p@}\hrule
\vskip 60\p@
  }}
  
\def\@makeschapterhead#1{\vspace*{8\p@}{\parindent \z@ \centering \reset@font
        \thickhrulefill
        \par\nobreak
        \vspace*{10\p@}\interlinepenalty\@M
        \hrule
        \vspace*{20\p@}\Huge \bfseries #1\par\nobreak
        \par
        \vspace*{20\p@}\hrule
        \vskip 60\p@}}

\ifpdf
\newcommand{\computeK}{\fpeval{floor(\value{page}/2)}}
\newcommand{\overlayfooterimage}{\includegraphics[width=1\textwidth]{images/gkp_plots/gkp_delta_\computeK.pdf}
}
\fi
     
\usepackage[numbers,sort&compress]{natbib} 
    
\definecolor{myblue}{HTML}{1d4e9e}
\definecolor{mylightblue}{HTML}{8ea7cf}
\definecolor{myorange}{HTML}{fb5012}
\definecolor{mylightorange}{HTML}{fecbb9}
\definecolor{mybrown}{HTML}{8e505a}
\definecolor{myyellow}{HTML}{FFC857}
\definecolor{quotecol}{HTML}{aafac8}
\newcommand{\pA}{\hyperref[paper:A]{A}\xspace}
\newcommand{\pB}{\hyperref[paper:B]{B}\xspace}
\newcommand{\pC}{\hyperref[paper:C]{C}\xspace}
\newcommand{\pD}{\hyperref[paper:D]{D}\xspace}
\newcommand{\pE}{\hyperref[paper:E]{E}\xspace}

\newcommand{\paperitem}[4]{\noindent
\phantomsection
\label{paper:#1}
\ifpdf
\begin{minipage}[t]{1cm}
    \vspace{-0.2cm}
    \centering
    \bfseries{\Huge #1}
  \end{minipage}\hspace{1em}\begin{minipage}[t]{\dimexpr\linewidth-2cm-1em\relax}
    {\small\emph{#2}}\\
    #3\\
    #4
  \end{minipage}\par
  \vspace{0.7cm}
\else
\phantomsection
\label{paper:#1}
\textbf{Paper #1}: \emph{#2}\\
\noindent #3\\
\noindent #4\\
\vspace{1cm}
\fi
}

\fancypagestyle{veryempty}{\fancyhf{} \fancyfoot{}
  \renewcommand{\headrulewidth}{0pt} 
}
\fancypagestyle{contentspage}{\fancyhead[LO]{}
}

\newcommand\thesistype{PhD}  \newcommand\thesistitle{From simulatability to universality of continuous-variable quantum computers}
\newcommand\thesissubtitle{} \newcommand\thesisauthor{Cameron Calcluth}

 \newcommand{\thesiscoverdescription}{\textbf{Cover}: A plot of the Wigner function of a Gottesman-Kitaev-Preskill state encoding a logical qutrit ``strange'' state with 24 dB squeezing. For an equivalent plot of the same state with 12 dB squeezing, see subplot \textbf{(c.1)} of Fig.~\ref{fig:wigners}.}  

\newcommand\thesisdepartment{Microtechnology and Nanoscience (MC2)}
\newcommand\thesiscity{Göteborg}
 \newcommand\thesisyear{2025}
\newcommand\thesisisbn{ISBN 978-91-8103-192-8}  \newcommand{\thesisissn}{\ifthenelse{\equal{\thesistype}{PhD}}{ISSN 0346-718X}{ISSN 1403-266X}}
\newcommand\thesisnumber{5650}

\title{\thesistitle}
\author{\thesisauthor}

\makeatletter
\definecolor{quotebg}{gray}{0.92}       \definecolor{quotebar}{RGB}{0,102,204}  

\tikzset{fancy quotes/.style={
    text width=\fq@width pt,
    align=justify,
    inner sep=1em,
    anchor=north west,
    minimum width=\linewidth,
},
  fancy quotes marks/.style={
    scale=8,
    text=myblue,
    inner sep=1pt,
  },
  fancy quotes opening/.style={
    fancy quotes marks,
  },
  fancy quotes background/.style={
    show background rectangle,
    inner frame xsep=0,
background rectangle/.style={
      fill=quotebg,
path picture={
        \draw[myblue,line width=9pt]
          (path picture bounding box.north west)
            -- 
          (path picture bounding box.south west);
      },
    },
  },
}

\ifpdf
\newenvironment{fancyquotes}[1][]{\noindent
  \tikzpicture[fancy quotes background]
\node[fancy quotes opening,anchor=north west] (fq@ul) at (0,0) {``};
\tikz@scan@one@point\pgfutil@firstofone(fq@ul.east)
    \pgfmathsetmacro{\fq@width}{\linewidth - 1.8*\pgf@x}
\node[fancy quotes,#1] (fq@txt) at (fq@ul.north west) \bgroup
}{\egroup;        \endtikzpicture
}
\else
\newenvironment{fancyquotes}[1][]{\noindent
 \emph{``#1''}
}{}
\fi

\makeatother

\NewDocumentCommand{\gkpket}{s O{0} o}{\ensuremath{\IfBooleanTF{#1}
      {|( #2 _L )_{\mathrm{GKP}\IfNoValueF{#3}{(#3)}}\rangle}
      {| #2 _{\mathrm{GKP}\IfNoValueF{#3}{(#3)}}\rangle}
}}

\NewDocumentCommand{\phase}{s o}{\ensuremath{\hat V  \IfBooleanT{#1}{^{(\text{CV})}\IfNoValueF{#2}{(#2)}}  \IfBooleanF{#1}{\IfNoValueF{#2}{_{#2}}}
}}

\NewDocumentCommand{\cz}{s}{\ensuremath{\hat {C}_{Z(j,k)}  \IfBooleanT{#1}{^{(\text{CV})}}
}}

\NewDocumentCommand{\fourier}{s o}{\ensuremath{\hat {F}\IfBooleanT{#1}{^{(\text{CV})}}\IfNoValueF{#2}{_{#2}}
}}

\NewDocumentCommand{\disp}{O{\cdot}}{\ensuremath{\hat{D}(#1)
}}

\NewDocumentCommand{\paperclass}{O{X}}{\ensuremath{{\mathcal Q}_{#1}
}}

\NewDocumentCommand{\lX}{o}{\ensuremath{\hat X_{\text{GKP}\IfNoValueF{#1}{(#1)}}
}}

\NewDocumentCommand{\lZ}{o}{\ensuremath{\hat Z_{\text{GKP}\IfNoValueF{#1}{(#1)}}
}}

\NewDocumentCommand{\posket}{O{x}}{\ensuremath{|\hat q = {#1}\rangle
}}

\NewDocumentCommand{\peakspace}{s o}{\ensuremath{\IfBooleanT{#1}{\tilde}{\ell}\IfNoValueF{#2}{_{#2}}
}}

\NewDocumentCommand{\clifford}{s O{d} o}{\ensuremath{\mathcal C^{#2}\IfNoValueF{#3}{_{#3\IfBooleanT{#1}{(\text{CV})}}}\IfNoValueT{#3}{_{\IfBooleanT{#1}{(\text{CV})}}}
}}

\NewDocumentCommand{\squeezing}{O{s}}{\ensuremath{\hat S(#1)
}}

\NewDocumentCommand{\rot}{O{\ensuremath{\theta}}}{\ensuremath{\hat R(#1)
}}

\NewDocumentCommand{\thetatau}{s}{\ensuremath{\tau
}}

\NewDocumentCommand{\thetaz}{s}{\ensuremath{\zeta
}}

\NewDocumentCommand{\symp}{O{F}}{\ensuremath{\mathrm{Sp}(2n,\mathbb #1)
}}

\NewDocumentCommand{\quditd}{o}{\ensuremath{d\IfNoValueF{#1}{_{#1}}
}}

\NewDocumentCommand{\xdisp}{o}{\ensuremath{\hat X(\IfNoValueF{#1}{#1}\IfNoValueT{#1}{\cdot})
}}

\NewDocumentCommand{\zeroto}{o}{\ensuremath{\mathbb N\IfNoValueF{#1}{_{#1}}
}}

\NewDocumentCommand{\poly}{o}{\ensuremath{\operatorname{poly}\IfNoValueF{#1}{(#1)}
}}

\NewDocumentCommand{\pdfunc}{O{x}}{\ensuremath{\operatorname{PDF}(#1)
}}

\NewDocumentCommand{\diagm}{O{\cdot}}{\ensuremath{\operatorname{diag}(#1)
}}

\NewDocumentCommand{\pauligroup}{o o}{\ensuremath{\mathcal G\IfNoValueF{#2}{_{#2}}\IfNoValueF{#1}{^{#1}}
}}

\NewDocumentCommand{\stab}{o}{\ensuremath{\hat g\IfNoValueF{#1}{_{#1}}
}}

\NewDocumentCommand{\charfunc}{s O{\cdot}}{\ensuremath{\chi\left(#2\right)
}}

\NewDocumentCommand{\wig}{s O{\cdot} o}{\ensuremath{\IfBooleanT{#1}{\bar}W\IfNoValueF{#3}{_{#3}}(#2)
}}

\NewDocumentCommand{\zgw}{O{\cdot} o}{\ensuremath{W^{\text{ZG}}\IfNoValueF{#2}{_{#2}}(#1)
}}

\NewDocumentCommand{\weyldv}{o}{\ensuremath{\hat{\bar{T}}\IfNoValueF{#1}{_{#1}}
}}

\NewDocumentCommand{\ppodv}{O{\mathbf t}}{\ensuremath{\bar{\hat A}_{#1}
}}

\NewDocumentCommand{\sumgate}{o}{\ensuremath{\textsc{s}\hat{\textsc{u}}\textsc{m}\IfNoValueF{#1}{_{#1}}
}}

\NewDocumentCommand{\hwop}{s O{\mathbf u}}{\ensuremath{\hat T_{#2}\IfBooleanT{#1}{^{\text{CV}}}
}}

\NewDocumentCommand{\zdisp}{O{\cdot}}{\ensuremath{\hat Z(#1)
}}\NewDocumentCommand{\gkpproj}{o}{\ensuremath{{\hat\Pi_{\text{GKP}}\IfNoValueF{#1}{(#1)}}
}}\NewDocumentCommand{\intreg}{}{\ensuremath{R
}}

\NewDocumentCommand{\genkraus}{s O{i} o}{\ensuremath{\hat P\IfBooleanT{#1}{^\dagger}_{#2}\IfNoValueF{#3}{(#3)}
}}\NewDocumentCommand{\genkrausprime}{s O{i} o}{\ensuremath{\hat P'\IfBooleanT{#1}{^\dagger}_{#2}\IfNoValueF{#3}{(#3)}
}}\NewDocumentCommand{\setgenkraus}{}{\ensuremath{P
}}

\NewDocumentCommand{\mappedstate}{O{P}}{\ensuremath{\hat \rho_L^{(#1)}
}}\NewDocumentCommand{\gkpec}{s O{\mathbf t}}{\ensuremath{\hat K\IfBooleanT{#1}{^\dagger}(#2)
}}

\newcommand{\jacobi}{\ensuremath{\vartheta}}
\newcommand{\sq}{\ensuremath{\Delta}}
\newcommand{\wf}{\ensuremath{\psi}}
\newcommand{\paramsq}{\ensuremath{s}}
\newcommand{\paramrot}{\ensuremath{\theta}}
\newcommand{\quditfactor}{\ensuremath{a}}
\newcommand{\unitary}{\ensuremath{\hat U}}
\newcommand{\sympmat}{\ensuremath{S}}

\newcommand{\paramdisp}{\ensuremath{c}}
\newcommand{\anglesset}{\ensuremath{\Theta}}

\newcommand{\transq}{\ensuremath{\hat z}}
\newcommand{\evolvedQ}{\ensuremath{\hat Q}}
\newcommand{\evolvedQs}{\ensuremath{\hat{\mathbf Q}}}
\newcommand{\customstab}{\ensuremath{g}}
\newcommand{\stabvec}{\ensuremath{\mathbf l}}

\newcommand{\halfsymp}{\ensuremath{\tilde{S}}}
\newcommand{\diagmat}{\ensuremath{\mathcal D}}
\newcommand{\stabgroup}{\ensuremath{\mathcal S}}
\newcommand{\stabspace}{\ensuremath{\mathcal V}}
\newcommand{\sympform}{\ensuremath{\Omega}}
\newcommand{\doffset}{\ensuremath{\varsigma}}

\newtheorem{corollary}{Corollary}
\newtheorem{theorem}{Theorem}

\newcommand{\algcol}[2]{\parbox[t]{.45\linewidth}{\vtop{\hbox{#1}}}\hfill
\parbox[t]{.55\linewidth}{\raggedright\tcp{#2}}}

\begin{document}

\frontmatter                  

\ifpdf
\thispagestyle{veryempty}

{
\begingroup
   \begin{tikzpicture}[
        remember picture,
        overlay,
        anchor=north west,
        inner sep=0pt
        ]
        
    \node (image) at (-12.17cm,4.5cm){\includegraphics[width=23.3cm]{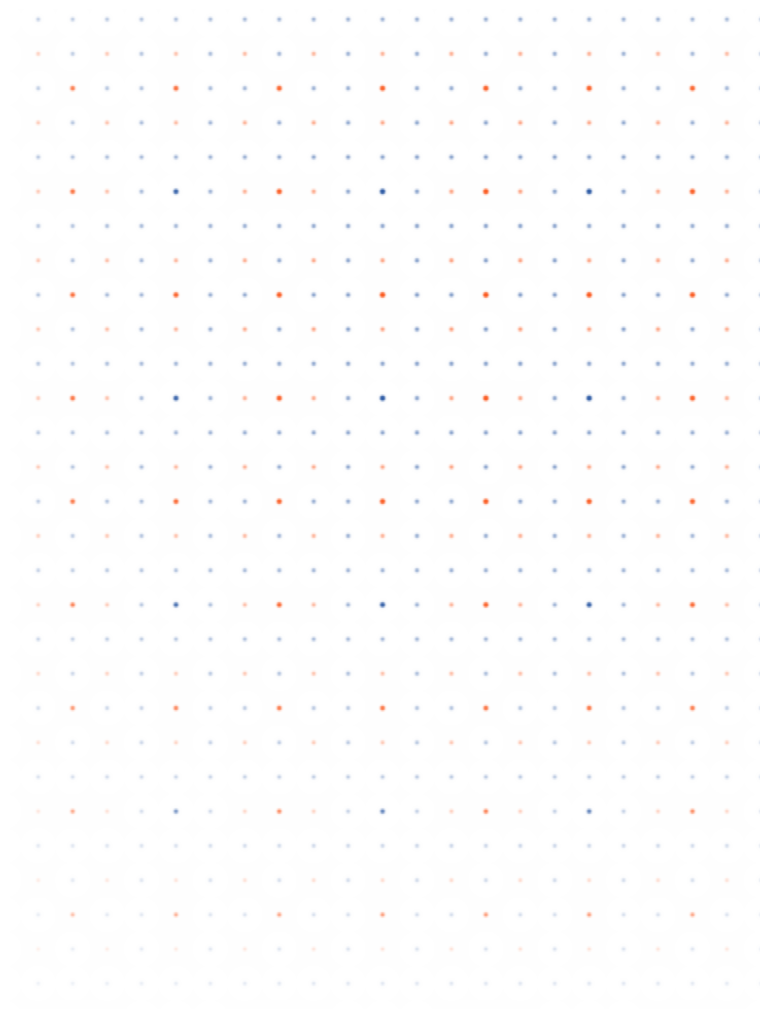}};
    \node[anchor=center,xshift=-0.37cm,yshift=10cm] at (current page.center)
    {\small\scshape Thesis for the degree of Doctor of Philosophy};
    
    \node[anchor=center,xshift=-0.37cm,yshift=6.8cm] at (current page.center)
    {\LARGE From simulatability to universality of};
\node[anchor=center,xshift=-0.37cm,yshift=5.85cm] at (current page.center)
    {\LARGE continuous-variable quantum computers};
\node[anchor=center,xshift=-0.37cm,yshift=2.7cm] at (current page.center)
    {\large \thesisauthor};
     \node[anchor=center,xshift=-0.37cm,yshift=-8cm] at (current page.center)
     {\includegraphics[width=30mm,clip]{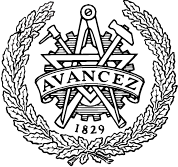}};
     \node[anchor=center,xshift=-0.37cm,yshift=-10cm] at (current page.center)
    {Department of \thesisdepartment};
     \node[anchor=center,xshift=-0.37cm,yshift=-10.5cm] at (current page.center)
    {\scshape CHALMERS UNIVERSITY OF TECHNOLOGY};
     \node[anchor=center,xshift=-0.37cm,yshift=-11cm] at (current page.center)
    {Göteborg, Sweden \thesisyear};
    \end{tikzpicture}
\endgroup
\parindent 0 pt
\centering
\thispagestyle{empty}         
\vspace{-0.2cm}


\vspace{10 pc}
\vspace{0.45cm}

\vskip 1pc
\vspace{0.2cm}
{\large \thesissubtitle}
\vspace{2 pc}

\vfill

\vspace{1cm}

\clearpage
}

\vspace*{50 pt}
{
  \thispagestyle{empty}         

  \parindent 0 pt

  \textbf{\thesistitle}

  \thesissubtitle

  \textsc{\thesisauthor}

  \thesisisbn

  \vskip 2pc
Acknowledgements, dedications, and similar personal statements in this thesis reflect the author's own views.
\vskip 3pc
  \copyright\enskip \textsc{\thesisauthor, \thesisyear.}

  \vskip 3pc

  \ifthenelse{\equal{\thesistype}{PhD}}{Doktorsavhandlingar}{Licentiatavhandlingar}
   vid Chalmers tekniska högskola

  \ifthenelse{\equal{\thesistype}{PhD}}{Ny serie nr \thesisnumber }{Technical report No. \thesisnumber }

  \thesisissn

  \vskip 2pc

  Department of \thesisdepartment

  Chalmers University of Technology

  SE--412 96 Göteborg, Sweden

  Telephone\enskip+ 46 (0) 31 -- 772 1000

  \vfill

  \thesiscoverdescription
\vskip 1pc
  \textbf{Odd-page footer image}: Position-basis wavefunction of a qubit ($d=2$) $0$-logical Gottesman-Kitaev-Preskill state, defined in Eq.~(\ref{eq:realistic-gkp-states}) with squeezing in dB equal to the page number divided by five.

  \vskip 2pc

  \vskip 1pc

  Chalmers Digitaltryck

  Göteborg, Sweden \thesisyear
}

\clearpage
 
\cleardoublepage 
\phantomsection               

\addcontentsline{toc}{chapter}{Abstract}  

\vspace{5cm}
\begingroup
  \let\origclearpage\clearpage
  \let\clearpage\relax

\noindent\textbf{\thesistitle}

\noindent\thesissubtitle

\noindent\textsc{\thesisauthor}

\noindent Applied Quantum Physics Laboratory

\noindent Department of \thesisdepartment

\noindent \MakeUppercase{Chalmers University of Technology}

\fi
\chapter*{Abstract}
\endgroup
\vspace*{-1cm}

\noindent Quantum computers promise to solve some problems exponentially faster than traditional computers, but we still do not fully understand why this is the case. While the most studied model of quantum computation uses qubits, which are the quantum equivalent of a classical bit, an alternative method for building quantum computers is gaining traction. Continuous-variable devices, with their infinite range of measurement outcomes, use systems such as electromagnetic fields. Given this infinite-dimensional structure, combined with the complexities of quantum physics, we are left with a natural question: when are continuous-variable quantum computers more powerful than classical devices?

This thesis investigates this question by exploring the boundary of which circuits are classically simulatable and which unlock a quantum advantage over classical computers.

Prior to the work conducted in this thesis, theorems of classical simulatability of continuous-variable quantum computations relied on positive phase-space representations of all circuit components. Circuits confined to Gaussian elements or those preserving positive Wigner functions are efficiently simulatable, whereas introducing Wigner-negative resources, which indicate non-classical behaviour, is necessary to achieve universality. Although necessary, Wigner negativity does not provide a sufficient condition to achieve universal quantum computation.

In this thesis, a series of proofs are presented demonstrating the efficient simulatability of progressively more complex circuits, even those with high amounts of Wigner negativity. Specifically, circuits initiated with highly Wigner-negative Gottesman-Kitaev-Preskill states, which form a grid-like structure in phase space, can be simulated in polynomial time.

The implications of these results extend to a new fundamental understanding of the computational power of continuous-variable quantum computers. Specifically, we demonstrate the first sufficient condition for achieving universality using continuous-variable devices. These results shine a light on the limits of our current understanding while also paving the way for further exploration of fundamental topics in quantum computing.

\vspace{0.1cm}
\textbf{Keywords:} \textit{Quantum computing, quantum information, continuous-variable quantum computing, quantum optics, quantum advantage, quantum resource theory, Bosonic codes, Gottesman-Kitaev-Preskill states}. 

\ifpdf

\cleardoublepage              \phantomsection               \addcontentsline{toc}{chapter}{Acknowledegments}  \chapter*{Acknowledgements}
First, I want to express my deep gratitude to Giulia Ferrini for all the support, guidance, mentorship, knowledge, and encouragement throughout my PhD. I am extremely thankful for all the hands-on training and the freedom to explore my own ideas, even from the beginning of my PhD. I really couldn't have asked for a better supervisor! 

I am also very grateful for the support of my co-supervisors. In particular, Alessandro Ferraro provided me with extremely generous amounts of time, advice, feedback, and ideas. Laura Garc\'a \'Alvarez provided encouragement, support, and patience, especially during the first half of my PhD.

Furthermore, I am grateful for G\"oran Johansson's time, advice, and guidance as my examiner. I would also like to say thank you to my opponent and the committee members for their time to read my thesis and attend my defence. I really appreciate their support during this crucial stage at the start of my career. I am equally grateful to Daniel Gottesman, who acted as my opponent during my mid-term seminar and with whom I had a chance to discuss and get feedback on my results.

I also really appreciate the time I spent with Xun Gao and Arthur Jaffe, and their invitation to host me for three months to explore a fascinating topic outside the scope of this thesis. I look forward to future collaborations. I have also been supported through stimulating and exciting collaborations with Eric R Anschuetz, Eduardo Alberto Bardales Espa\~na, Alex Maltesson, Nicolas Reichel, Ludvig Rodung, Oliver Hahn, Timo Hillmann, Juani Bermejo-Vega, Ulysse Chabaud, Jack Davis, and more. Thank you to everyone who has been a part of this journey. I particularly thank Alex and Giulia for providing feedback on the many drafts of this thesis, up to the final hours!

I am also extremely grateful to Linda Brånell for providing world-class administrative assistance throughout my PhD. I would also like to thank the many other administrative support staff and other non-academic staff, such as cleaners, who provide the foundations for our ability to do research.

I want to thank my family for believing in me, and my friends for providing me with all the amazing times during my time in Sweden. I am also thankful to Sam for supporting me and starting this journey with me. Finally, I am so grateful to my partner, Esse, for all the love, kindness, and support. 

\vskip 1.5pc

\noindent \thesisauthor

\noindent \thesiscity, May\  2025   
\cleardoublepage              

\fi
\phantomsection               \addcontentsline{toc}{chapter}{List of Publications}  \chapter*{List of Publications}
This thesis is based on the following appended papers.\\ All are available as open-access publications.
\vspace{1cm}

\paperitem{A}{Efficient simulatability of continuous-variable circuits with large Wigner negativity}{Laura~García-Álvarez, Cameron~Calcluth, Alessandro~Ferraro, and \linebreak Giulia Ferrini}{\href{https://doi.org/10.1103/PhysRevResearch.2.043322}{Physical Review Research \textbf{2}, 043322 (2020)}. \\ Reference \cite{garcia-alvarez2020}. [doi:\href{https://doi.org/10.1103/PhysRevResearch.2.043322}{10.1103/PhysRevResearch.2.043322}]}

\paperitem{B}{Efficient simulation of Gottesman-Kitaev-Preskill states with Gaussian circuits}{Cameron~Calcluth, Alessandro~Ferraro, and Giulia~Ferrini}{\href{https://doi.org/10.22331/q-2022-12-01-867}{Quantum \textbf{6}, 867 (2022)}.\\ Reference \cite{calcluth2022}. [doi:\href{https://doi.org/10.22331/q-2022-12-01-867}{10.22331/q-2022-12-01-867}]}

\paperitem{C}{Vacuum provides quantum advantage to otherwise simulatable architectures}{Cameron~Calcluth, Alessandro~Ferraro, and Giulia~Ferrini}{\href{https://doi.org/10.1103/PhysRevA.107.062414}{Physical Review A \textbf{107}, 062414 (2023)}. \\ Reference \cite{calcluth2023}. [doi:\href{https://doi.org/10.1103/PhysRevA.107.062414}{10.1103/PhysRevA.107.062414}]}

\paperitem{D}{Sufficient condition for universal quantum computation using Bosonic circuits}{Cameron~Calcluth, Nicolas~Reichel,  Alessandro~Ferraro, and Giulia~Ferrini}{\href{https://doi.org/10.1103/PRXQuantum.5.020337}{PRX Quantum \textbf{5}, 020337 (2024)}. \\ Reference \cite{calcluth2024}. [doi:\href{https://doi.org/10.1103/PRXQuantum.5.020337}{10.1103/PRXQuantum.5.020337}]}

\paperitem{E}{Classical simulation of circuits with realistic Gottesman-Kitaev-Preskill states}{Cameron~Calcluth, Oliver~Hahn, Juani~Bermejo-Vega, Alessandro~Ferraro, and Giulia~Ferrini}
{\href{https://doi.org/10.48550/arXiv.2412.13136}{arXiv:2412.13136} (2024) [Physical Review Letters (\href{https://doi.org/10.1103/xmtw-g54f}{to be published})]. \\ Reference \cite{calcluth2025}. [doi:\href{https://doi.org/10.1103/xmtw-g54f}{10.1103/xmtw-g54f}]}

\ifpdf
\vspace{1cm}
 \cleardoublepage              \phantomsection               \addcontentsline{toc}{chapter}{List of Acronyms}  \chapter*{List of Acronyms}

\begin{tabular}{ l c l }

 CCR & -- & Canonical commutation relations\\
 CV & -- & Continuous variable\\
 DV & -- & Discrete variable \\
 GKP & -- & Gottesman-Kitaev-Preskill\\
 MSD & -- & Magic state distillation\\
 PDF & -- & Probability density function\\
 QA & -- & Quantum advantage\\
 QC & -- & Quantum computation\\
 QEC & -- & Quantum error correction\\
 ROM & -- & Robustness of magic\\
 SGKP & -- & Simulatable Gottesman-Kitaev-Preskill\\
 SNF & -- & Smith normal form\\
 SSD & -- & Subsystem decomposition\\
 UQC & -- & Universal quantum computation\\
 WLN & -- & Wigner logarithmic negativity\\
 ZGW & -- & Zak-Gross Wigner\\

\end{tabular}

\cleardoublepage
\phantomsection
\addcontentsline{toc}{chapter}{Contents} 
\pdfbookmark[0]{Contents}{Contents}
\begingroup
  \pagestyle{contentspage}
  \tableofcontents
\endgroup

\mainmatter                   

\pagestyle{fancy}

\fancyhf{}                  \fancyhead[LE]{\textit{\nouppercase{\leftmark}}}
\fancyfoot[CO]{\thepage}
\fancyhead[RO]{\textit{\nouppercase{\rightmark}}}
\fancypagestyle{plain}{\fancyhf{}\fancyfoot[CO]{\thepage}\renewcommand{\headrulewidth}{0pt}  }

\AddToShipoutPictureBG{\ifodd\value{page}
    \begin{tikzpicture}[remember picture,overlay]
      \node[anchor=south]
        at ([xshift=0cm,yshift=1.4cm]current page.south)
        {\overlayfooterimage};
    \end{tikzpicture}\fi
}

\fi

    \clearpage{}\chapter{Introduction}

\label{ch:introduction}

\DropCap{Q}{uantum}{0} computers have the potential to perform some tasks exponentially faster than classical computers, a concept known as quantum advantage (QA)~\cite{zhong2020,hangleiter2023,anschuetz2023,harrow2017,farhi2016,arute2019}. In principle, this exponential increase in computational power could transform a range of fields. For example, it has been suggested that quantum computers could simulate complex quantum systems more accurately, leading to the development of new materials with enhanced properties~\cite{bauer2020,tacchino2020,altman2021}. They could also accelerate drug discovery by simulating molecular interactions and potentially identifying new treatments for diseases~\cite{li2024,santagati2024}.

While there exist a number of algorithms to justify these claims~\cite{QuantumAlgorithmZoo,montanaro2016,biamonte2017,cerezo2021,nielsen2010}, the field of quantum algorithm development is still relatively new. A strategy to develop new algorithms is to first identify examples of subroutines of algorithms that we know cannot be performed on classical computers~\cite{chitambar2019}. This allows us to partition the components of a quantum computer (e.g., states, operations, measurements) into two subsets: those that can be simulated efficiently on a classical computer and those capable of unlocking the ability to perform tasks with a QA.
Understanding which algorithms \textit{cannot} offer speedups over classical computers helps to identify which components, when added to the circuits, enable the ability to perform tasks that \textit{do} have a speedup, unlocking QA.

This compilation thesis\footnote{A compilation thesis, or \textit{kappa}, common in Scandinavia, is written as an introductory summary of the papers published during the PhD, rather than a traditional monograph~\cite{gustavii2012}.} addresses both regimes. In particular, we demonstrate that large classes of quantum circuits are efficiently simulatable with classical computers. In addition, we also show that this result can be extended with a small exponential overhead to a more complex class of circuits highly relevant to experiments. Complementing these simulatability results, we also develop a new framework to understand when quantum computers are capable of achieving a QA, based on our simulatability results.

In this introductory Chapter, we begin in Section~\ref{sec:qc} with a high-level introduction to quantum computing and outline the key concepts that are used throughout this thesis. Then, in Section~\ref{sec:intro-motivation} we motivate the research on which this thesis is based. In Section~\ref{sec:key-objectives}, we define key objectives that this thesis aims to solve. Finally, in Section~\ref{sec:outline}, we provide an outline of the remaining Chapters of this thesis.

\section{Quantum computing}
\label{sec:qc}
The majority of research in the field of quantum computing focuses on building quantum computers using qubits~\cite{burkard2023, chatterjee2021, kjaergaard2020, monroe2021}, which are the quantum equivalent of a bit. Qubits are a special case of the more general $d$-level system, the qudit. Both qubits and qudits are discrete-variable (DV) systems and encode quantum information in finite-dimensional Hilbert spaces consisting of states that can exist in superpositions of different ``levels'', i.e., basis states. Physically, such levels can be implemented in a variety of systems such as superconducting circuits~\cite{kjaergaard2020}, nuclear spins~\cite{abobeih2022}, semiconductor spins~\cite{burkard2023,chatterjee2021}, trapped ions~\cite{monroe2021,bruzewicz2019}, and neutral atoms~\cite{negretti2011, saffman2010, saffman2016,bluvstein2024}. The study of simulatability and resources in DV quantum computers is relatively advanced~\cite{chitambar2019,nielsen2010}. For example, the Gottesman-Knill theorem proves that a particularly important class of circuits, known as the Clifford circuits,  are efficiently simulatable~\cite{gottesman1997,gottesman1999a,gottesman1999b,hostens2005,nielsen1997}. These circuits consist of stabiliser states acted on by Clifford operations and measured with Pauli measurements. Based on this simulatability theorem, numerous results have been studied to investigate which states can unlock universal quantum computation (QC)~\cite{bravyi2012,bu2023,campbell2010,campbell2012,campbell2017,hahn2022,heinrich2019,howard2017,howard2014,litinski2019,raussendorf2020,reichardt2005,bravyi2005,seddon2021,yoganathan2019,zurel2020,zurel2024,bermejo-vega2017}. 

Despite the significant progress made across these diverse hardware platforms, scaling qubit-based systems to the millions of physical qubits required for fault-tolerant universal QC (UQC) presents substantial engineering challenges~\cite{devitt2013,nielsen2010}.
The delicate nature of quantum states makes qubits highly susceptible to errors caused by interactions with their environment or imperfect control signals.
These noise sources limit how long quantum information can be reliably stored and the accuracy of quantum gate operations.
To overcome these limitations and enable fault-tolerant computation, quantum error correction (QEC) is necessary.

An alternative approach to DVQC is to use continuous-variable (CV) devices~\cite{braunstein2005,ferraro2005,kok2010,adesso2014,serafini2017}. CV quantum computers encode quantum information in continuous degrees of freedom, such as the amplitude and phase of a Bosonic field.
CVQC offers distinct potential advantages, such as the ability to encode a large amount (theoretically infinite) of information in a single mode, and a natural encoding of continuous degrees of freedom relevant for physical simulations and optimisation problems. CVQC can be implemented on various physical platforms, each with its own strengths and weaknesses.
Prominent examples include optical systems, which leverage the modes of light and readily available components like beam splitters and photodetectors~\cite{pfister2019, takeda2019}, microwave circuits utilising superconducting elements like Josephson junctions to engineer complex interactions and states~\cite{gu2017, hillmann2020, meaney2014, gao2019, romanenko2020}, trapped ions, where vibrational modes can serve as CVs~\cite{raynal2010}, and optomechanical systems coupling light and mechanical motion~\cite{aspelmeyer2014, cabot2017}.
The choice of platform often depends on the specific application and the current state of experimental technology.

As with DVQC, CVQC also faces significant challenges due to noise. However, compared to qubit-based systems, CVQCs offer an interesting and potentially advantageous opportunity: the ability to perform error correction within a single mode~\cite{gottesman2001}.

\subsection{Error correction and Bosonic codes}

In DV quantum computing, QEC protocols encode logical quantum information across multiple physical qubits so that errors affecting individual physical qubits can be detected and corrected without disturbing the protected logical state. This method of encoding the states in multiple physical qubits, i.e., on devices in different positions in space, ensures that errors caused by environmental errors on one qubit can be corrected due to the information also existing across the other systems. 

In CVQC, it is possible to encode qubit information in a single mode and perform error correction using only that mode. These error correction codes are known as Bosonic codes~\cite{albert2025}.

The Gottesman-Kitaev-Preskill (GKP) code is a leading example of a Bosonic code. This code protects the logical information of a qubit (or, more generally, qudit) by defining the logical qubit states as two displaced Dirac combs in the position-basis wavefunction. Importantly, this structure has the convenient property that the Fourier transform of the wavefunction is another Dirac comb. In practice, this means that the wavefunction in both position and momentum space has a comb-like structure and that small errors in either position or momentum correspond to small shifts of a comb. It is possible to design measurement routines to detect the modular position or momentum, which means that we can detect these small displacements without revealing the logical content of the state. This allows us to perform error correction on the state without revealing and collapsing the wavefunction of the encoded state.

Unfortunately, perfectly encoded GKP states are unphysical because a Dirac comb is non-normalisable. However, approximate versions of GKP states exist, whereby each delta function in the Dirac comb is replaced by a thin Gaussian peak, and the entire wavefunction has a large Gaussian envelope. 
These physical GKP states have been realised experimentally in multiple platforms, including superconducting microwave cavities~\cite{campagne-ibarcq2020,kudra2022}, trapped-ion mechanical oscillators~\cite{fluhmann2019}, and quantum optics~\cite{konno2024}.

\subsection{Simulation}

Classical simulatability refers to a classical computer's ability to reproduce a quantum process's output efficiently.
A quantum process is considered efficiently classically simulatable if the computational resources, such as time and memory (often referred to as ``space''), required by a classical algorithm scale polynomially with the number of qudits or modes $n$ of the quantum system. 
This definition encompasses three flavours: strong simulation, weak simulation and probability estimation. Strong simulation involves calculating the probability of each possible outcome. Weak simulation requires generating samples from the output distribution according to their probabilities. I.e., the classical device produces outputs like a quantum computer would. Probability estimation is a different type of simulation and refers to the task of computing an additive approximation to the actual probability. Strong simulation is the most difficult to achieve~\cite{terhal2002,pashayan2020,chabaud2020b}.

While strong simulation is generally computationally harder than weak simulation, both provide benchmarks against which the capabilities of quantum computers can be compared \cite{montanaro2016, watrous2009}.
In either case, if a classical computer could, even weakly, simulate a quantum algorithm using the same time and space, it would render the quantum algorithm useless.

Understanding the boundary between classically efficiently simulatable and non-classically simulatable quantum algorithms helps delineate the limits of classical computing and identify potential areas where quantum computers can offer a significant speedup or QA~\cite{preskill2018, harrow2017, yung2019}.
Furthermore, studying simulatability provides tools for verifying the correct operation of quantum devices, even in regimes where full classical simulation is intractable \cite{hangleiter2017, spagnolo2014, dimeglio2024}.

Similar to the Gottesman-Knill theorem for DVQC, efficiently simulatable circuits exist in CVQC.
For example, Gaussian circuits are classically simulatable~\cite{braunstein2005,weedbrook2012,adesso2014,ferraro2005,cerf2007}, establishing non-Gaussianity as a necessary resource for QA.
These cases illustrate scenarios where the computational boundary remains firmly on the classical side.

Conversely, certain quantum algorithms and computational tasks are considered intractable for even the most powerful classical computers.
Famous examples include Shor's algorithm for factoring large numbers \cite{shor1999}, which provides an exponential speedup over known classical algorithms, and certain sampling problems such as Boson sampling \cite{aaronson2013, broome2013, tillmann2013, spagnolo2014, wang2019, zhong2020} and Random Circuit Sampling~\cite{bouland2019}.
These problems are ideal candidates for demonstrating a QA.

While Gaussian operations on Gaussian states are simulatable, the introduction of non-Gaussian states, such as GKP states~\cite{gottesman2001,baragiola2019}, or non-Gaussian gates, like the cubic phase gate~\cite{lloyd1999,budinger2024}, is generally required for universal CVQC.
Quantum algorithms for problems believed to be classically intractable, such as those requiring UQC, generally require computational resources that scale exponentially with the system size on a classical computer.
While any QC could technically --- with infinite time --- be simulated classically (e.g., by tracking the state vector), this simulation typically incurs an exponential overhead in resources relative to the size of the quantum system, rendering it practically impossible for large systems.

\section{Motivation}
\label{sec:intro-motivation}
In this thesis, we investigate the boundary between efficiently simulatable circuits and those capable of UQC for CVQC.

We investigate the simulatability of a large class of circuits containing GKP states. Previous theorems based on the Gaussian character of the circuit elements are incapable of assessing their simulatability because GKP states are highly non-Gaussian~\cite{garcia-alvarez2020}. Simulating circuits with GKP states was previously very challenging because of the strong non-classical features of these states. However, using various new techniques, we find that these circuits are, in fact, simulatable under certain conditions. These results can also be extended to simulating realistic GKP states with an exponential (but surprisingly small) overhead. This advance opens the door to practical benchmarks of Bosonic quantum processors, offering a tool for near-term experiments aiming to achieve fault-tolerance. In addition, it enables the focused design of new quantum algorithms that can achieve QA.

Through our simulatability theorems, we can also develop a unified framework for identifying which CV resources promote these otherwise simulatable circuits to UQC. As a striking example, we show that adding only the vacuum state --- ordinarily a ``free'' Gaussian resource --- renders the GKP model universal. Building on this insight, we introduce a general class of maps from CV to DV, which are implementable within the simulatable class of circuits and use measures of resourcefulness to assess the resulting logical states. This framework yields the first sufficient criterion for universality in Bosonic circuits and enables quantitative comparison of candidate resource states. Note that this contrasts with the case consisting of otherwise all-Gaussian components, whereby a sufficient condition to promote Gaussian circuits to universality remains to be found.

\section{Objectives}
\label{sec:key-objectives}
The overarching goal of this thesis is to investigate the boundary between classically simulatable and truly universal CVQCs, focusing on GKP encodings. We can understand this overarching goal through three key objectives:

\begin{itemize}
\setlength\itemsep{0.1em}
\item \textbf{Objective 1} (\textit{Classical simulation of ideal GKP states}): Building on a combination of techniques devised initially for qubits and new techniques devised from analytic number theory, we will show that circuits involving ideal GKP states can be efficiently simulated on a classical computer under certain conditions.
\item \textbf{Objective 2} (\textit{Classical simulation of realistic GKP states}) We will demonstrate that it is possible to simulate realistic GKP states under certain conditions. Here, we expect that the algorithm will scale exponentially with respect to the number of modes. However, we aim to make this exponential overhead as small as possible, to the point that the simulation of practical circuits is tractable.
\item \textbf{Objective 3} (\textit{Framework to assess the resourcefulness of CV states}): We will develop a new resource theory for CV systems, introduce a method to measure the resourcefulness of states based on state decompositions, and establish a criterion for achieving universality.
\end{itemize}

\section{Outline of the thesis}
\label{sec:outline}
This thesis is organised into four remaining Chapters.

First, Chapter~\ref{ch:background} reviews the necessary background of this thesis, including DV stabiliser formalisms, CV Gaussian states and operations, Bosonic codes, and phase-space representations necessary for our simulation techniques. This Chapter does not represent an exhaustive survey of CVQC; for a more comprehensive treatment, the reader is referred to textbooks~\cite{cerf2007,kok2010,serafini2017} and review articles~\cite{braunstein2005,weedbrook2012,pfister2019,brady2024}.

Following this brief review, Chapter~\ref{ch:simulation-of-ideal-gottesman-kitaev-preskill-states} addresses Objective~1 by presenting efficient classical simulation algorithms for circuits incorporating ideal (infinitely squeezed) GKP states, rational symplectic transformations, continuous displacements, and homodyne detection. We summarise a variety of techniques which yield different classes of simulatable circuits, including deriving closed-form probability distributions using number-theoretic methods and extending the stabiliser formalism to accommodate these circuits. We demonstrate the simulatability of these circuits even though they are highly non-classical.

In Chapter~\ref{ch:simulation-of-realistic-gottesman-kitaev-preskill-states}, corresponding to Objective~2, we generalise our results to realistic (finitely squeezed) GKP states. We develop a novel weak-simulation framework based on the Zak-Gross Wigner (ZGW) function, achieving a runtime with a small, but exponential, scaling in the squeezing parameter, thus making simulations in the practical regime tractable.

Chapter~\ref{ch:a-framework-for-universality-in-continuous-variable-quantum-computation} introduces a resource-theoretic framework addressing Objective~3, wherein continuous-to-discrete mappings are implemented within the simulatable class to evaluate the resourcefulness of states to promote a simulatable class of circuits to universality. This yields the first sufficient criterion for universal CVQC and unlocks the ability to compare possible resource states quantitatively.

Finally, in the Conclusion, we summarise our findings, discuss the implications for the computational power of CV devices, and outline promising directions for future research.

\clearpage{}
    \clearpage{}\chapter{Background}

\label{ch:background}

\DropCap{I}{n}{0.2} this Chapter, we introduce the technical tools and concepts that underly this thesis. These include the formal definition of qubits and qudits, CVQC and Bosonic encodings, quasiprobability distributions, and mappings to associate qubit-like quantum information to CV states.

\section{Discrete-variable quantum computing}

\label{sec:discrete-variable-quantum-computing}

DVQC operates in finite-dimensional Hilbert spaces, with the $d$-dimensional qudit as its fundamental unit.
The most common form of qudit discussed in the literature is the $d=2$ qudit, more commonly known as the qubit. We will begin by defining qubits before moving on to the higher-dimensional $d>2$ qudits.

\subsection{Qubit-based quantum computing}

\label{sub:qubit-based-quantum-computing}

A qubit, the quantum analogue of a classical bit, can exist not only in the computational basis states $|0\rangle$ and $|1\rangle$ but also in a superposition state $\alpha_0|0\rangle + \alpha_1|1\rangle$, as well as its probabilistic mixtures~\cite{nielsen2010}.
Here, $\alpha_0$ and $\alpha_1$ are complex numbers satisfying the normalisation condition $|\alpha_0|^2 + |\alpha_1|^2 = 1$.
This ability to exist in superposition and exhibit entanglement allows DVQC systems to explore multiple computational paths simultaneously. Importantly, however, when measuring a qubit, it collapses to only one of two states. Upon measurement in the computational basis $\{|0\rangle, |1\rangle\}$, the qubit state collapses to $|j\rangle$ with probability $|\alpha_j|^2$.

Operations on qubits transform states to other states.
A particularly important class of operations is the Pauli group.
The single qubit Pauli group $\pauligroup$ consists of the $2\times 2$ identity matrix $\mathbbm{1}$, the Pauli-X operator $\hat X=\ket 1 \bra 0+\ket 0 \bra 1$, the Pauli-Z operator $\hat Z=\ket 0 \bra 0 - \ket 1 \bra 1$ and the Pauli-Y operator $\hat Y=i\hat X\hat Z$, along with all multiplicative factors $\{\pm 1, \pm i\}$.
The $n$-qubit Pauli group consists of all unique length-$n$ tensor products of elements selected from these groups~\cite{sakurai2017,nielsen2010}.

The full set of qubit states, including probabilistic mixtures, can be represented as a density matrix $\hat \rho$, which is a $2^n\times 2^n$ Hermitian, positive semi-definite matrix with $\Tr(\hat \rho)=1$~\cite{sakurai2017}. The set of single qubit states can equivalently be represented by a real 3-vector, known as the Bloch vector $\mathbf r^{(\hat \rho)}$. This vector is defined such that $|\mathbf r^{(\hat \rho)}| \leq 1$, where $|\cdot|$ is the standard Euclidean norm~\cite{nielsen2010}. The Bloch vector is related to the density matrix as
\begin{align}
    \hat \rho = \frac 1 2 \left(\mathbbm 1 + r_1^{(\hat \rho)}\hat X+ r_2^{(\hat \rho)}\hat Y+ r_3^{(\hat \rho)}\hat Z\right).
\end{align}

\subsection{Quantum computing with higher-dimensional systems}

\label{sub:quantum-computing-with-higher-dimensional-systems}

Qudit technology with $d>2$ is emerging as an alternative to the conventional qubit implementations of QC and information processing. They offer several potential advantages for quantum computing, including reducing circuit complexity for certain tasks, simplifying experimental implementations, and enhancing the efficiency of various quantum algorithms~\cite{wang2020}.
Qudit-based quantum computing systems can be implemented on various physical platforms such as photonic systems~\cite{gao2020,lu2020}, ion traps~\cite{klimov2003}, nuclear magnetic resonance~\cite{dogra2014,gedik2015}, and molecular magnets~\cite{leuenberger2001}.

Formally, qudits can exist in a superposition of $d$ states, i.e., $\ket{\psi}=\sum_{j=0}^{d-1} \alpha_j\ket j$ where, like with qubits, the state is normalised to $1$, i.e., $\sum_{j=0}^{d-1} |\alpha_j|^2=1$.
Compared to the standard two-level qubit, a qudit offers a significantly larger state space, enabling the encoding and manipulation of more information within a single physical system.

As for qubits, it is possible to define a group analogous to the Pauli group for arbitrary dimensions.
The single qudit $d$-dimensional Pauli group consists of combinations of the operators $\hat X_d=\sum_{j=0}^{d-1} \ket{j+1}\bra j$ and $\hat Z_d=\sum_{j=0}^{d-1} \omega_d \ket j \bra j$, where $\omega_k=e^{2 \pi i /k}$ is the $k$-th root of unity and addition inside the ket is modulo $d$. Unless explicitly stated otherwise, we use this notation for modular addition throughout the rest of this thesis.
Formally, the full $n$-qudit $d$-dimensional Pauli group is defined as~\cite{hostens2005,garcia-alvarez2020}
\begin{align}
\pauligroup[n][d]=\left\{\omega^{u}_{D}\hat X_d^v Z_d^w : u \in \mathbb Z_D \text{ and } v,w \in \mathbb Z_d\right\},
\end{align}
where
\begin{align}
D=\begin{cases}
d \quad & \text{for odd } d,\\
2d \quad & \text{for even } d.
\end{cases}
\end{align}
We use the $d$-dimensional Pauli operators to form the useful Heisenberg-Weyl group~\cite{gibbons2004} in terms of the elements
\begin{align}
\label{eq:hwop}
\hwop = \omega_D^{\mathbf{u_X}\cdot \mathbf{u_Z}/2} \hat X_d^{\mathbf{u_X}}\hat Z_d^{\mathbf{u_Z}},
\end{align}
where we use the notation $\mathbf{u_X}$ to refer to the first $n$ elements of $\mathbf u$ and $\mathbf{u_Z}$ to refer to the last $n$ elements, and also $\hat X_d^{\mathbf v}=\hat X_d^{v_1} \otimes \dots \hat X_d^{v_n}$.

\subsection{Introduction to the stabiliser formalism}

\label{sub:introduction-to-the-stabiliser-formalism}

The stabiliser formalism is a powerful tool that is useful for understanding certain classes of quantum states and circuits~\cite{gottesman1997,gottesman1999a,gottesman1999b,aaronson2004,nielsen2010,gheorghiu2014, hostens2005, appleby2008,farinholt2014}.
It provides a method to characterise quantum states by a small set of operators, which grows linearly with the number of qudits, rather than the state vectors, which grow exponentially. The formalism was initially developed for qubits~\cite{gottesman1997,gottesman1999a}, and was shortly after extended to qudits of odd dimension~\cite{gottesman1999b}, and finally generalised to arbitrary dimensions~\cite{hostens2005}.

A stabiliser $\stabgroup$ for an $n$-qudit state $|\psi\rangle$ or subspace $\stabspace$ is an abeliean subgroup of $\pauligroup[n][d]$ such that $\stab |\psi\rangle = |\psi\rangle$ for all $\stab \in \stabgroup$.
In other words, it is the $+1$ eigenspace of all operators in $\stabgroup$~\cite{gottesman1997,gottesman1999a}. For $\stabgroup$ to be non-trivial, it cannot contain the element $-\hat{\mathbbm{1}}$.

The group $\stabgroup$ is typically described by a set of independent generators $\{\stab_1, \dots, \stab_l\}$ such that $\stabgroup = \langle \stab_1, \dots, \stab_l \rangle$.
A stabiliser group $\stabgroup$ generated by independent commuting operators has a stabilised subspace $\stabspace$ of dimension $d^{n-l}$.
Hence, if we have $l=n$ stabilisers, the stabiliser group uniquely stabilises a single state~\cite{gottesman1997,gottesman1999a,gottesman1999b}.
The set of states for which a valid, unique stabiliser group exists is called the stabiliser states.

The stabiliser formalism is particularly useful for describing the evolution of stabiliser states under a class of unitary operations, known as Clifford operations. 
The $d$-dimensional Clifford group is the subset of the unitaries for which Pauli operators are mapped to different Pauli operators under conjugation.
Formally, it is defined as
\begin{align}
\clifford[n][d]=\{\unitary \stab \unitary^\dagger \in \pauligroup[n][d]: \stab \in \pauligroup[n][d]\}.
\end{align}

For qubits, the Clifford group $\clifford[n][2]$ is generated by the Hadamard $\hat H$, phase\footnote{The phase gate is usually referred to as $\hat S$ in the qubit literature. We avoid this because we later refer to squeezing using this notation.} $\phase$, and controlled-NOT gates.
For qudits with $d>2$, the Clifford group is generated by the $d$-dimensional equivalents: the Fourier transform $\fourier[d]$, the phase gate $\phase[d]$, and $\sumgate[d]$ gate (which all reduce to the qubit equivalents for $d=2$).
These are defined as
\begin{align}
\phase[d]=\sum_{j=0}^{d-1} \omega^{j(j-\doffset)}_{2d} \ket k \bra j,
\end{align}
where 
\begin{align}
\doffset = \begin{cases}
1 \quad & \text{if } d \text{ odd,}\\
0 \quad & \text{if } d \text{ even,}
\end{cases}
\end{align}
\begin{align}
\fourier[d]=\frac{1}{\sqrt d}\sum_{j,k=0}^{d-1} \omega^{jk}_d \ket k \bra j,
\end{align}
and
\begin{align}
\sumgate[d]= \sum_{j=0}^{d-1} \ket{j+k}\bra{j}\otimes\ket{k}\bra{k}.
\end{align}

The stabiliser formalism also provides a framework for describing the outcomes of measurements of the Heisenberg-Weyl observables $\hwop$ (and the rest of the Pauli group by multiplying by an arbitrary global phase) and the resulting state transformations.
If $\hwop$ commutes with all stabiliser generators $\stabgroup$, the state is already an eigenstate of $\hwop$, yielding a deterministic $\pm 1$ outcome without altering the state or its stabiliser.
If $\hwop$ anti-commutes with one or more generators, the outcome is probabilistic, and outcomes $+1$ and $-1$ are equally likely.
The stabiliser group is updated based on the measurement outcome to reflect the post-measurement state~\cite{nielsen2010}.

An important consequence of the stabiliser formalism is the Gottesman-Knill theorem~\cite{gottesman1997,gottesman1999a}, which proves that circuits consisting of stabiliser states, Clifford operations and Pauli measurements are efficiently simulatable --- i.e., simulatable in a polynomial time with respect to the number of qudits --- on a classical computer. This is because the description of the states can be tracked with a number of variables that scales quadratically with the number of qudits, the operations act as $2n\times 2n$ matrix multiplications and measurements involve checking the commutation relations of a linear number of stabilisers with the measurement operator~\cite{hostens2005} (which in practice involves a $2n$-vector dot product).
Follow-up studies extended this idea to much richer gate sets known as normaliser circuits that include quantum Fourier transforms, group automorphisms, and quadratic phase gates --- first for finite Abelian registers and later for infinite-dimensional systems~\cite{bermejo-vega2014,bermejo-vega2016a,bermejo-vega2017}.

\section{Continuous-variable quantum computing}

\label{sec:continuous-variable-quantum-computing}

Building upon the concepts introduced for DV quantum computing in Section~\ref{sec:discrete-variable-quantum-computing}, we now focus on CV systems, which represent an alternative paradigm for QC.

While the DV paradigm encodes quantum information in finite degrees of freedom, such as the two basis states of a qubit, CV systems utilise continuous degrees of freedom.
These variables are represented by operators acting on infinite-dimensional Hilbert spaces, characteristic of CV systems.

A prominent example of a CV system is an electromagnetic field mode. Even though the following formalism can be applied to any CV system, we use the nomenclature of quantum optics throughout this thesis. The relevant observables with a continuous spectrum are the amplitude and phase quadrature of the field, denoted $\hat q$ and $\hat p$ respectively, often referred to as the position and momentum operators. These satisfy the eigenket equations
\begin{align}
    \label{eq:eigenkets}
    \hat q \ket{\hat q=x}=x\ket{\hat q=x},\quad
    \text{and}\quad \hat p \ket{\hat p=p}=p\ket{\hat p=p}.
\end{align}

\subsection{Quantum mechanical framework}

\label{sub:quantum-mechanical-framework}

The fundamental mathematical description of CV systems is rooted in the canonical commutation relations (CCR).
For a system with $n$ modes, represented by pairs of canonical operators $\hat{q}_j$ and $\hat{p}_j$ for $j \in \{1,...,n\}$, the CCR are given by\footnote{We use the convention that $\hbar=1$ throughout this thesis.} $[\hat{q}_j, \hat{p}_k] = i\delta_{jk}$ \cite{adesso2014, wilde2016a, serafini2017}.
Their Hilbert space is the $n$-fold tensor product of the Hilbert space of a single mode, defined as $L^2(\mathbb R)$~\cite{serafini2017}. These operators can also be used to define the photon creation and annihilation operators,
\begin{align}
    \hat a^\dagger &= \frac{1}{\sqrt{2}}(\hat q-i\hat p),\\
    \hat a &= \frac{1}{\sqrt{2}}(\hat q+i\hat p),
\end{align}
respectively.

The CCR can be expressed compactly using the real, anti-symmetric symplectic form $\sympform$, which is defined as
\begin{align}
\label{eq:sympform}
\sympform = \begin{pmatrix}0&\mathbbm 1\\-\mathbbm 1&0\end{pmatrix}.
\end{align}
To do so, we write the set of canonical operators in vector form\footnote{Note that this vector can equivalently be written in the order $\hat{\mathbf r}=(\hat q_1,\hat p_1,\dots,\hat q_n,\hat p_n)$, in which case the symplectic form $\Omega$ is instead expressed as a permutation of the form given in Eq.~(\ref{eq:sympform}).} as $\hat{\mathbf r}=(\hat{\mathbf q}^T,\hat{\mathbf p}^T)^T=(\hat q_1,\dots,\hat q_n,\hat p_1,\dots,\hat p_n)^T$. This allows us to express the CCR as\footnote{This is sometimes expressed~\cite{serafini2017} using the notation $[\hat{\mathbf r},\hat{\mathbf r}^T]=i \Omega$, whereby $[\mathbf a,\mathbf b]$ is defined as the outer product. However, throughout this thesis, we use this same notation to instead refer to the symplectic dot product defined in Eq.~(\ref{eq:comm-meaning}).}~\cite{puri2001,ferraro2005,serafini2017}
\begin{align}
    [\hat r_j,\hat r_k]=i\Omega_{jk}.
\end{align}

This formalism provides a canonical framework for representing the system's operators and transformations.
Formally, the real symplectic matrices $\symp[R]$ are defined as the group of $2n\times 2n$ real matrices $S$ such that $S^T\Omega S=\Omega$~\cite{arvind1995,gosson2006}. It is often convenient to express a symplectic matrix in block form,
\begin{align}
    \label{eq:symp-block}
    S=\begin{pmatrix}
        A&B\\C&D 
    \end{pmatrix}.
\end{align}

\subsection{Gaussian states and operations}

\label{sub:gaussian-states}

Gaussian states are a useful fundamental class of CV quantum states, often compared to the stabiliser states in DV quantum computing. Gaussian states are completely characterised by their first and second moments of the quadrature operators $\hat{\mathbf r}$~\cite{braunstein2005,ferraro2005,weedbrook2012,watrous2009,serafini2017}.

An alternative and equivalent useful definition views Gaussian states as the set of states arising from quadratic Hamiltonians of the form
\begin{align}
\label{eq:gauss-ham}
\hat{H} = \frac{1}{2} \hat{\mathbf r}^T H \hat{\mathbf r} + \hat{\mathbf r}^T \mathbf r,
\end{align}
where the Hamiltonian matrix $H$ is a real symmetric matrix and $\mathbf r$ is a real $2n$-vector.
Note that such Hamiltonians are, by definition, at most quadratic in the canonical operators.

As we will see in Subsection~\ref{sub:classical-simulation-gaussian}, the Gaussian nature of their phase space representation makes them mathematically tractable to simulate in polynomial time, with respect to the number of modes of the system~\cite{braunstein2005,ferraro2005,watrous2009,safranek2015,wolf2006}.
Examples of Gaussian states include the vacuum, coherent, squeezed, and thermal states.
States such as Fock states (for $n>0$) and cat states are non-Gaussian.
This natural boundary between Gaussian and non-Gaussian states will be explored in detail throughout this thesis.

Gaussian operations are defined as the set of unitary operations with Hamiltonians of the form given in Eq.~(\ref{eq:gauss-ham}), i.e., $e^{i\hat H}$.
Gaussian operations can always be separated into a symplectic transformation, represented in terms of a symplectic matrix $\sympmat=e^{\Omega H}$, and a displacement $\mathbf r$.

The displacement operator is the most basic type of operator in CV quantum computing and is defined as
\begin{align}
\label{eq:disp}
\disp[\mathbf r]=e^{i [\mathbf r,\hat{\mathbf r}]} =e^{i {\mathbf {r_X}}\cdot \mathbf{r_Z}} e^{i\mathbf {r_X} \cdot \hat{\mathbf p}}e^{-i \mathbf{r_Z} \cdot \hat{\mathbf q}} 
\end{align}
for a phase space vector $\mathbf r \in \mathbb R^{2n}$, where we have introduced the notation~\cite{delfosse2017}
\begin{align}
    \label{eq:comm-meaning}
    [\mathbf a,\mathbf b]=\mathbf a^T \Omega \mathbf b,
\end{align}
which is used throughout this thesis.
Given the frequency with which we apply displacements only in position or momentum, we introduce the operators $\hat X(\mathbf s)=e^{-i {\mathbf s}\cdot \hat{\mathbf{p}}}$ and $\hat Z(\mathbf s)=e^{i {\mathbf s}\cdot \hat{\mathbf{q}}}$, where here $\mathbf s \in \mathbb R^n$.

The reason why any unitary Gaussian operation can be separated into a symplectic transformation and displacement can be understood in terms of the useful commutation relation
\begin{align}
\label{eq:comm-disp-symp}
\unitary_{\sympmat} \disp[\mathbf r] = \disp[\sympmat^{-1}\mathbf r] \unitary_{\sympmat},
\end{align}
which implies that the Hamiltonian's displacement terms in Eq.~(\ref{eq:gauss-ham}) can be ``pulled out'' to one side.

We also identify the structure of the set of unique Gaussian operations from this commutation relation. They can be represented concisely with the notation $\text{HW}(n)[\operatorname{Sp}(2n,\mathbb R)]$, which is the semi-direct product of the Heisenberg-Weyl group (representing displacements) with the symplectic group~\cite{bartlett2002}.

\section{Universal quantum computation}
\label{sec:universality}

In DVQC, universality refers to the ability to prepare any state within the Hilbert space, act on those states with any unitary, and measure them with any Kraus operator~\cite{lloyd1995,divincenzo1995}. 

Clifford circuits are not universal and therefore require additional operations to be added to the set to become universal. This can be achieved by the addition of access to non-stabiliser gates. For example, for qubits, the $\hat T=\phase^{1/2}$ gate unlocks universality for DVQC. This gate implements a relative $\pi/4$ phase on the $\ket 1$ computational basis state only.

Universality with CVQC can be interpreted in two different ways. First, it is sometimes used to refer to the ability to approximately prepare any possible state, operation and measurement with a finite number of operations up to arbitrary error~\cite{lloyd1999}. Alternatively, it may mean the ability to encode universal DVQC. We usually refer to the second of these two types of universality in this thesis. In either case, non-Gaussianity is a required resource to unlock universality.

Generally, the non-Gaussian elements needed to achieve QA can be accomplished by either starting with or generating non-Gaussian states, performing non-Gaussian operations, or utilising non-Gaussian measurements.
Specific non-Gaussian resource states are known to promote Gaussian circuits to universality.
A key example of a non-Gaussian resource state is the cubic phase state, defined as the action of 
the cubic phase gate $e^{i \gamma\hat q^3 }$, with real parameter $\gamma$, acting on a squeezed vacuum state.
This state is particularly powerful as it enables the implementation of any non-Gaussian gate~\cite{lloyd1999,arzani2025} operation.

For example, it can be used to realise the logical $\hat T$ gate within the GKP encoding scheme~\cite{gottesman2001} or, more generally, it can be used to reproduce the CV cubic phase gate, which is sufficient to promote Gaussian operations to universality~\cite{braunstein2005, lloyd1999,arzani2025}.
Recent experimental progress has demonstrated the generation of cubic phase states in various platforms~\cite{kudra2022, sakaguchi2023,houhou2022}.
\section{Gottesman-Kitaev-Preskill codes}

\label{sec:encoding-schemes-for-continuous-variable-quantum-computing}

In this Section, we review the GKP code, which is a type of Bosonic code~\cite{chuang1997,gottesman2001,albert2018,grimsmo2020,brady2024,albert2025}. The GKP code is a method to encode discrete quantum information, such as that represented by qubits and qudits, within CV systems.
In this respect, this code provides a pathway towards achieving fault-tolerant QC using CV hardware \cite{gottesman2001}.

\subsection{Gottesman-Kitaev-Preskill states}

\label{sub:gottesman-kitaev-preskill-gkp-codes}

GKP states define a protected code space within the infinite-dimensional Hilbert space of a CV mode, characterised by a set of commuting stabiliser operators \cite{gottesman2001,royer2020}.
For a single mode encoding a qudit, these stabilisers are specific position and momentum displacement operators: $\lX[d]^d=\xdisp[d\peakspace[d]]$ and $\lZ[d]^d=\zdisp[d\peakspace[d]]$, where $\peakspace[d]=\sqrt{2 \pi/d}$.
The GKP code space consists of states $\ket{\psi_{\text{GKP}}}$ such that $\lX[d]^d\ket{\psi_{\text{GKP}}} = \lZ[d]^d\ket{\psi_{\text{GKP}}} =\ket{\psi_{\text{GKP}}}$.
Logical Pauli operators, corresponding to displacement operators $\lX[d]=\xdisp[\peakspace[d]]$ and $\lZ=\zdisp[\peakspace[d]]$,  commute with the stabilisers but act non-trivially on the encoded state \cite{gottesman2001}.

It is possible to show that, as a consequence of the codespace being stabilised by displacement operators with spacing $d\peakspace[d]$, the state vector representation of an ideal GKP codeword is given by
\begin{align}
\gkpket[j][d]=\sum_m \ket{\hat q=(j+dm)\peakspace[d]},
\end{align}
where $\peakspace[d]=\sqrt{2 \pi/d}$, and $\ket{\hat q=\cdot}$ are the eigenkets of the position operator as given in Eq.~(\ref{eq:eigenkets}). Equivalently, the wavefunction of the GKP codeword\footnote{Note that $\gkpket[d]$ is technically not a state because it is not normalisable (which is also why it is non-physical). However, for simplicity we often refer to it as one, including in the appended papers.} is given by
\begin{align}
    \label{eq:gkp-wf}
    \psi_{\text{GKP},j}(x)=\sum_m \delta\left(x-\peakspace[d](j+d m)\right).
\end{align}

In practice, ideal GKP codewords, requiring infinite energy and squeezing, are not physically realisable.
Instead, experiments focus on generating approximate GKP states with finite squeezing, characterised by finite energy and Gaussian-blurred features in their phase-space representation.
The quality of these realistic states, often quantified by their squeezing parameter $\Delta$, or in terms of its ``squeezing level'' in decibels --- i.e., $\Delta$ corresponds to a squeezing level of $-10\log_{10}\Delta^2$ dB --- strongly affects the performance of GKP-based QC and error correction.
The level of squeezing currently expected to be sufficient for achieving universality with GKP states is 9.75 dB~\cite{aghaeerad2025}.

Realistic GKP states can be defined in terms of their wavefunction as~\cite{gottesman2001,menicucci2014,matsuura2020}
\begin{align}
\label{eq:realistic-gkp-states}
\psi_{\text{GKP},j}^\Delta(x)= \frac{1}{\sqrt{2\pi}\Delta}e^{-x^2 \Delta^2/2} \jacobi(\frac{x}{2\ell}-\frac{j}{2},\frac{i\pi\Delta^2}{2\ell^2}),
\end{align}
where $\jacobi(z,\tau)$ is the Jacobi theta function~\cite{berndt1998} defined as
\begin{align}
\label{eq:jacobi}
\jacobi(z,\tau)=\sum_{m \in \mathbb Z} e^{\pi i m^2 \tau+2\pi i mz}.
\end{align}
Note that this definition can be extended to the multidimensional case, also known as the Sigel theta function, to
\begin{align}
\label{eq:siegel}
\jacobi(\mathbf z,\tau)=\sum_{\mathbf m \in \mathbb Z} e^{\pi i \mathbf m^T\tau \mathbf m+2\pi i \mathbf m^T \mathbf z}
\end{align}
where in the second definition $\tau$ is a matrix.
This definition will be useful later.

Fault-tolerant UQC generally requires a combination of Gaussian operations and non-Gaussian resources, such as logical magic states. Achieving universality requires adding a logical non-Clifford gate, which can be implemented directly via a cubic phase gate or, intriguingly, via magic-state injection using a Gaussian state and GKP error correction~\cite{baragiola2019}. Gaussian operations, while readily implementable, are insufficient for UQC. Non-Gaussian states, like GKP states, cat states, or non-Gaussian operations (e.g., cubic phase gates), provide the necessary magic resource.

Despite their theoretical appeal, GKP codes face significant experimental challenges in quantum optical setups.
Preparing high-fidelity, high-squeezing GKP states is technically demanding due to limitations in generating strong non-Gaussianity and precisely controlling Bosonic modes.
Furthermore, GKP states are sensitive to photon loss, a prevalent noise source in many physical platforms, which can degrade state quality and introduce logical errors.
Overcoming these limitations through improved state preparation, robust error correction strategies, and noise mitigation techniques remains an active area of research \cite{marek2024, serafini2017, noh2022}.

\subsection{Logical qubit gates}

\label{sub:logical-qubit-gates}

As explained in Section~\ref{sec:universality}, achieving universal encoded-qubit QC in CV systems requires a set of gates that, when combined with state preparation and measurement, can approximate any unitary operation on the encoded logical qubits.
As discussed, this set includes the Cliffords (generated by the Hadamard, phase gate and CNOT gate), any pure non-stabiliser state, or any state with a sufficiently high fidelity to a target magic state (see the later Section~\ref{sec:resource-theories}).

For GKP-encoded qubits, the set of Clifford operations is produced using Gaussian operations.
Specifically, logical Pauli gates $\lX$ and $\lZ$ are realised by applying displacement operators in phase space, the logical Hadamard gate is achieved by applying the CV Fourier transform $\fourier=e^{i \frac{\pi}{4}(\hat q^2+\hat p^2)}$, and logical Phase gate $\phase$ is achieved by applying a shear transformation $\phase*[s]=e^{s\hat q^2/2}$ with parameter $s=1$.
The CNOT gate is achieved by the CV SUM gate, defined as $\sumgate^{(\text{CV})}=e^{-i\hat q_1 \hat p_2}$.
To achieve the universal set, we require a $\hat T=\ket{0}\bra 0 + e^{\pi i/4} \ket 1 \bra 1$ gate, which for GKP-encoded qubits, can be implemented using the cubic phase gate $e^{i\gamma \hat q^3}$ along with the shear and displacement gates~\cite{gottesman2001}.

\subsection{Error correction}

\label{sub:error-correction}

As mentioned, a key strength of GKP codes lies in their error correction capabilities.
Small displacement errors in the encoded mode can be detected by measuring the GKP stabilisers. Specifically, by measuring the modular position $\hat{q} \mod{d\peakspace[d]}$ and modular momentum $\hat{p} \mod{d\peakspace[d]}$)~\cite{gottesman2001,glancy2006,grimsmo2021}.

\begin{figure}[h!]
    \centering
    \includegraphics[width=0.8\textwidth]{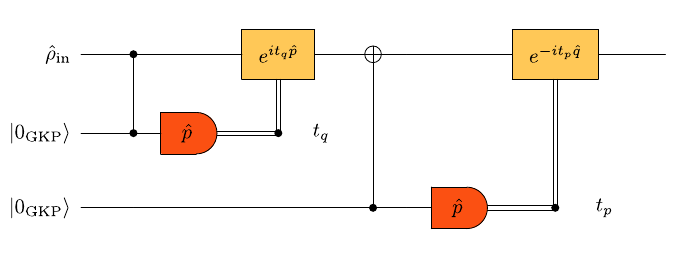}
     \caption{Circuit gadget implementing $\hat K_{\text{EC}}(\mathbf t)$. Mode $1$ is the top mode, which takes an input state and outputs a modified state. Mode $2$ and $3$ below are auxiliary modes which have a fixed input and once measured can be discarded. The measurement outcomes are denoted as $t_p$ and $t_q$.}
     \label{fig:error-correction}
    \end{figure} 
These measurements can be performed non-destructively, for instance, by coupling the encoded mode to auxiliary GKP states via Gaussian operations and measuring the auxiliary states \cite{mensen2020, campagne-ibarcq2020}.
A circuit of this error-correction routine is shown in Fig. \ref{fig:error-correction}.
The measurement outcome reveals the error syndrome, indicating the magnitude of the displacement error.
Applying an opposite displacement operation based on the syndrome measurement corrects the error and restores the state to the code space.

It is often convenient to express the action of GKP error correction in terms of a Kraus operator, which depends on measurement outcomes $t_x,t_p$ from the error-correction circuit. Specifically, the Kraus operator is given by
\begin{align}
    \label{eq:gkp-ec-kraus}
    \hat K_{\text{EC}}(\mathbf t)=\gkpproj \disp[\mathbf t],
\end{align}
where
\begin{align}
    \label{eq:gkp-proj}
    \hat \Pi_{\text{GKP}}=\sum_{j=0}^{d-1}\ket{j_{\text{GKP}(d)}}\bra{j_{\text{GKP}(d)}}
\end{align}
is the GKP projection operator for a general GKP-encoded qudit of dimension $d$.
Note that this is a non-trivial expression of the effect of the circuit, since, in practice, displacements are applied after the measurement. This Kraus operator was first defined in Ref.~\cite{baragiola2019} and derived explicitly in Paper \pD.
The state, after error correction, is proportional to
\begin{align}
\label{eq:after-ec}
\hat{\bar \rho}(\mathbf t)=\hat K_{\text{EC}}(\mathbf t) \hat \rho \hat K_{\text{EC}}^\dagger(\mathbf t).
\end{align}
The norm of the operator given in Eq.~(\ref{eq:after-ec}) will provide the probability of this outcome.

\section{Subsystem decompositions}

\label{sec:subsystem-decompositions}

Building upon the CV stabiliser formalism introduced in Section~\ref{sub:introduction-to-the-stabiliser-formalism}, the stabiliser subsystem decomposition (SSD) provides a method to convert a CV state to a DV state.

The practical SSD process involves projecting the initial state onto the GKP encoded subspace through GKP error correction and averaging over all the possible outcomes measured during error correction.
The circuit for implementing this stabiliser SSD process is the same as the error correcting gadget shown in Fig.~\ref{fig:error-correction}.
The state must then be averaged over the measurement outcomes.
The stabiliser SSD can therefore be expressed as
\begin{align}
\Tr_S(\rho)=\int_R \dd \mathbf t \hat{\bar \Pi}_{\text{GKP}}\hat D(\mathbf t) \hat \rho \hat D^\dagger(\mathbf t)\hat{\bar \Pi}^\dagger_{\text{GKP}},
\end{align}
where now we use the projector $\hat{\bar \Pi}_{\text{GKP}}=\sum_{j=0}^{d-1} \big|j\big\rangle\big\langle j_{\text{GKP}(d)}\big|$, which maps an ideal encoded GKP state to a qudit state, and $R=[0,\peakspace[d])^{\times n}$, i.e., we integrate over a ``Zak patch''~\cite{shaw2024}.
Further details on this notation for qubits are provided in Ref.~\cite{shaw2024} and Paper \pD.

\section{Phase space representations}

\label{sec:phase-space-representations}
In this Section, we review phase space representations of quantum states. We begin with the characteristic and Wigner function of CV states, and then review the analogous functions for DVQC. We then explore the formal properties of Wigner functions in general and demonstrate the covariance of Wigner functions under certain operations. Finally, we introduce the ZGW function, which is a recently developed Wigner function tailored to analysing GKP states.

\subsection{Phase space representations for continuous-variable systems}

\label{sub:characteristic-and-wigner-function-for-continuous-variable-systems}

The characteristic function, denoted by $\charfunc[\mathbf r]$, is a complete phase-space representation of a CV quantum operator $\hat B$, including quantum states, i.e., where $\hat B=\hat \rho$.
It is defined as~\cite{serafini2017}
\begin{align}
\charfunc[\mathbf r] = \text{Tr}[\disp[\mathbf r] \hat B].
\end{align}
The characteristic function contains complete information about the state and is therefore equivalent to the operator $\hat B$ itself.
Following the approach in Ref.~\cite{serafini2017}, this function is the standard ($s=0$) symmetrically ordered characteristic function.

The Wigner quasiprobability function,\footnote{In the rest of this thesis, we refer to this function simply as the ``Wigner function'' or the ``CV Wigner function'' to distinguish it from other generalised Wigner functions.} denoted $\wig[\mathbf r]$, provides an alternative phase-space perspective as a quasiprobability distribution.
It is defined as the symplectic Fourier transform of the standard characteristic function $\charfunc[\mathbf r]$, given by the integral
\begin{align}
\label{eq:wig-def}
\wig[\mathbf r] = \frac{1}{(2 \pi^2)^n} \int_{\mathbb R^{2n}} \dd \mathbf r' e^{i\mathbf r^T \Omega \mathbf r'} \charfunc[\mathbf r'].
\end{align}

A fundamental property distinguishing the Wigner function from classical phase-space distributions is that it can have negative values.
These negative regions are key indicators of non-classical features of the state.

\subsection{Phase space representations for discrete-variable systems}

\label{sub:wigner-function-for-discrete-variable-systems}

A phase space representation analogous to the continuous characteristic function and Wigner function can be defined for odd-dimensional DV systems, known as the Gross Wigner function~\cite{gross2006, gross2007,gross2008}.

Specifically, the discrete characteristic function $\charfunc*[\mathbf m]$ for an operator $\hat B$ is defined by its expansion coefficients with respect to the Heisenberg-Weyl operators as
\begin{align}
\charfunc*[\mathbf m] = d^{-n} \Tr(\hwop[\mathbf m]^\dagger \hat B),
\end{align}
for $\mathbf m\in\mathbb Z_d^{2n}$.

Furthermore, it is possible to take the discrete symplectic Fourier transform of this function to recover the DV Wigner function as
\begin{align}
    \label{eq:dv-wig}
    \wig*[\mathbf t]=\Tr(\ppodv[\mathbf t]\hat B),
\end{align}
in terms of the qudit phase point operator
\begin{align}
    \ppodv[\mathbf t]=\frac{1}{d^n} \sum_{\mathbf k\in\mathbb Z^{2n}_d}\hwop[\mathbf k] \omega^{-[\mathbf t,\mathbf k]}. 
\end{align}

Like its continuous counterpart, the discrete Wigner function is a quasiprobability distribution, but now defined on the discrete phase space $\mathbb Z_d^{2n}$.
As with the CV Wigner function, unlike a true probability distribution, it can take negative values for certain quantum states.
This negativity is an indicator of non-classicality.

\subsection{Properties of the Wigner Function}

\label{sub:properties-of-the-wigner-function}

The Wigner function $\wig[\mathbf r]$ serves as a phase-space representation of a quantum state $\hat{\rho}$ satisfying a number of properties. 
We now list the following properties, which are known as the Stratonovich-Weyl axioms~\cite{stratonovich1956,gracia-bondia1988,varilly1989,brif1999,davis2024}. Here, for simplicity, we use the notation of CV Wigner functions, which includes integration over real space. However, analogous properties hold for DV Wigner functions, whereby integration is replaced with a summation.

Although the Wigner function can be negative --- unlike some other quasiprobability distributions --- the Wigner function is always \textbf{real}.

The Wigner function also satisfies the \textbf{standardisation} condition: its integral over the entire phase space equals the trace of the density operator.
Formally, this means that, for an operator $\hat A$, we have
\begin{align}
\int_{\mathbb R^{2n}} \dd\mathbf r \,\wig[\mathbf r] = \Tr(\hat{A}) = 1.
\end{align}

This property implies an additional important feature of the Wigner function.
By integrating the Wigner function over only $\mathbf{r_X}$ or $\mathbf{r_Z}$, we obtain the marginal probability distribution for the conjugate variables.
For example, for a single-mode state $\hat\rho$, we can write $\wig[x,p]$.
Hence, integrating over $p$ gives us the probability density in position,i.e.,
\begin{align}
\int_{\mathbb R} \dd p \,W(x,p)= \bra x \hat\rho \ket x=\pdfunc[x].
\end{align}
Similarly, integrating over position $x$ gives the momentum probability density.

Additionally, the Wigner function satisfies the \textbf{traciality} property.
In particular, integrating over the product of the Wigner function of two operators $\hat A,\hat B$ is equivalent to calculating the trace of their product, i.e.,
\begin{align}
\int_{\mathbb R^{2n}} \dd \mathbf r\, \wig[\mathbf r][\hat A]\wig[\mathbf r][\hat B] = \Tr(\hat A\hat B).
\end{align}

Furthermore, the Wigner function is a \textbf{linear} map $\hat A\mapsto \wig[\cdot][\hat A]$ and two Wigner functions are identical if and only if the states they represent are identical.
Finally, the Wigner function satisfies \textbf{covariance} under Heisenberg-Weyl operations, meaning that a displaced state $\disp[\mathbf r']\hat \rho \hat D^\dagger(\mathbf r')$ will have a Wigner function that is displaced in phase space by $\mathbf r'$.

\subsection{Wigner function evolution under quantum operations}

\label{sub:wigner-function-evolution-under-quantum-operations}

The Wigner function provides a complete description of a quantum state in phase space. It is a powerful tool for visualising and analysing the evolution of both discrete- and CV systems under various quantum operations.
Understanding how quantum gates transform the Wigner function is useful for interpreting their effect on quantum states.
In this Subsection, we begin with a discussion of the evolution of the Wigner function for CV circuits, before addressing the DV case.

In CVQC, Gaussian unitary operations, generated by Hamiltonians at most quadratic in the canonical operators, induce symplectic transformations and linear translations in phase space.
These transformations preserve the Gaussian nature of quantum states and their Wigner functions, meaning a Gaussian state evolved under Gaussian operations remains Gaussian. This property is a generalisation of the covariance property of Heisenberg-Weyl operations to the broader set of all Gaussian operations.
Specifically, for a Gaussian unitary $\unitary=\unitary_{\sympmat}\disp[\bar{\mathbf r}]$, the Wigner function of the output state $\unitary \hat \rho \unitary^\dagger$ evaluated at phase-space point $\mathbf r$ is related to the input state's Wigner function by a transformation of the phase-space coordinates~\cite{ferraro2005,serafini2017,kok2010}
\begin{align}
\wig[\mathbf r][\unitary \hat \rho \unitary^\dagger]=\wig[\sympmat\mathbf r+\bar{\mathbf r}][\hat \rho].
\end{align}
Here, $\sympmat\in \symp[R]$ is a $2n \times 2n$ symplectic matrix representing the linear part of the transformation, and $\bar{\mathbf r}$ is a $2n$-dimensional real vector representing the displacement in phase space.

In DVQC, an analogous property holds for Wigner functions describing qudits of odd-prime dimension. To understand this, we first note that all Clifford operations can be interpreted as an affine transformation in precisely the same form as the CV Wigner function. Specifically, a Clifford operator\footnote{The description of Clifford operators as an integer symplectic affine transformation can be generalised to all dimensions~\cite{hostens2005,gheorghiu2014}. However, the Gross Wigner function for qudits of non-prime dimension is undefined.} is described by an integer symplectic matrix $S\in \operatorname{Sp}(2n,\mathbb Z_d)$ and an integer $2n$-vector $\bar{\mathbf r}\in\mathbb Z_d^{2n}$. Given a Clifford operation $\hat U$ described by a symplectic matrix $S$ and integer vector $\bar{\mathbf r}$, we find that \cite[Theorem 7]{gross2006}
\begin{align}
    \wig*[\mathbf r][\hat U\hat \rho\hat U^\dagger]=\wig*[S\mathbf r+\bar{\mathbf r}][ \hat\rho],
\end{align}
which can be intuitively understood as equivalent to the CV case.

In contrast to Gaussian operations in CV and Clifford operations in DV, non-Gaussian operations, or non-Clifford operations, induce nonlinear and non-affine transformations on the Wigner function in phase space, which have been shown to be a resource for nonclassicality~\cite{gross2006,albarelli2018,takagi2018}.
In CVQC, examples of non-Gaussian gates include the cubic phase gate, Kerr interaction, or photon subtraction or addition operations.
These operations can deform the Wigner function in complex ways and are necessary to generate Wigner negativity.

\subsection{Classical simulation of Gaussian circuits}
\label{sub:classical-simulation-gaussian}
As previously mentioned, circuits consisting of Gaussian states, operations and homodyne measurements are efficiently simulatable.
Specifically, it has been shown that any circuit composed exclusively of input states with non-negative Wigner functions, Gaussian-preserving operations, and measurements described by positive-Wigner POVM elements admits an efficient classical simulation via sampling from the associated phase-space distributions~\cite{mari2012,veitch2012,pashayan2015}. 

In practice, to simulate one round of computation, one first samples a phase-space point according to the positive Wigner functions of the inputs (which can therefore be treated as a probability distribution), propagate this point through each Gaussian map via the corresponding symplectic affine transformations, and finally sample a random outcome from the Wigner representation of the measurement operator. Because everything remains non-negative, the output vector of position and momentum points corresponds exactly to a simulated result of measuring in either position or momentum. Furthermore, the entire algorithm runs in polynomial time in the number of modes and gates~\cite{mari2012,veitch2012}.

\subsection{Advanced phase space techniques: Zak-Gross Wigner function}

\label{sec:zak-gross-wigner-function}
An alternative type of quasiprobability distribution was recently introduced for assessing states in the context of SGKP circuits~\cite{davis2024}. The ZGW function~\cite{debievre1996,kowalski2007,ligabo2016,davis2024} offers significant advantages for systems exhibiting translational symmetry, such as GKP states.
The ZGW function relies on the Zak transform~\cite{zak1967,zak1968,zak1972,pantaleoni2023}, a tool specifically adapted for functions with periodic properties.

The ZGW function for a single-mode state $\hat\rho$ was defined in Ref.~\cite{davis2024} as
\begin{align}
\zgw[s,t]=\frac{1}{2\pi} \sum_{m,n}(-1)^{mn}e^{i\ell (nt-ms)}\Tr(\hat D(n\ell,m\ell)\rho),
\end{align}
where $s,t\in\mathbb T_{\ell d}=[0,\ell d)$.

A defining property of the ZGW function is its periodicity in both the modular ``position'' and ``momentum'' variables, i.e.,
\begin{align}
    \zgw[s, t] = \zgw[s+d\peakspace[d], t]  = \zgw[s, t+d\peakspace[d]],
\end{align}
with unit cell $\mathbb T^{2}_{\ell d}$.

For ideal, infinitely squeezed stabiliser qudit GKP states in odd dimensions, the ZGW function is non-negative everywhere, typically represented as a comb of Dirac delta functions on the finite phase space.
In fact, for qudits of odd dimension $d$, we find that (Theorem 2 of Ref.~\cite{davis2024}) the ZGW function of a single-mode CV state $\hat \rho$ can be expressed in terms of the DV Gross Wigner function as
\begin{align}
    \zgw[s,t][\hat\rho]=\wig*[s,t][\hat{\bar \rho}(s,t)],
\end{align}
where $\hat{\bar \rho}(s,t)$ is the CV state after GKP error correction, as defined in Eq.~(\ref{eq:after-ec}).

For realistic, finitely-squeezed GKP states and magic non-stabiliser GKP states required for universal computation, the ZGW function exhibits negativity~\cite{davis2024}, indicating their non-classical nature.

\section{Resources for universal quantum computation}

\label{sec:resource-theories}

A quantum resource theory provides a rigorous framework for classifying quantum states and operations based on their usefulness for specific tasks~\cite{chitambar2019}. It does this by defining a set of free states, operations and measurements.
The states, operations or measurements that \textit{cannot} be prepared or implemented using only the free resources are considered ``resourceful''.
The central goal of any resource theory is to quantify this ``resourcefulness'' relative to the defined free set~\cite{chitambar2019}. Within this context, key resources include entanglement~\cite{horodecki2009}, non-Gaussianity~\cite{genoni2010,albarelli2018,takagi2018}, Wigner negativity~\cite{kenfack2004,mari2012}, and states exhibiting magic or stabiliser non-extent~\cite{veitch2014,howard2014}.
These resources are necessary to overcome the limitations of classical simulation, demonstrating non-classical phenomena beyond what is possible with Gaussian resources alone.

In the following two subsections, we address the main results of resource theories relevant to this thesis in both CV and DV systems.

\subsection{Resources for discrete-variable quantum computation}

As we have discussed, circuits initiated with stabiliser states, acted on by measurements drawn exclusively from the Clifford group and measured with Pauli measurements, are efficiently simulatable classically via the Gottesman-Knill theorem~\cite{gottesman2001}. To achieve universality, we must introduce magic gates such as the logical $\hat T$ gate for qubits, which, given a supply of these gates, unlocks the ability to produce any gate and thereby achieve UQC. Alternatively, we can implement a magic gate via quantum gate teleportation by consuming a magic $T$ state. This state is defined as
\begin{align}
    \ket{T}= \cos \beta \ket 0 +e^{i\pi /4}\sin\beta \ket 1, \quad \cos(2\beta)=\frac{1}{\sqrt 3}.
\end{align}
The $T$ state is not unique; many states can promote Clifford circuits to universality.

Magic measures provide a method to quantify how resourceful certain states are. A prominent example of a magic measure is the robustness of magic (ROM)~\cite{pashayan2015,heinrich2019,howard2017}, which, roughly speaking, assesses how far away the state is from the set of stabiliser states. The full multi-qubit definition is given in Paper \pD, but we provide the single-qubit definition here since it can be written in a particularly convenient form. Specifically, the ROM of a single qubit is given by
\begin{align}
    \label{eq:rom-single}
    \mathcal R^{(1)}(\hat \rho)=\abs{\Tr(\hat \rho \hat X)}+\abs{\Tr(\hat \rho \hat Y)}+\abs{\Tr(\hat \rho \hat Z)}.
\end{align}
Equivalently, it is the $1$-norm of the Bloch vector. We also note that it is directly proportional to the fidelity of the $\ket T$ state (as explicitly demonstrated in Paper \pD), which is defined as a measure of closeness
\begin{align}
    F(\hat\rho,\hat \rho')=\Tr(\sqrt{\sqrt{\rho}\rho' \sqrt{\rho}})
\end{align}
for two states $\hat \rho$ and $\hat \rho'$. Furthermore, this relation has implications for the number of states $\hat \rho$ required to achieve a set of universal operations. 

Specifically, magic state distillation (MSD) is a process that enables the ability to convert many copies of a state $\hat \rho$, which has a fidelity to a target magic state (e.g., $\ket T$) above a certain threshold, to a small number of state with a higher fidelity to the target magic state~\cite{bravyi2005,campbell2012,campbell2017,litinski2019,reichardt2005}. In practice, this means that any state $\hat \rho$ that has a fidelity to a target magic state above some threshold can promote the circuit to universality. This process can be iterated, meaning that applying this process multiple times will produce higher and higher quality states. The number of input states required to achieve a target output state is monotonically related to the initial fidelity. Furthermore, since ROM is also monotonically related to the fidelity of a qubit state to the $\ket T$ state, ROM is therefore directly related to the number of copies required to distil a target magic state and unlock UQC.

For qudits, similar theorems exist~\cite{campbell2012}. One qudit state that we use later for dimension $d=3$ is the strange state~\cite{veitch2014}, which is one of the most non-stabiliser states, defined as
\begin{align}
\label{eq:strange}
    \ket S=\frac{1}{\sqrt 2}\left(\ket 1+ \ket 2\right).
\end{align}

\subsection{Resources for continuous-variable quantum computation}
\label{sub:resources-for-cvqc}
As we have seen in Subsection~\ref{sub:classical-simulation-gaussian}, achieving the computational power beyond classical simulation requires introducing non-Gaussian resources. These may be non-Gaussian initial states or the application of non-Gaussian gates, analogous to the need for magic states or non-Clifford gates in the DV model~\cite{gottesman2001, gu2009}.

In CV systems, resource theories usually define Gaussian states, operations and measurements (possibly with feed-forward) as free~\cite{albarelli2016,albarelli2018,takagi2018}. 
A key theoretical result establishing the link between Wigner function positivity and classical simulatability is the Mari-Eisert theorem, which can be seen as a CV analogue of the Gottesman-Knill theorem \cite{mari2012}.
This theorem states that quantum circuits initialised with states, operations and measurements that are Wigner non-negative can be efficiently simulated on a classical computer.
Conversely, computations that introduce or amplify Wigner negativity beyond a certain threshold are generally considered hard to simulate classically.
Hudson's theorem \cite{hudson1974} provides a proof that a pure state has a non-negative Wigner function if and only if it is a Gaussian state, which was later extended to qudits of odd prime dimension, where non-negative discrete Wigner functions characterise stabiliser states \cite{gross2006, veitch2012}.

Significant theoretical effort has focused on identifying and quantifying the properties of states that act as resources for universal CV computation.
In the CV setting, a central role is played by Wigner negativity, a signature of non-classicality in the phase-space representation of quantum states \cite{wigner1932, hillery1984, cahill1969a}, which is known to be a necessary resource for achieving universality.

Different measures have been proposed to quantify the amount of Wigner negativity within the framework of non-Gaussianity resource theories \cite{albarelli2018, yadin2018, kwon2018, hahn2024}.
These include the negativity volume \cite{kenfack2004,albarelli2018}, and the Wigner Logarithmic Negativity (WLN) \cite{kenfack2004, albarelli2018, kwon2018}.

\section{Classical simulation beyond efficiently simulatable circuits}
\label{sec:classical-simulation-beyond}
Here we review a technique to simulate circuits that are not classically efficiently simulatable. By this, we mean that the cost of using this technique to simulate the circuits scales exponentially with the number of qudits or modes. While there are many other types of simulation techniques in the literature~\cite{bravyi2019,cirac2021,ferrie2011,orus2019,bravyi2016a,bu2019,zurel2020,zurel2024,zurel2024b}, we focus on a technique that is particularly relevant for Paper \pE. Specifically, a method to estimate the probability density function (PDF) of a DV quantum system (although it can be extended to CV) using phase space techniques~\cite{pashayan2015}. 

In general, including states or operations with Wigner negativity presents a significant challenge for classical simulation. Techniques such as those given in Ref.~\cite{mari2012,veitch2012} break down when any negativity is added. However, Ref.~\cite{pashayan2015} introduces a method to account for this negativity, albeit with an exponential overhead. This is particularly useful in cases with a small amount of negativity. For example, if a circuit's negativity is constant (i.e., does not scale with the number of qudits), then the simulation algorithm is technically efficiently simulatable. In practice, the negativity will usually scale with the number of modes. However, this technique can still be useful if the negativity is relatively small, which amounts to a small exponential overhead in the simulation time. This boundary between classical and quantum computational power in CV systems is closely linked to resource theories, which we have mentioned in the previous Section \cite{chitambar2019, lami2018, pashayan2015, zurel2024}.

The algorithm presented in Ref.~\cite{pashayan2015} addresses this challenge by considering the decompositions of states in the Wigner representation as given in Eq.~(\ref{eq:dv-wig}). Furthermore, in this formalism, unitaries $\hat U$ can be represented as Wigner functions or transformations of a general Wigner function. However, for this thesis, we do not review the action of operations and refer instead to the covariance property of Clifford operations as discussed in Subsection~\ref{sub:wigner-function-evolution-under-quantum-operations}. Kraus measurement operators are also represented by the same Wigner function given in Eq.~(\ref{eq:dv-wig}). 

The main aim of the simulation algorithm can therefore be reduced to the problem of simulating the action of (Wigner-positive) measurements of quantum states, whereby the state has some amount of negativity. Here, we describe the simulation algorithm for DV systems, although an analogous algorithm can be defined in CV~\cite{pashayan2015} assuming finite resolution. Since we can no longer use the phase space method discussed in Subsection~\ref{sub:classical-simulation-gaussian}, due to the negative values of the Wigner function, we instead make use of the 1-norm of the function, which we explicitly define using the notation $\mathcal M_{\hat \rho}=|\bar W|_1=\sum_{\mathbf r}|\wig*[\mathbf r]|$, which can be interpreted to be the negativity of the function. To perform a simulation, we first define a probability function $\mathfrak p(\mathbf r)=|\wig*[\mathbf r]| / \mathcal M_{\hat \rho}$, where the negativity $\mathcal M_{\hat \rho}$ also serves as a normalisation factor for our probability function. Then, we return an estimation of the true probability function $\Pr(\mathbf r)$ according to samples selected from $\mathfrak p$. Given the random sample $\mathbf r'$ selected from $\mathfrak p$, we estimate the true probability function at that point to be
\begin{align}
    \tilde{\mathfrak p}(\mathbf r')= \mathcal M_{\hat\rho} \operatorname{sign}(\wig*[\mathbf r']) \wig*[\mathbf r'][\hat K]
\end{align}
where $\hat K$ is a measurement Kraus operator.

As is explained in more detail in Ref.~\cite{pashayan2015} and Paper \pE, the expectation value of this function (i.e., the true average over all samples obtained in this way) is equal to the function that we get when calculating the probability function exactly, i.e., $\langle \mathfrak p(\mathbf r') \rangle=\Pr(\mathbf r')$.
Taking many samples makes it possible to build a picture of the true probability function. To be exact, to approximate the true probability function up to additive error $\epsilon$ and success probability $1-\delta$, we require a number of samples $N$ such that
\begin{align}
    N=\frac{2}{\epsilon^2} \mathcal M_{\hat\rho}^2 \log(2/\delta).
\end{align}
Note that the number of samples required to estimate up to an arbitrary error and success probability scales with the negativity of the Wigner function.\clearpage{}
    \clearpage{}\chapter{Simulation of ideal Gottesman-Kitaev-Preskill states}

\label{ch:simulation-of-ideal-gottesman-kitaev-preskill-states}
\DropCap{I}{n}{0.2} this Chapter, we present the results of this thesis that deal with the simulatability of circuits with ideal, infinitely-squeezed GKP states. Although such states are technically not experimentally feasible, they exhibit nice mathematical properties that, as we will see, allow for the design of efficient simulation algorithms for specific circuit classes that we shall detail.

We begin with an overview of the different classes of operations that we have demonstrated are efficiently simulatable in Section~\ref{sec:circuits-initiated-with-bosonic-code-states} before going into a more technical overview of the other simulation algorithms for each class of operations in Sections.
The first technique, detailed in Paper \pA and summarised in Section~\ref{sec:high-dimensional-encodings}, leverages a mapping of the CV problem onto a DV equivalent representation.
The second of these methods, first demonstrated in Paper \pB and outlined in Section~\ref{sec:analytic-techniques-to-track-change-of-state}, involves analysing the wavefunction of the state in the Schr\"odinger picture.
This technique relies heavily on analytic number theory~\cite{apostol1986}.
A third technique, derived in Paper \pC and discussed in Section~\ref{sec:stabiliser-formalism}, involves adapting the stabiliser formalism specifically for CV systems.
This approach provides a novel framework for tracking the evolution of states and operators under specific transformations. 

\section{Circuits initiated with Bosonic code states}

\label{sec:circuits-initiated-with-bosonic-code-states}

This Section introduces simulatable GKP (SGKP) circuits, a class of CV circuits acting on ideal GKP states, along with various subclasses of these circuits.
These circuits are defined by their property of preserving the GKP stabiliser group throughout the computation, and they are the focus of the classical simulation results presented in Papers \pA, \pB, \pC.
Functionally, SGKP circuits serve as an alternative\footnote{I.e., an alternative to Gaussian circuits.} CV analogue of Clifford circuits in DVQC, due to their defining characteristic of preserving the stabiliser structure.
The SGKP circuits considered here begin with ideal GKP states. They are then acted on with rational symplectic transformations (a dense subset of Gaussian operations~\cite{calcluth2023}), continuous displacements, and homodyne measurements.
The most general set of circuits, which contains the different subsets discussed in this Chapter, is displayed in Fig.~\ref{fig:gkp-stabiliser}. We denote the set of operations that were proved simulatable in Paper \pA as $\paperclass[A]$, the two classes proved simulatable in Paper \pB as $\paperclass[B]$ and $\paperclass[B]$, and the class proved simulatable in Paper \pC as $\mathcal Q_C$, which is dense in the set of Gaussian operations. 
\begin{figure}[h!]
     \centering
     \includegraphics[width=0.7\textwidth]{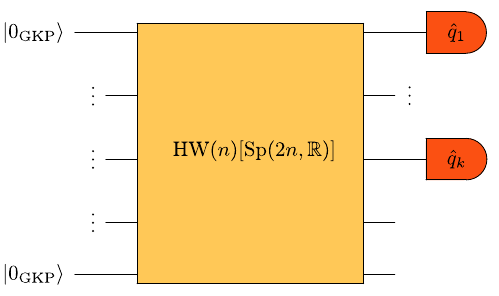}
     \caption{Schematics of the general circuit class that is considered. The input states are $0$-logical GKP-encoded qubit states. The states are acted on by generic Gaussian operations, using the notation defined in Subsection~\ref{sub:gaussian-states}. Homodyne detection of $k$ modes follows, corresponding to the measurement of the quadratures $\hat q_1,\dots \hat q_k$.}
     \label{fig:gkp-stabiliser}
    \end{figure}

The set $\paperclass[A]$ consists of CV operations corresponding to the encoded Clifford operations on the logical \textit{qudit} subspace of the GKP code, combined with arbitrary rational displacements.
The set of operations can be represented as $\paperclass[A]=\clifford*[][d]$, i.e., the set of GKP-encoded qudit Clifford operations.

In Paper \pB, we demonstrated that a different class of operations were efficiently simulatable. We showed that if we restrict to the case of single-mode measurements, we can simulate circuits which consist of operations $\paperclass[B]$. In contrast, we are limited to a smaller set $\paperclass[B']$ if we measure multiple modes. Neither of these sets completely contains or is contained by $\paperclass[A]$.

The restricted class $\paperclass[B]$ of multimode Gaussian operations we consider for the case of single-mode measurements is defined as the direct product of the Heisenberg-Weyl group $\text{HW}(n)$ and a restricted set of symplectic matrices $\text{RSp}(2n,\mathbb R)$,
\begin{align}
    \label{eq:classB}
    \paperclass[B]=\text{HW}(n)\times\text{RSp}(2n,\mathbb R),
\end{align}
where we define
\begin{align}
\label{eq:RSp-set}
\operatorname{RSp}(2n,\mathbb R)
=\left\{\,\begin{pmatrix}A&B\\C&D\end{pmatrix}\in\operatorname{Sp}(2n,\mathbb R)
\; :\;
\begin{aligned}
&A_{1i}=0\;\lor\;B_{1i}=0\;\lor\;\frac{A_{1i}}{B_{1i}}\in\mathbb Q_{(2)},\\
&\forall\,i\in\{1,\dots,n\}
\end{aligned}\right\},
\end{align}
whereby $\mathbb Q_{(2)}$ is formally defined as the localisation of the integers $\mathbb Z$ at the prime ideal $2\mathbb Z$~\cite{atiyah2018}. In simpler terms, it is the set of rational numbers that, when in its simplest form, has an odd denominator.

The set of operations that we show to be simulatable for multimode measurements in Paper \pB, denoted $\paperclass[B']$, is represented by the class of operations
\begin{align}
\label{eq:class-bprime}
\paperclass[B'] = \text{HW}(n) \times \operatorname{DSp}(2n,\mathbb{R}),
\end{align}
where we have defined
\begin{align}
\label{eq:dsp-set}
\operatorname{DSp}(2n,\mathbb{R})
=\left\{\,
\begin{pmatrix}\tilde A&0\\\tilde C&\tilde A^{-T}\end{pmatrix}
\begin{pmatrix}\diagm[\cos\paramrot]&\diagm[\sin\paramrot]\\
-\diagm[\sin\paramrot]&\diagm[\cos\paramrot]\end{pmatrix}
:
\begin{array}{l}
\tilde A\in \operatorname{GL}(n,\mathbb{R})\cap \operatorname{Sym}(n),\\
\tilde C^T\tilde A\in \operatorname{Sym}(n),\\
\paramrot=(\paramrot_1,\dots,\paramrot_n)\in\anglesset^n
\end{array}
\right\}.
\end{align}
Here, we have introduced several notations. First the set $\anglesset$ of allowed angles is defined as \begin{align}     \label{eq:anglesset}     \anglesset={\{\paramrot \in \mathbb R: \cot\paramrot =u/v \in \mathbb Q_{(2)} \}} \cup \{0,\pi\}, \end{align} where $\mathbb Q_{(2)}$ is defined as before. The notation $\operatorname{GL}(n,\mathbb R)$ refers to the real general linear group, i.e., the set of invertible $n\times n$ real matrices, and $\operatorname{Sym}(n)$ refers to the set of $n\times n$ symmetric matrices, i.e., where $\tilde A^T=\tilde A$.

In Paper \pC, we predominantly deal with the set that is the semi-direct product of the real displacements and the rational symplectic operations $\text{HW}(n)[\symp[Q]]$. Technically, we extend this to include a larger set, which we denote $ \paperclass[C]\supset \text{HW}(n)[\symp[Q]]$. However, since the rational symplectic matrices are dense in the set of real symplectic matrices, we refer mostly to the set $\text{HW}(n)[\symp[Q]]$ in our analysis. For completeness, we can define the full set $\paperclass[C]$ formally as
\begin{align}
    \label{eq:paperclass-C}
    \paperclass[C] = \text{HW}(n)[\symp[Q]] \cup \text{HW}(n) \times \operatorname{CSp}(2n,\mathbb R),
\end{align}
whereby
\begin{align}
    \label{eq:CSp}
    \operatorname{CSp}(2n,\mathbb R)=\left\{\,
\begin{pmatrix}\tilde A&0\\\tilde C&\tilde A^{-T}\end{pmatrix}
{\mathfrak s}(X,Y)
:
\begin{array}{l}
\tilde A\in \operatorname{GL}(n,\mathbb{R})\cap \operatorname{Sym}(n),\\
\tilde C^T\tilde A\in \operatorname{Sym}(n),\\
{\mathfrak s}(X,Y) \in \tilde{\operatorname{\mathfrak K}}
\end{array}
\right\},
\end{align}
where we have $\tilde{\operatorname{\mathfrak K}}$ as the elements of the unimodular symplectic matrices $\operatorname{\mathfrak K}$~\cite{arvind1995} such that $X^TX,Y^TX$ are rational matrices. The unimodular symplectic matrices (i.e., the maximal compact subgroup) can be defined in terms of their elements
\begin{align}
    \operatorname{\mathfrak K}= \left\{{\mathfrak s}(X,Y)=\mqty(X&Y\\-Y&X) \in \operatorname{Sp}(2n,\mathbb R): \begin{aligned} &X^TX+Y^TY=\mathbbm 1,\\
    &XX^T+YY^T=\mathbbm 1,\\
    &X^TY,XY^T \in \operatorname{Sym}(n)
    \end{aligned}\right\}.
\end{align}

The relationship between these different classes of operations is shown in Fig.~\ref{fig:venn}. We see that the set $\paperclass[A]$ extends the set of operations simulatable by the qubit-encoded Clifford group. We also see that the set of operations $\paperclass[B]$, which we proved simulatable for single-mode operation, neither contains nor is contained by the set of GKP-encoded Clifford operations. Furthermore, it extends beyond (but does not fully contain) the set of rational symplectic operations with continuous displacements. The set of operations $\paperclass[B']$, which we proved are simulatable for multimode measurements, also offers an alternative regime to the Clifford group simulations. It is smaller than $\paperclass[B]$ and is, unlike $\paperclass[B]$, contained within $\paperclass[C]$. $\paperclass[C]$ contains the full set of rational symplectic operations and continuous displacements. Although it does contain $\paperclass[B']$, it is unknown whether it contains $\paperclass[B]$. As a reminder, $\operatorname{Sp}(2n,\mathbb Q)$ is dense in  $\operatorname{Sp}(2n,\mathbb R)$ and, therefore, the set of operations in $\paperclass[C]$ can approximate any operation in $\operatorname{Sp}(2n,\mathbb R)$ up to arbitrary resolution.

We now review each of the simulation algorithms in detail.

\definecolor{colB}{HTML}{FB5012}
\definecolor{rationals}{HTML}{AAFAC8}
\definecolor{colA}{HTML}{00ff99}
\definecolor{colBdash}{HTML}{FFC857}

\definecolor{colC}{HTML}{fd67ff}
\definecolor{colBC}{HTML}{e77942}

\definecolor{colClif}{HTML}{0ea69c}
\definecolor{colClifB}{HTML}{1D4E9E}
\definecolor{Bdashrationals}{HTML}{e1da7f}
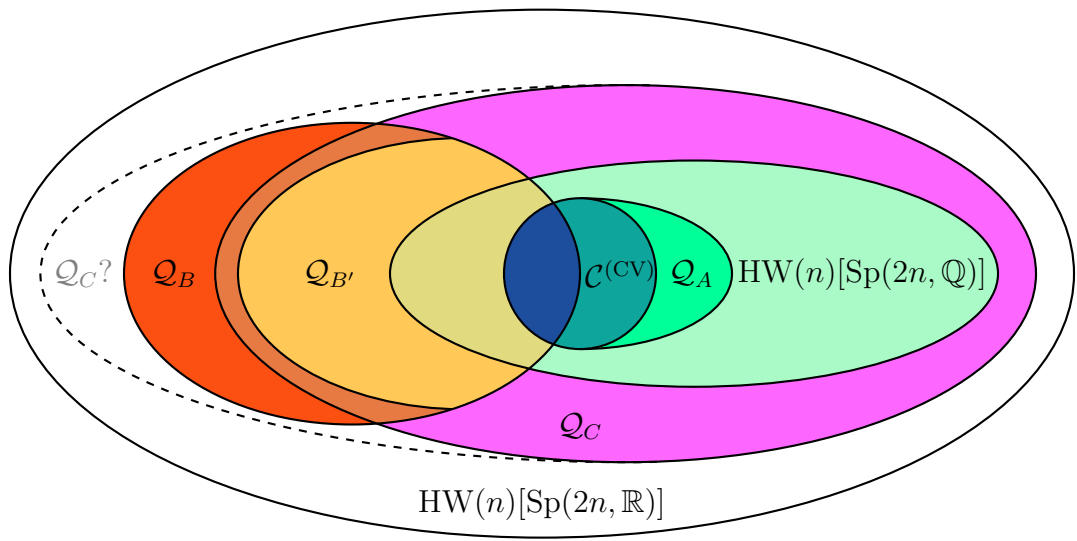
\begin{figure}[ht]
    \centering
         \begin{tikzpicture}[fill=gray]

    \fill[color=colC] (3.6,0) ellipse (5.4 and 2.5);
    \fill[color=rationals] (4.5,0) ellipse (4 and 1.5);
\fill[fill=colB] (0,0) ellipse (3 and 2);
\scope
    \clip (0,0) ellipse (3 and 2);
    \fill[fill=colBdash] (1.5,0) ellipse (3 and 1.8);
    \fill[fill=colBC, even odd rule]
    (3.6,0) ellipse (5.4 and 2.5)   (1.5,0) ellipse (3 and 1.8);        \endscope

    \scope
    \clip (0,0) ellipse (3 and 2);
    \fill[color=Bdashrationals, even odd rule]
    (4.5,0) ellipse (4 and 1.5)
    (3,0) circle (1);
    \endscope
    
    \scope
    \clip (3,-3) rectangle (6,3);
    \fill[color=colA] (3,0) ellipse (2 and 1);
    \draw[thick] (3,0) ellipse (2 and 1);
    \endscope
    
    \fill[fill=colClif] (3,0) circle (1);
\scope
    \clip (3,0) circle (1);
    \fill[fill=colClifB] (0,0) ellipse (3 and 2);
    \endscope
    \node[below] at (2.5,-2.7) {$\text{HW}(n)[\text{Sp}(2n,\mathbb R)]$};
    \node[left] at (8.5,0) {$\text{HW}(n)[\text{Sp}(2n,\mathbb Q)]$};
    \node[below] at (3,-1.7) {$\mathcal Q_C$};
    
\draw[thick] (0,0) ellipse (3 and 2);
\scope
    \clip (0,0) ellipse (3 and 2);
    \draw[thick] (1.5,0) ellipse (3 and 1.8);
    \endscope
\draw[thick] (3,0) circle (1);

\draw[color=black, thick] (4.5,0) ellipse (4 and 1.5);

    \draw[color=black, thick] (3.6,0) ellipse (5.4 and 2.5);

    \scope
    \clip (-4.5,-3) rectangle (4, 3);
    \draw[color=black, dashed, thick] (3.6,0) ellipse (7.7 and 2.5);
    \endscope
    \node[left,opacity=0.5] at (-3,0) {$\mathcal Q_C$?};
    
    \node[left] at (4.9,0) {$\mathcal Q_{A}$};
    \node[left] at (-1.9,0) {$\mathcal Q_{B}$};
    \node[left] at (0.2,0) {$\mathcal Q_{B'}$};
    \node[left] at (4.15,0) {$\mathcal C^{(\text{CV})}$};
    \draw[thick] (2.5,0) ellipse (7 and 3.5);
    \end{tikzpicture}
    \caption{The classes of circuits we consider are displayed as a Venn diagram. We know that all $\paperclass[A],\paperclass[B],\paperclass[B'],\paperclass[C],\clifford*[]$ are contained within the set of Gaussian operations $\text{HW}(n)[\text{Sp}(2n,\mathbb R)]$. We also know that the logical qubit encoded Clifford group $\clifford*[]$ is completely contained by $\paperclass[A]$. Furthermore, we know that $\clifford*[]$ and $\paperclass[A]$ are contained within the set of rational symplectic operations and continuous displacements $\text{HW}(n)[\text{Sp}(2n,\mathbb Q)]$. $\paperclass[B']$ is completely contained by $\paperclass[B]$. Neither is contained by nor completely contains any of the other aforementioned sets except the set of Gaussian operations. Finally, the set $\paperclass[C]$ contains the set of rational symplectic operations and continuous displacements $\text{HW}(n)[\text{Sp}(2n,\mathbb Q)]$ and the set $\paperclass[B']$ but it is unknown whether it completely contains $\paperclass[B]$ (hence the dashed ellipse around $\paperclass[B']$). The size of each of the regions in the figure does not represent the size of the space contained.
    }
    \label{fig:venn}
\end{figure}

\section{High-Dimensional Encodings}

\label{sec:high-dimensional-encodings}

This Section focuses on the results of Paper \pA.
This work establishes a link between the simulatability of a subset of CV circuits and the generalised Gottesman-Knill theorem for DV systems of higher dimension.
We identify vast families of circuits initiated with GKP states (and, in Paper \pA, rotation-symmetric Bosonic code states) that are classically efficiently simulatable.
These families encompass circuits initiated with encoded logical basis states undergoing sequences of specific unitary operations related to the encoding's discrete symmetries in a higher logical dimension, and are measured with encoded Pauli measurements.

The result is achieved by mapping the CV states and operations onto an effective higher-dimensional qudit space.
Crucially, operations that might not correspond to logical Clifford gates in a standard qubit encoding can become equivalent to Clifford operations within this higher-dimensional qudit picture.
This correspondence allows for efficient simulation by leveraging extensions of the Gottesman-Knill theorem applicable to higher-dimensional DV systems.

\subsection{Simulation beyond the logical basis}

\label{sub:techniques}

In this Subsection, we describe how this technique simulates low-dimensional GKP states by simulating high-dimensional GKP states and their associated encoded operations.
The core technique of our first simulation algorithm for GKP-Gaussian circuits exploits the insight that GKP states encoded in a lower dimension ($\quditd[1]$) can be formally interpreted as superpositions of logical states within a higher-dimensional GKP encoding ($\quditd[2]$).
This allows us to leverage the simpler structure of operations in the higher dimension for classical simulation.

Stabiliser GKP states encoded in some dimension $\quditd[1]$ have the same representation in the CV Hilbert space as certain encoded stabiliser states in dimension $\quditd[2]$, where $\quditd[2]=\quditfactor^2\quditd[1]$, $\quditd[1],\quditd[2]\in \{2,3,\dots\}$ and $a \in \{1,2,\dots\}$.
To understand why, we can inspect the stabilisers of the states.
For example, consider the state $\gkpket$.
We know that the ideal $\gkpket$ state in $\quditd[1]=2$ is stabilized by $\langle \lX[2]^2, \lZ[2] \rangle$.
Consider GKP states in $\quditd[2]=8$, stabilized by $\langle \lX[8]^4, \lZ[8]^2 \rangle$, where $\lX[8]=e^{-i \sqrt{\pi} \hat p/2}$ and $\lZ[8]=e^{i\sqrt{\pi}\hat q/2}$.
Observing the relation between these operators, we find $\lX[8]^2 = \lX[2]$ and $\lZ[8]^2 = \lZ[2]$.
This implies that the state stabilized by $\langle \lX[2]^2, \lZ[2] \rangle$ is precisely the same state stabilized by $\langle\lX[8]^4, \lZ[8]^2\rangle$.
In the context of the higher $\quditd[2]=8$ dimension, the state defined by these stabilisers is identified as a specific state within the GKP $\quditd[2]$ encoding, which is a qudit stabiliser state~\cite{gottesman1999b, hostens2005,gheorghiu2014}.
We can also find this relation by inspecting the position-based representation of the states.
Specifically, the GKP qubit state $\gkpket[0][2]$, as a $\quditd[1]=2$ encoding, is equivalent to a superposition of logical states in a $\quditd[2]=8$ encoding.
This equivalence is explicitly shown by the normalized superposition $\gkpket[0][2] = \frac{1}{\sqrt{2}}(\gkpket[0][8]+\gkpket[4][8])$, where we use $\gkpket[j][d]$ to denote the $j$-th logical state in a $\quditd$-dimensional GKP encoding.
The $0$-logical GKP state in dimension $\quditd[1]=2$ is given by $\gkpket=\sum_{n \in \mathbb Z}\posket[2n\sqrt\pi]$, meanwhile we express the state $\gkpket[0][8]+\gkpket[4][8]=\sum_n \posket[4\sqrt{\pi}n]+\posket[4\sqrt{\pi}n+4\sqrt{\pi}/2]$, leading to the same expression in the position basis.

This principle holds generally: any $\quditd[1]$-dimensional GKP logical state can be expressed as a superposition of $\quditd[2]$-dimensional GKP logical states, where $\quditd[2] = \quditfactor^2 \quditd[1]$ for some integer $\quditfactor$.
This underlying mathematical equivalence, applicable across different dimensions and logical states, is the foundation for the simulation algorithm.
It allows a specific class of operations, acting on these lower $\quditd[1]$-dimensional states, to be reinterpreted as simpler, classically tractable operations in the higher-dimensional encoding.

As detailed in Paper \pA, interpreting the states and operations within the higher $\quditd[2]$-dimensional framework identifies a specific class of circuits that become classically tractable, despite potentially appearing complicated when considered solely within the lower $\quditd[1]$ picture.
Specifically, the simulation algorithm presented in Paper \pA efficiently handles circuits composed of encoded $\quditd[2]$-dimension Clifford operations and displacements by multiples of $\peakspace[\quditd[2]]$, for any arbitrary choice of the encoding factor $\quditfactor$.

Specifically, given a $\quditd[1]$-dimensional qudit, which ordinarily have peaks separated by $\peakspace[\quditd[1]]=\sqrt{2\pi/\quditd[1]}$, we can now include displacements in $\hat q$ and $\hat p$ by any multiple of $\peakspace[\quditd[2]]=\peakspace[\quditd[1]]/a$ for any arbitrary choice of $\quditfactor$.
This finding unlocks the ability to use the $\quditd[2]$-dimensional stabiliser formalism, which, as we have seen in Section~\ref{sec:stabiliser-formalism}, provides a powerful framework for describing and simulating Clifford circuits, which map stabiliser states to stabiliser states.

For this specific class of GKP-based circuits, defined by GKP states and encoded Clifford operations over the finer $\quditd[2]$ grid, the simulation complexity scales polynomially with the number of modes.
This enhanced resolution implies that the simulation algorithm efficiently handles encoded displacement operations that translate the state by amounts smaller than the native $\quditd[1]$ lattice spacing, specifically any multiple of $\peakspace[\quditd[2]]$, which can be made arbitrarily small.

\subsection{Example}

\label{sub:example}

This Subsection provides a concrete example demonstrating the implementation and analysis of high-dimensional GKP encodings using the simulation techniques discussed in this Section.
We will illustrate how stabiliser GKP states can be manipulated using the class $\paperclass[A]$ consisting of a restricted set of Gaussian operations, corresponding to higher-dimensional GKP-encoded Clifford operations.
By examining explicit circuit constructions, we highlight the practical aspects of working with these complex encodings.

One of the simplest examples of circuits that are simulatable using the techniques of Paper \pA is one in which we start with a $\gkpket$ state, apply a small displacement $r$ in position, followed by a Clifford operation, such as a Fourier transform, and finally measure using homodyne measurement in the position basis.
The PDF of this circuit is given as $\pdfunc=|\langle\hat q=x|\fourier \xdisp[r] \gkpket|^2$.
Previous theorems cannot provide a method to calculate this evolution in a way that generalises to more complex cases, such as for multiple modes (in practice, in this simple single-mode case, an analysis in the Schr\"odinger picture would suffice).

We can use our technique to simulate any displacement $x=x'\sqrt{\pi}$ where $x'\in \mathbb Q$.
To do so, we first write $x=u/v$ and then choose $\quditfactor=v$, such that $\peakspace[\quditd[2]]=\sqrt{2\pi/\quditd[2]}=\sqrt{2\pi/(2\quditfactor^2)}=\sqrt\pi/\quditfactor=\sqrt\pi /v$, which allows us to produce displacements $\lX[\quditd[2]]^u=\xdisp[u\peakspace[\quditd[2]]]=\xdisp[\sqrt\pi u/v]$. These displacements are hence interpreted as Pauli gates in dimension $\quditd[2]$, covered by the simulatability results of Ref.~\cite{hostens2005}.
The main insight of our paper can be visualised\footnote{Figure adapted from presentations by Laura Garc\'ia \'Alvarez.} in Fig.~\ref{fig:gkp-displace}, as a simple example of a displacement of a $0$-logical qubit by half the peak distance.
\begin{figure}[h]
    \centering
   \begin{tikzpicture}
\definecolor{mygrey}{RGB}{150,150,150}    

\def\len{0.8}
  \foreach \x in {0,...,8} {
    \draw[dashed,mygrey] (\x,-0.1) -- (\x,2.9);
  }
\draw[solid,black] (-0.2,0) -- (8.2,0);
\draw[solid,black] (-0.2,1) -- (8.2,1);
\draw[solid,black] (-0.2,2) -- (8.2,2);
\foreach \x in {0,...,1} {
    \draw[->, line width=1, myblue] (4*\x+2,2) -- ++(0,\len);
    \draw[->, line width=1, myblue] (4*\x+2+0.5,1) -- ++(0,\len);
    \draw[->, line width=1, myblue] (4*\x+2+1,0) -- ++(0,\len);
    \draw[->, line width=1, myorange] (4*\x+0.5,1) -- ++(0,\len);
    \draw[->, line width=1, myorange] (4*\x+1,0) -- ++(0,\len);
  }
  \foreach \x in {0,...,2} {
    \draw[->, line width=1, myorange] (4*\x,2) -- ++(0,\len);
  }

\node[scale=0.9,right] at (9.41,2+0.4) {\(\gkpket[0][2]=\gkpket[0][8]+\gkpket[4][8]\)};
  \node[scale=0.9,right] at (8.3,1+0.4) {\(\hat X(\frac{\sqrt{\pi}}{2})\gkpket[0][2]=\gkpket[1][8]+\gkpket[5][8]\)};
  \node[scale=0.9,right] at (9.41,0+0.4) {\(\gkpket[1][2]=\gkpket[2][8]+\gkpket[6][8]\)};
  
  \node[right] at (3.78,-0.4) {${\scriptstyle 0}$};
  \node[right] at (4.6,-0.4) {${\scriptstyle\sqrt\pi}$};
  \node[right] at (5.5,-0.4) {${\scriptstyle2\sqrt\pi}$};
  \node[right] at (6.5,-0.4) {${\scriptstyle3\sqrt\pi}$};
  \node[right] at (7.5,-0.4) {${\scriptstyle4\sqrt\pi}$};
  \node[right] at (2.4,-0.4) {${\scriptstyle-\sqrt\pi}$};
  \node[right] at (1.4,-0.4) {${\scriptstyle-2\sqrt\pi}$};
  \node[right] at (0.4,-0.4) {${\scriptstyle-3\sqrt\pi}$};
  \node[right] at (-0.6,-0.4) {${\scriptstyle-4\sqrt\pi}$};
  
\draw[solid,myorange, line width=1] (11.47,0.12) -- (12.83,0.12);
\draw[solid,myblue,line width=1] (13.33,0.12) -- (14.69,0.12);
\draw[solid,myorange,line width=1] (11.47,1.07) -- (12.83,1.07);
\draw[solid,myblue,line width=1] (13.33,1.07) -- (14.69,1.07);
\draw[solid,myorange,line width=1] (11.47,2.12) -- (12.83,2.12);
\draw[solid,myblue,line width=1] (13.33,2.12) -- (14.69,2.12);
\end{tikzpicture}
    \caption{Representation of the set of displacements made available in the set of operations $\paperclass[A]$. From top to bottom, consider starting with the $0$-logical GKP-encoded qubit state. We apply a displacement of $\sqrt{\pi}/2$, which is outside of the logical code space. This displacement can be alternatively understood as a displacement in $d=8$, acting on a qubit, encoded within a qudit, encoded as a GKP state. A second displacement by this distance, brings the state back into the code space and into the $1$-logical GKP-encoded qubit state.}
    \label{fig:gkp-displace}
\end{figure}
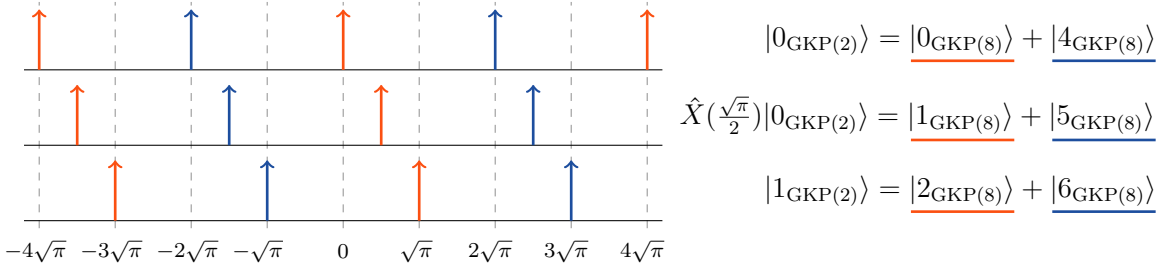 
\subsection{Implications}

\label{sub:implications}

As established in Paper \pA, the classical simulatability of GKP-type circuits utilising operations within class $\paperclass[A]$ is restricted by certain conditions.
These conditions relate to the class of GKP input states, the set of allowed quantum operations, and the nature of measurements performed.
Specifically, circuits initiating with GKP-encoded stabiliser states, followed by encoded Clifford operations, and arbitrary displacements and measured with homodyne measurement are shown to be classically simulatable.
Deviation from this class, for instance, by incorporating GKP-encoded non-stabiliser states such as magic states or applying non-Clifford operations, generally results in a non-simulatable circuit and potentially (but not necessarily) capable of demonstrating QA.

This result provides a concrete example of a circuit class that, at first sight, appears to be difficult to simulate but is actually efficiently simulatable. By demonstrating a broad class of non-trivial circuits that are \textit{unable} to unlock an advantage, the result narrows the search to find which factors \textit{do} allow us to cross the boundary from simulatability to universality.

While, as we have seen, Wigner negativity is essential for non-Gaussian QC, its role as a resource for achieving universal QA is nuanced and requires careful analysis \cite{mari2012, veitch2013, rahimi-keshari2016}.
This work demonstrates that efficient classical simulation is possible for specific classes of circuits utilising high-dimensional encodings despite significant Wigner negativity.

\subsection{Alternative algorithm}

\label{sub:alternative-algorithm}

This Subsection outlines an alternative and conceptually more straightforward method for simulating the class of CV circuits, consisting of operations $\paperclass[A]$, which were proved to be efficiently simulatable in Paper \pA.

Specifically, the class consists of CV circuits comprising gates  $\paperclass[A]$ that implement logical Clifford operations on GKP-encoded qudit states and rational displacement operations.
This alternative approach emerged from subsequent research, specifically, in Paper \pC, that built upon the work presented in Paper \pA, and it is included here for comparison and completeness.

The set of efficiently simulatable operations $\paperclass[A]$ consists of two types: $\quditd[2]$-dimensional encoded Clifford operations and rational (up to a factor of $\sqrt{\pi}$) displacements.
In the CV formalism, these correspond to Gaussian operations, specifically unitary gates implementing symplectic transformations (for the Cliffords) and displacement operators.

First, we see that the set of $\quditd[2]$-dimensional encoded Clifford operations is efficiently simulatable.
This set consists of a group of operators generated by the encoded Fourier transform, the encoded phase gate and the encoded CNOT gate.
In CV, the encoded versions of these gates are precisely the same in all even dimensions $\quditd[2]$ (which is even by definition).
Hence, the set of encoded $\quditd[2]$-dimensional encoded Clifford operations is equivalent to the set of $2$-dimensional encoded Clifford operations.

Second, rational displacements can be efficiently handled by utilising the commutation relation between a displacement operator $\disp[\mathbf r]$ and a unitary gate $\unitary_{\sympmat}$ implementing a symplectic transformation $\sympmat$.
The relation given in Eq.~(\ref{eq:comm-disp-symp}) shows how a displacement operator is transformed when commuted past a unitary gate.
Repeating this rule throughout a circuit allows all displacement operators to be moved to one side (i.e., the right-hand side), separating the sequence of symplectic transformations from the displacements.

This separation of symplectic and displacement operations significantly simplifies the simulation process.
Instead of tracking the state's evolution through interleaved operations, we can simulate the sequence of symplectic transformations on the initial state and then apply the transformed displacement operation as a single final step.

While this alternative approach offers a conceptually more straightforward route to simulate this specific subclass of circuits in hindsight, it is important to recall why the methods developed in Paper \pA were significant.
The techniques presented in Paper \pA provided the first proof of efficient simulatability for a class of circuits that explicitly exhibit high degrees of Wigner negativity, a key resource for QC.

\section{Analytic techniques to track change of state}

\label{sec:analytic-techniques-to-track-change-of-state}

In this Section, we present the results of Paper \pB, which investigates another two subsets of operations, $\paperclass[B]$ and $\paperclass[B']$, from the set $\symp[R]$.
These results are derived using entirely different techniques compared to the previous Section.
Instead, we use results from analytic number theory to derive expressions for the PDFs of the measurement results of a general circuit.

\subsection{PDF of Rotated GKP State}

\label{sub:pdf-of-rotated-gkp-state}

This Subsection presents the analytical derivation of the position-basis PDF for an ideal single-mode GKP state after undergoing arbitrary single-mode Gaussian operations, focusing on analysing the PDF structure after a phase-space rotation.
We calculate the evolution of the logical state $\gkpket$ in the Schr\"odinger picture and evaluate the state's wavefunction in the position basis under these operations.
As shown in Eq.~(\ref{eq:gkp-wf}), the ideal single-mode GKP state encoding the logical codeword $\gkpket$ has a wavefunction which is an infinite comb of Dirac delta functions, centred at positions $x = 2m\sqrt{\pi}$ for integer $m$.

While previous work has provided analytical forms for the wavefunction of ideal GKP codewords under specific restricted sets of Gaussian operations, such as those corresponding to GKP-encoded Clifford gates \cite{gottesman2001, bartlett2002}, a comprehensive analytical description for the outcome of position measurements after an arbitrary single-mode Gaussian operation had not been explored until our work.
Paper \pB addresses this gap by deriving a general analytical expression for the PDF in the position basis for an ideal single-mode GKP state subjected to any single-mode Gaussian transformation.

To achieve this, we make use of the property that unitary single-mode Gaussian operations, described by $2\times2$ symplectic matrices, can be decomposed via the Iwasawa decomposition \cite{arvind1995} into a product of a shear, a squeezing, and a rotation.
When considering position-basis measurements, shear operations correspond to momentum shifts, leaving the position-basis PDF unchanged.
Similarly, displacement operations only result in a simple real shift of the PDF's argument.
Consequently, the non-trivial effects of any single-mode Gaussian operation on the position distribution of a GKP state are captured by a rotation of angle $\paramrot$, single-mode squeezing with parameter $\paramsq$, and a position-basis displacement with parameter $\paramdisp$.

This state's wavefunction in the position basis is given by 
\begin{align}
\wf_{\paramsq,\paramrot,\paramdisp}(x) = \bra{q=x}\xdisp[\paramdisp]\squeezing \rot \gkpket.
\end{align}

We see immediately that we can rewrite the rotated, squeezed and displaced wavefunction $\wf_{\paramsq,\paramrot,\paramdisp}(\cdot)$ in terms of the rotated wavefunction $\wf_{\paramrot}(\cdot)$ as
\begin{align}
\wf_{\paramsq,\paramrot,\paramdisp}(x) =&\wf_{\paramsq,\paramrot}(x-\paramdisp) \\
=&\wf_{\paramrot}((x-\paramdisp)/\paramsq),
\end{align}
and hence, we only need to consider the wavefunction $\wf_{\paramrot}(x)$ of the rotated GKP state.

The full calculations, given in Appendix A1 of Paper \pB, demonstrate that the PDF of a rotated GKP state can be described as
\begin{align}
    \abs{\wf_{\paramrot}(x)}^2\propto&\frac{1}{ \sin\paramrot} \abs{\jacobi(\thetaz=-x \csc\paramrot /\sqrt\pi;\thetatau=2\cot\paramrot)}^2,
\end{align}
where $\jacobi$ is the Jacobi theta function defined in Eq.~(\ref{eq:jacobi}).

Through an analytic number theory argument, we demonstrate that the PDF simplifies significantly for a specific set of angles $\anglesset$, which was defined in Eq.~(\ref{eq:anglesset}).
The condition that $\theta$ must be a member of $\anglesset$ arises from the structure of the GKP lattice in phase space and its symmetry properties under phase-space rotations.
For $\paramrot\in\anglesset$, the PDF of the rotated GKP state takes the form of a Dirac comb: \begin{align}\label{eq:rotated-wfs-main}\abs{\psi_{\theta}(x)}^2=\sum_{m}\delta(x-m \peakspace*).\end{align}
The spacing $\peakspace*$ of this comb depends on the specific form of $\cot\paramrot$, i.e.,
\begin{align}\label{eq:deltacases-main}
\peakspace*=\begin{cases}     \sqrt{\pi}\sin\theta/v\quad & \text{ if } \cot\theta=u/v: u\in\mathbb Z,v\in\mathbb Z_{\text{odd}},\\     2\sqrt\pi\quad &\text{ if } \theta=k\pi \text{ for } k \in \mathbb Z.
\end{cases}\end{align}

This analytical result demonstrates that, for the considered rotation angles $\anglesset$, a position measurement on the transformed state yields a squeezed and displaced Dirac comb.
Note that in the following Sections we will prove this result for all angles with rational cotangent.
Meanwhile, evaluating the PDF for rotations with irrational cotangent remains an open problem, as Appendix I details.
The spacing of this comb, determined by $\peakspace*$, explicitly depends on the rotation angle $\paramrot$ and the integer $v$ associated with its cotangent.
The squeezing parameter $\paramsq$ uniformly rescales this spacing across all outcomes.

In realistic experimental settings, finite squeezing broadens these delta peaks into Gaussians, leading to a normalisable distribution, a topic we explore further in Chapter~\ref{ch:simulation-of-realistic-gottesman-kitaev-preskill-states}.

\subsection{Multimode circuits with single-mode measurement}

\label{sub:multimode-circuits-with-single-mode-measurement}

Leveraging the results of the PDF for a single mode circuit SGKP circuit, this Subsection investigates the classical simulatability of the restricted class of CV quantum circuits that are initialised with multiple ideal GKP states, acted on by a multimode operation selected from the class  $\paperclass[B]$ and feature a final measurement on a single mode in the position basis. A circuit diagram representing these circuits is shown in Fig.~\ref{fig:single-mode-class}.
As a reminder, the class of operations $\paperclass[B]$ consists of continuous displacements and symplectic operations selected from the set $\operatorname{RSp}(2n,\mathbb R)$ defined in Eq.~(\ref{eq:RSp-set}).

The constraints on the symplectic matrix elements are precisely those that ensure that the individual transformed modes' Heisenberg measurement operators $\transq_j^{\paramrot_j}$ resulting from the operation, when measured in the position basis, yield probability distributions that retain the structure of a scaled GKP position-basis PDF, as characterised by the parameters $\peakspace*_j$ and $\paramsq_j$ derived in the previous Subsection.
Without these constraints, the position-basis PDFs of the transformed modes would not necessarily be simple scaled Dirac combs, potentially making the final integral evaluation intractable.

Crucially, the structure of this PDF, which is non-zero only at specific discrete points $x_1 = \sum_{i=1}^{n}\paramsq_i m_i \peakspace*_i+\paramdisp$ determined by integer combinations of $s_i\peakspace*_i$, indicates that we are sampling from a discrete set of outcomes.
Therefore, the simulation of sampling from this circuit can be performed efficiently by generating $n$ random integers $m_i$ for each shot and computing the corresponding outcome.

The representation of the PDF simplifies to a Dirac comb in a structured way, which enables efficient classical simulation.
This simplification is a direct consequence of the specific structure of the GKP input states and the constraints imposed on the symplectic operations within the class $\paperclass[B]$, ensuring that the transformed single-mode states retain their Dirac-comb-like position-basis probability distributions.

The core reason for the efficiency lies in that, despite the PDF being continuous, it is non-zero only at a discrete, calculable set of points.
Simulating sampling from such a distribution reduces to sampling integers from a discrete set and computing the outcome, a task polynomial in the number of modes and the precision required for the $\paramsq_i$ and $\peakspace*_i$ values.
This contrasts with the general case, where simulating a continuous PDF might require computationally expensive integration.

\begin{figure}[!ht]
     \centering
     \includegraphics[width=0.8\textwidth]{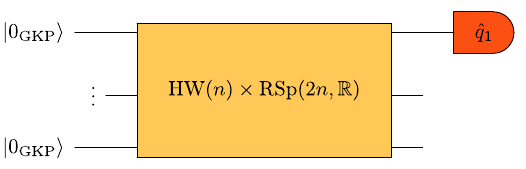}
     \caption{Schematics of the circuit considered in Subsection \ref{sub:multimode-circuits-with-single-mode-measurement}. The circuit is initialised with $0$-logical GKP states. The operations considered are in ${\text{HW}(n)\times \text{RSp}(2n,\mathbb R)}$, which is a restricted set of Gaussian operations, defined in Eq.~(\ref{eq:classB}). Homodyne detection follows, corresponding e.g. to the measurement of the quadrature $\hat q_1$.}
     \label{fig:single-mode-class}
    \end{figure} \subsection{General multimode circuits}

\label{sub:general-multimode-circuits}

Extending the prior analysis of ideal GKP states to the \textit{true} multimode setting --- i.e., including multimode measurements --- presents significant theoretical challenges.
This Subsection details simulation techniques for a specific class of multimode GKP circuits consisting of operations selected from the set $\paperclass[B']$, defined in Eq.~(\ref{eq:class-bprime}), which consists of continuous displacements and symplectic operations where the symplectic matrix is selected from the restricted set $\operatorname{DSp}(2n,\mathbb R)$, defined in Eq.~(\ref{eq:dsp-set}). The set of circuits described by these operations and measurements is shown in Fig.~\ref{fig:multimode-class}, also outlined in terms of circuit decompositions in Paper \pB.

The analysis given in Paper \pB  demonstrates that circuits with operations selected from the class $\paperclass[B']$ are classically simulatable with a runtime that scales polynomially with the number of modes $n$.
The computational cost of evaluating the PDF for a given measurement outcome $\mathbf x$ is determined by the cost of computing the elements $a^{(j)}_i$ (from the matrix $\tilde A$) and the parameters $\peakspace*_i$.
As discussed, the parameters $\peakspace*_i$ can be computed efficiently in $O(n)$ time.
The matrix $\tilde A$ is of size $n \times n$, requiring the computation of $n^2$ coefficients $a^{(j)}_i$.

\begin{figure}
     \centering
     \includegraphics[width=0.8\textwidth]{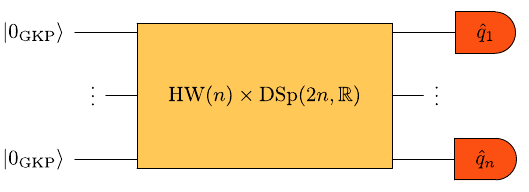}
     \caption{Circuit class considered in Subsection \ref{sub:general-multimode-circuits}. The circuit is initialised with $0$-logical GKP states, acted on by operations from the set ${\text{HW}(n)\times \text{DSp}(2n,\mathbb R)}$, defined in Eq.~(\ref{eq:class-bprime}), followed by homodyne detection of multiple modes.}
     \label{fig:multimode-class}
\end{figure} 
\section{Stabiliser formalism}
\label{sec:stabiliser-formalism}
In this Section, we develop an efficient classical simulation technique applicable to the largest class of simulatable operations considered in this thesis, represented by operations selected from $\paperclass[C]$, defined in Eq.~(\ref{eq:paperclass-C}). As a reminder, these circuits consist of continuous displacements, symplectic matrices, whereby the matrix elements are rational, along with an extension to a subset of the non-rational symplectic matrices. We begin by giving a high-level overview of the simulation technique, before demonstrating how to derive the PDF of the circuit.

\subsection{Simulation technique}

\label{sub:simulation-technique}

To develop an efficient classical simulation technique, we use inspiration from the \linebreak Gottesman-Knill theorem for DV systems~\cite{gottesman1999a,gottesman1999,gottesman1999b}. The method defines an alternative CV version of the stabiliser formalism to simulate circuits consisting of those with operations selected from the class $\paperclass[C]$.
This technique also allows efficient classical simulation of quantum circuits employing GKP states and operations within the Clifford group.
Our method focuses on deriving the PDF by exploiting the structure imposed by the GKP stabilisers and the evolution under Gaussian and logical Clifford operations.

This method reveals that the output PDF for circuits initialised with ideal GKP states and acted upon by Gaussian and logical Clifford operations is a sum of delta functions peaked on a lattice in phase space.

We first make use of the fact that the action of Gaussian operations on the phase space variables $(\hat q, \hat p)$ is a linear transformation described by a symplectic matrix $\sympmat \in \symp[R]$ and a displacement vector $\vec c$.
We identify periodicity relations in the PDF by analysing the commutation relations of the Heisenberg-evolved measurement projectors and stabilisers of the initial GKP states.

Following this, we identify the non-zero points of the PDF by constructing bespoke stabilisers from Heisenberg-evolved measurement operators.
These two conditions, which are explained in more detail in the following two Subsections, provide a sufficient restriction to identify the circuit's PDF.

\subsection{Periodicity relations}

\label{sub:periodicity-relations}

A general Gaussian operation $\unitary\in \paperclass[C]$ transforms, in the Heisenberg picture, the measurement operators $\hat q_j$ according to~\cite{bartlett2002,kok2010,serafini2017}
\begin{align}
\label{eq:evolution-operators-q}
    \evolvedQ_j=\unitary^\dagger \hat q_j \unitary=\sum_i A_{i,j}\hat q_i+B_{i,j}\hat p_i+c_j
\end{align}
where $A,B$ are the blocks of $S$ as defined in Eq.~(\ref{eq:symp-block}).
The vector $\mathbf c\in \mathbb R^n$, with elements $c_j$, describes the displacement in position.
We will now prove that these circuits can be simulated in the strong sense by calculating the PDF.
The PDF can be written in the Heisenberg picture using Eq.~(\ref{eq:evolution-operators-q}) as
\begin{align}
    \label{eq:fullpdf-maintext}
    \pdfunc[\evolvedQs=\mathbf x]=\bra{\mathbf 0_{\text{GKP}}}\left(\bigotimes_{j=1}^n \ket{\evolvedQ_j=x_j}\bra{\evolvedQ_j=x_j}\right) \gkpket.
\end{align}
Inserting the stabilisers of the $0$-logical GKP state, $\lX^2$ and $\lZ$, into this expression and using commutation relations informs us that the PDF has the same value when transforming the position variables of the PDF according to the periodicity relation
\begin{align}
    \label{eq:periodicitymultimode}
    x_j\to x_j'=x_j+\sqrt\pi \sum_k 2a_{k}^{(j)}m_k+b_{k}^{(j)}m'_k,
\end{align}
where $m_k,m_k'\in\mathbb Z$.

Tracking the evolution of a Gaussian state thus requires updating these \linebreak ${(2n)(2n+1)/2 + 2n}$ real parameters.
However, while important, the periodicity relation does not tell us what the values of the PDF are at any point; it only informs us when multiple points have the same value. To fully characterise the PDF, we require more information.

\subsection{Non-zero points}

\label{sub:non-zero-points}
Here, we show that it is possible to identify where the PDF is zero and non-zero. This, in combination with the periodicity relation in Eq.~(\ref{eq:periodicitymultimode}), will provide us with all the information required to fully characterise the PDF.

To do so, we first manufacture additional stabilisers in terms of the Heisenberg measurement operators of the form
\begin{align}
\label{eq-stabiliser-g}
    \customstab(\stabvec)=&e^{i\phi(\stabvec)}\prod_j e^{i\sqrt\pi l_j \hat Q_j}
\end{align}
where $\stabvec$ is an $n$-vector of real coefficients $l_j$ and
\begin{align}
\phi(\stabvec)= -\frac 1 2\pi \stabvec^T AB^T\stabvec-\sqrt\pi \stabvec\cdot \mathbf{c}\;.
\end{align}
Using Eq.~(\ref{eq:evolution-operators-q}), one realizes that the operator $g(\stabvec)$ will be a stabiliser whenever ${(A^T\stabvec)_k=0\mod 1}$ and ${(B^T\stabvec)_k=0\mod 2}$.
This imposes a constraint on the values of $\mathbf{x}$ for which the PDF is non-zero.
The constraint arises from the fact that applying the candidate stabiliser Eq.~(\ref{eq-stabiliser-g}) to the state $\gkpket$ and calculating the resulting wavefunction in the position basis with respect to the Heisenberg-evolved position eigenstates results in a multiplying phase.
Since we know that the wavefunction should be unchanged, the resulting multiplying phase should be zero.
The points at which the wavefunction (and therefore also the PDF in Eq.~(\ref{eq:fullpdf-maintext})) is non-zero can therefore be found by solving the equation
\begin{align}
    \label{eq:condition2-maintext}
    \sqrt\pi \stabvec^T \mathbf{x}-\frac 1 2 \pi \stabvec^T AB^T\stabvec-\sqrt\pi \stabvec^T \mathbf{c}=0\mod 2\pi
\end{align}
for all $\stabvec$ such that $\customstab(\stabvec)$ is a valid stabiliser.

To solve this constrained equation, we first find the allowed vectors $\stabvec$.
This can be achieved by introducing the matrix $\halfsymp$, which is defined as $\halfsymp^T=\mqty(A & \frac 1 2 B)$.
Then the constraint on the allowed values of $\stabvec$ is given by $\halfsymp\stabvec=\mathbf{k}$ where $\mathbf{k}$ is a vector of $2n$ integers.
The Moore-Penrose pseudoinverse~\cite{moore1920,penrose1955,ben-israel2003} $\halfsymp^+$ provides a method to find solutions of the form $\stabvec=\halfsymp^+\mathbf{k}$.
The solutions of $\stabvec$ can be found by first finding the Smith decomposition~\cite{newman1972,newman1997,ben-israel2003} of the matrix $v \halfsymp$, where $v$ is the smallest integer for which the elements of the matrix $v \halfsymp$ are all integers.
Using the Smith decomposition of
\begin{align}
    \label{eq:snf}
    v \halfsymp=V\mathcal DU
\end{align}
we identify which integer choices of $\mathbf{k}$ will provide valid solutions of $\stabvec$.
We find that the vectors $\stabvec$ can be expressed as $\stabvec=R\mathbf{m}$ where $\mathbf{m}$ is an $n$-vector of integers \cite{stackexchangesmith} and $R$ is defined as
\begin{align}
    \label{eq:mat-R}
    R=\tilde S^+V\mqty(\mathbbm 1\\0).
\end{align}
We can then rewrite Eq.~(\ref{eq:condition2-maintext}) as a system of linear equations of the form
\begin{align}
    \frac{1}{\sqrt\pi}R^T(\mathbf{x}-\mathbf{c})=\mathbf{t} \mod 2
    \label{eq:condition2-system}
\end{align}
where $\mathbf{t}$ is the main diagonal of the matrix $T=\frac 1 2 R^TAB^TR$, i.e.,
\begin{align}
    \label{eq:t-vec}
    t_i= \frac 1 2 (R^TAB^TR)_{ii} \quad \forall \,i\in\{1,\dots,n\}.
\end{align}
Eq.~(\ref{eq:condition2-system}) allows us to solve the constrained equation (\ref{eq:condition2-maintext}) to find that the PDF is non-zero exclusively at the points
\begin{align}
    \label{eq:non-zeros}
    \mathbf{x}=\sqrt\pi R^{-T}(\mathbf{t}+2\mathbf{m})+\mathbf{c}.
\end{align}

\subsection{Algorithm}

\label{sub:algorithm}

In this Subsection, we present the detailed algorithm for efficiently calculating the full PDF of the measurement outcomes for the class of circuits described, directly from the associated symplectic matrix $\sympmat$ and displacement vector $\vec c$.
We also provide a rigorous analysis of the algorithm's computational complexity.

By combining the two expressions Eq.~(\ref{eq:periodicitymultimode}) and Eq.~(\ref{eq:non-zeros}) in the previous two Subsections, we find that the non-zero points are all related by the periodicity relation. Hence, the value of all non-zero points is precisely the same.
We can therefore characterise the full and exact PDF of the multimode measurement as
\begin{align}
    \label{eq:pdf-paperC}
    \pdfunc[\mathbf x]=\sum_{\mathbf m\in \mathbb Z^n}\delta(\mathbf x-\sqrt\pi R^{-T}(\mathbf{t}+2\mathbf{m})-\mathbf c).
\end{align}

We now describe the algorithm to evaluate the PDF in polynomial time.
Given the symplectic matrix in block form, i.e., Eq.~(\ref{eq:symp-block}), we need to evaluate $R^{-T}$, where $R$ is defined in Eq.~(\ref{eq:mat-R}) and $\mathbf{t}$, defined in Eq.~(\ref{eq:t-vec}), in order to write the PDF.
$R^{-T}$ is given in terms of $\halfsymp^T$ and $V$, where $V$ is the unimodular matrix arising from the Smith decomposition of $v \halfsymp$ given in Eq.~(\ref{eq:snf}).
The vector $\mathbf{t}$ is expressed in terms of $R$ in Eq.~(\ref{eq:t-vec}), however, it can be expressed in terms of the diagonal elements of the product of two upper left and right blocks of $V$, i.e., $V^{(1,1)}$ and $V^{(2,1)}$, respectively, as\footnote{See Eq.~(B33) in Appendix B of Paper \pC.}
\begin{align}
    T=V^{(1,1)T}V^{(2,1)}.
\end{align}

To find the matrix $V$, we first need to calculate the lowest common multiple, i.e., $\operatorname{lcm}(\cdot,\dots,\cdot)$, of all the denominators, i.e., $\operatorname{den}(\cdot)$, of the elements $\halfsymp$.
Formally, we could write
\begin{align}
    \label{eq:appendix-sigma-formal}
    v = \operatorname{lcm}(\operatorname{den}(\halfsymp_{1,1}),\dots,\operatorname{den}(\halfsymp_{1,n}),\operatorname{den}(\halfsymp_{2,1}),\dots,\operatorname{den}(\halfsymp_{2n,n})).
\end{align}
Then we multiply the matrix $\halfsymp$ by $v$ to produce an integer matrix $v \halfsymp$.
We can perform a Smith normal form (SNF) decomposition~\cite{newman1997} on this matrix to identify the $2n\times 2n$ unimodular matrix $V$, the $2n\times n$ diagonal matrix $\diagmat$ and the $n\times n$ unimodular matrix $U$,
\begin{align}
    \label{eq:appendix-direct-snf}
    v \halfsymp=VDU.
\end{align}
We can discard the matrices $\diagmat,U$.

As explained in more detail in Paper \pC, the transpose-inverse of $R$ can be directly evaluated as
\begin{align}
    R^{-T}=\tilde S^TV^{-T}\mqty(\mathbbm 1 \\0)=\mqty(A& \frac 1 2 B)V^{-T}\mqty(\mathbbm 1 \\0).
\end{align}
Furthermore, the matrix $T$ can be calculated from $V$ as
\begin{align}
    T=&V^{(11)T} V^{(2,1)}
\end{align}
and the vector $\mathbf{t}$ is simply the diagonal entries of $T$.
The PDF is then given by Eq.~(\ref{eq:pdf-paperC}).

We express the algorithm formally in Algorithm~\ref{alg:pdf-from-symplectic}. To summarise the computational cost of running the algorithm, consider the following analysis.
\begin{itemize}
  \item Steps 1-3 require \(\mathcal O(n^2)\) operations.
  \item Smith decomposition (Step 4) and the matrix multiplications (Steps 5-6) each require \(\mathcal O(n^3)\) operations.
  \item Final assembly of \(\mathbf t\) (Steps 7-8) requires \(\mathcal O(n^2)\) operations.
\end{itemize}
Applying each step, we find that we require a maximum of $\mathcal O(n^3)$ operations in total.
We can therefore conclude that the entire algorithm for finding the PDF is polynomial in the number of modes $n$.

\begin{algorithm}[htbp]
\label{eq:pdf-algo}
\caption{Efficient construction of the PDF from the description of the Gaussian operation}
\label{alg:pdf-from-symplectic}
\KwIn{Symplectic matrix $S\in\mathbb R^{2n\times2n}$,\quad
  displacement vector $\vec c\in\mathbb R^{2n}$.}
\KwOut{Inverse-transpose of the matrix $R$ and shift vector $\mathbf t$, which fully characterises the PDF given in Eq.~(\ref{eq:pdf-paperC}).
}
\BlankLine
\algcol{${\halfsymp} \;\leftarrow\;\bigl(A^T,\tfrac12\,B^T\bigr)^T ;$}{Form half-symplectic block $\halfsymp$ from $S$.}

\algcol{$v \;\leftarrow\; \mathrm{lcm}\bigl\{\operatorname{den}(\halfsymp_{ij})\bigr\} ;$}{Lowest common multiple of all\\ denominators.}

\algcol{$\tilde S' \;\leftarrow\; v\,\halfsymp ;$}{Scale to integer matrix.}

\algcol{$U, \Sigma, V \;\leftarrow\; \operatorname{SNF}(\tilde S') ;$}{Compute SNF of integer\\ matrix.}

\algcol{$V^{-T} \;\leftarrow\; (V^{-1})^T ;$}{Inverse-transpose of the right\\ unimodular factor.}

\algcol{$R^{-T} \;\leftarrow\; \halfsymp^{T}\,V^{-T}\,\begin{pmatrix}\mathbbm1_n&0_n\end{pmatrix}^T ;$}{Evaluate $R$ from $\tilde S$ and $V$.}

\algcol{$T \;\leftarrow\; V^{(1,1)\,T}\,V^{(2,1)} ;$}{Evaluate $T$ from blocks of $V$.}

\algcol{$t_i \;\leftarrow\; T_{ii} ;$}{Diagonal entries of $T$ give the shifts.}
\end{algorithm}

This result signifies that the circuits under consideration, specifically those composed of initial GKP stabiliser states, rational symplectic operations, and homodyne detection, are strongly simulatable classically.

\clearpage{}
    \clearpage{}\chapter{Simulation of realistic Gottesman-Kitaev-Preskill states}

\label{ch:simulation-of-realistic-gottesman-kitaev-preskill-states}

\DropCap{W}{hile}{-0.3} the results of the previous Chapter demonstrated efficient simulatability of a large class of circuits involving infinitely squeezed states that are otherwise impossible to simulate with classical devices, experimental circuits will never be able to create perfectly ideal GKP states. However, a good approximation to ideal GKP states (for sufficiently high squeezing) are the realistic GKP states, described in Eq.~(\ref{eq:realistic-gkp-states}). These states can be considered a generalisation of ideal GKP states, such that when $\Delta\to 0$, we recover the ideal case. This Chapter explores a recent method to simulate circuits involving realistic GKP states, which is more suited to experimental applications.

Before our work in Paper \pE, no algorithm existed to simulate computations with GKP codes in a practical time in the regime of interest for fault tolerance, namely high --- but not infinite --- squeezing. 
Standard celebrated simulation techniques~\cite{mari2012,veitch2012,rahimi-keshari2016}, relying on phase-space representations like the Wigner function, face intractability issues stemming from the large negativity associated with the highly non-Gaussian GKP states.
Others, such as Refs. \cite{bourassa2021, hahn2024, dias2024a}, have considered simulating GKP states, however, these algorithms work best with low squeezing.

Furthermore, the techniques in Chapter~\ref{ch:simulation-of-ideal-gottesman-kitaev-preskill-states} are, as far as we are aware, not possible to generalise to the case of finitely squeezed states. While ideal GKP states possess perfect translational symmetry and have a stabiliser group that maintains a consistent structure under the evolution of Gaussian operations, this is not true for finitely squeezed states.
For example, these realistic GKP states only approximately preserve the stabiliser group under Gaussian evolution~\cite{royer2020}.

To address the problem of realistic GKP states, our framework is based on the very recently introduced quasiprobability distribution, the ZGW function~\cite{davis2024} (see Section~\ref{sec:zak-gross-wigner-function}), whose theory we substantially develop and extend to multiple modes.

The remainder of this Chapter is structured as follows.
In Section~\ref{sec:techniques}, we introduce our multimode extension of the ZGW function, originally defined for a single mode, along with its key properties.
An analysis of the ZGW of realistic GKP states is given in Section~\ref{sec:zak-gross-wigner-function-of-realistic-gkp-states}. Finally, our novel simulation algorithm, along with details of its steps and complexity analysis, is presented in Section~\ref{sec:simulation-algorithm}.

\section{Multimode Zak-Gross Wigner function}

\label{sec:techniques}
In this Section, we introduce our multimode extension to the ZGW function, which is required for the simulation algorithm defined in Paper~\pE. We begin with its definition before showing that it satisfies the modified Stratonovich-Weyl axioms and covariance under GKP-encoded Clifford operations.

We first define the scaled displacement operators $\hwop*[\mathbf a]$, which are displacements in CV phase space but also reflect the structure of the odd-dimensional Weyl operators $\hwop$ introduced in Eq.~(\ref{eq:hwop}). Specifically we define
\begin{align}
    \hwop*[\mathbf a]= \disp[-\peakspace[d]\mathbf a],
\end{align}
which yields the relation
\begin{align}
    \hwop*[\mathbf a]=e^{i \pi \mathbf{a_X}^T\mathbf{a_Z}/d}\hwop*[\mathbf{a_X}]\hwop*[\mathbf{a_Z}]
\end{align}
in parallel to the equivalent expression for DV operators given in Eq.~(\ref{eq:hwop}). It also has the related commutation relation
\begin{align}
    \hwop*[\mathbf a]\hwop*[\mathbf b]=e^{2\pi i[\mathbf a,\mathbf b]/d}\hwop*[\mathbf b]\hwop*[\mathbf a].
\end{align}

We define the odd-dimensional ZGW function as
\begin{align}
    \label{eq:hwop-comm-cv}
    \zgw[\boldsymbol\eta][\hat \rho]= \Tr(\hat \rho \hat A_{\boldsymbol\eta}) ,
\end{align}
where the phase point operator $\hat A_{\boldsymbol\eta}$ is defined as
\begin{align}
\label{eq:phasepoint}
\hat A_{\boldsymbol\eta}=&\frac{1}{(2\pi)^n}\sum_{\mathbf m\in\mathbb Z^{2n}} e^{i\ell[\mathbf m,{\boldsymbol\eta}]+i\pi \mathbf{m_X}^T\mathbf{m_Z}}\hwop*[\mathbf m],
\end{align}
with ${\boldsymbol\eta} \in [0,d\ell)^{\times 2n}$ and with $d$ a positive odd integer.

\subsection{Modified Stratonovich-Weyl axioms}

\label{sub:modified-stratonovich-weyl-axioms}

We now introduce the modified Stratonovich-Weyl axioms, extending the framework presented in Ref.~\cite{davis2024} to multimode systems.

Compared to the standard Stratonovich-Weyl axioms (introduced initially in Ref.\cite{stratonovich1956} and discussed in Ref.\cite{brif1999}), these modified axioms, introduced for the single-mode case in Ref.~\cite{davis2024}, constitute a weaker set of conditions.
The most significant difference is the relaxation of the strict one-to-one correspondence between operators and phase-space functions that characterises the standard formulation.
While less general than the standard axioms, this modified set still provides crucial properties, including (a weaker version of) linearity, reality, standardisation, covariance under displacement, and a modified version of traciality.
The multimode ZGW function is shown to satisfy these modified axioms in Paper \pE.

We begin with the same axioms as the original ones defined in Subsection~\ref{sub:properties-of-the-wigner-function}. In particular, the ZGW function must be \textbf{real} for Hermitian operators, i.e.,
\begin{align}
    \zgw[\cdot][\hat B^\dagger]=(\zgw[\cdot][\hat B])^*.
\end{align}
It must also satisfy \textbf{standardisation}, whereby
\begin{align}
    \int \dd \boldsymbol{\eta} \zgw[\boldsymbol{\eta}][\hat B]=\Tr(\hat B),
\end{align}
in addition to \textbf{covariance} under the Weyl operators, which means
\begin{align}
    W_{\hat T_{\mathbf b} \hat C \hat T_{\mathbf b}^\dagger}({\boldsymbol\eta})=W_{\hat C}({\boldsymbol\eta} + \ell {\mathbf b}) \quad \text{ for all }\quad{\mathbf b\in \mathbb R^{2n}}.
\end{align}
The final two axioms differ from their original. In particular, although we require \textbf{linearity} in the sense that
\begin{align}
    \zgw[\boldsymbol{\eta}][\hat B+\hat C]=\zgw[\boldsymbol{\eta}][\hat B]+\zgw[\boldsymbol{\eta}][\hat C],
\end{align}
we do not require a one-to-one correspondence. Finally, to satisfy the modified \linebreak Stratonovich-Weyl axioms, the ZGW function must be \textbf{linear}, i.e.,
\begin{align}
    W_{\hat C+\hat D}({\boldsymbol\eta})=W_{\hat C}({\boldsymbol\eta})+W_{\hat D}({\boldsymbol\eta}).
\end{align}
We do not require that the function be a one-to-one map.
Finally, it should satisfy a modified form of \textbf{traciality}, where we only require the property to hold if at least one of the operators has undergone a twirling map $\mathcal E(\cdot)$, defined as
\begin{align}
    \mathcal E(\hat C)=\int_R \dd \mathbf s \gkpproj(\mathbf s)\hat C \gkpproj(\mathbf s),
\end{align}
where $R=[0,\ell)^{\times 2n}$ is a Zak patch~\cite{shaw2024} 
and we define $\gkpproj[\mathbf s]=\hwop*[\mathbf s]\gkpproj \hwop*[-\mathbf s]$. Furthermore, $\gkpproj$ is the GKP projector defined in Eq.~(\ref{eq:gkp-proj}).
Specifically, for the ZGW function to satisfy this modified version of traciality, we must have
\begin{align}
    \int \dd {\boldsymbol\eta}  \zgw[{\boldsymbol\eta}][\hat C]\zgw[{\boldsymbol\eta}][\hat D]=\Tr(\mathcal E(\hat C)\hat D)=\Tr(\hat C\mathcal E(\hat D))=\Tr(\mathcal E(\hat C)\mathcal E(\hat D)).
\end{align}

\subsection{Covariance under encoded Clifford operations}
We now briefly summarise the technique used to extend the covariance property discussed in the previous Subsection to all GKP-encoded Clifford operations. The full proof is given in Appendix C of Paper \pE.
Note that, due to the covariance property of the modified Stratonovich-Weyl axioms, we know that displacements in phase space can be easily tracked. In other words, given a displacement $\mathbf b$, the ZGW function itself is displaced proportionally.

To extend to GKP Clifford operations, first note that, as discussed in Subsection~\ref{sub:introduction-to-the-stabiliser-formalism}, the set of Clifford operations is generated by the phase gate $\phase[d]$, the Fourier transform $\fourier[d]$ and the SUM gate $\sumgate$. Each of these operations is achieved using Gaussian operations in the GKP encoding. Specifically, the phase gate is achieved using the CV shear gate $\phase*$, which is represented by the symplectic matrix
\begin{align}
    \mqty(1&0\\1&1),
\end{align}
the $d$-dimensional Fourier transform is achieved with the CV Fourier transform, which has the symplectic matrix
\begin{align}
    \mqty(0&1\\-1&0),
\end{align}
and finally, the SUM gate is achieved with the CV SUM gate, which is represented across two modes with the symplectic matrix
\begin{align}
    \mqty(1&0&0&0\\1&1&0&0\\0&0&1&-1\\0&0&0&1).
\end{align}
As we explicitly show in Paper \pE, these matrices together are a set of generators for the symplectic group $\operatorname{Sp}(2n,\mathbb Z_d)$ over the ring of integers modulo $d$~\cite{hua1949}.

To prove covariance under Clifford operations, we use that
\begin{align}
\label{eq:wig-covar}
    \zgw[\boldsymbol{\eta}][\hat U\hat \rho \hat U^\dagger]=&\frac{1}{(2\pi)^n} \sum_{\mathbf a\in \mathbb Z^{2n}} \Tr(\hat U\hat \rho \hat U^\dagger e^{i\ell[\mathbf a,{\boldsymbol\eta}]+i\pi \mathbf{a_X}^T\mathbf{a_Z}}\hwop*[\mathbf a]) \nonumber \\
    =&\frac{1}{(2\pi)^n} \sum_{\mathbf a\in \mathbb Z^{2n}} \Tr(\hat\rho e^{i\ell[ \mathbf a, {\boldsymbol\eta}]+i\pi \mathbf{a_X}^T\mathbf{a_Z}}\hwop*[S\mathbf a]).
\end{align}
We then note that for integer symplectic matrices, we have the relation
\begin{align}
    (S\mathbf a)^T_X(S\mathbf a)_Z=\mathbf{a_X}^T\mathbf{a_Z}+\mathbf t^T\mathbf a \mod 2,
\end{align}
where $\mathbf t$ is an integer vector derived from the block elements of the symplectic matrix, as detailed in Appendix C of Paper \pE. This equation can be substituted into Eq.~(\ref{eq:wig-covar}) and rearranged to find
\begin{align}
    \frac{1}{(2\pi)^n} \sum_{\mathbf b\in \mathbb Z^{2n}} \Tr(\hat\rho e^{i\ell[ \mathbf b,S\boldsymbol\eta-{\mathbf t}]+i\pi \mathbf{b}_{X}^{T}\mathbf{b}_{Z}}\hwop*[\mathbf b])=\zgw[S\boldsymbol{\eta}-\mathbf t][\hat \rho].
\end{align}
Therefore, we see that the action of Clifford operations transforms the parameter $\boldsymbol{\eta}$ of the ZGW function according to a symplectic transformation and an integer displacement described by the vector $\mathbf t$.

\section{Zak-Gross Wigner function of realistic Gottesman-Kitaev-Preskill states}

\label{sec:zak-gross-wigner-function-of-realistic-gkp-states}

The ZGW function provides a useful phase-space representation for analysing the properties of GKP states and, as we will see, developing simulation techniques for them.
We remind the reader that the ZGW function of ideal GKP states was already discussed in Subsection~\ref{sec:zak-gross-wigner-function}. There, we saw that the ZGW function of ideal stabiliser GKP states is always positive in odd dimensions. This has important implications for simulatability, as we will see in the next Section. Before describing the simulation algorithm, we first calculate the ZGW function for realistic GKP states.

The ZGW function representation for a single realistic 0-logical GKP state can be represented in terms of a multidimensional Jacobi theta function $\jacobi$, defined in Eq.~(\ref{eq:siegel}).
Specifically, the ZGW function of such a state with squeezing parameter $\sq$ and qudit dimension $\quditd$ is given by
\begin{align}
    \zgw[u,v][{\text{GKP}}]\propto \jacobi(\Gamma;\mathbf z).
\end{align}
This expression accounts for the finite squeezing parameter $\sq$, distinguishing it from the ideal GKP state scenario where $\sq \to 0$.
Here, the parameters $\mathbf z$ and $\Gamma$ are defined in terms of the phase-space variables $(u,v)$ and state parameters as $\mathbf z=(v/(d\ell),-u/(d\ell),0,0)^T$ and 
\begin{align}
    \Gamma=\frac 1 2 \mqty(\frac{i}{d\Delta^2}& 1 & -\frac{i}{\Delta^2}&\frac{i}{\Delta^2}\\
    1& \frac{i\Delta^2}{d} &1 & 1\\
    -\frac{i}{\Delta^2}&1&\frac{i(2d\pi +4d\pi \Delta^4)}{2 \pi \Delta^2}&-\frac{id}{\Delta^2}\\
    \frac{i}{\Delta^2}&1&-\frac{id}{\Delta^2}&\frac{id(\pi+2\pi\Delta^4)}{\pi \Delta^2}).
\end{align}

Unlike the ideal GKP state, whose ZGW function is concentrated on a grid of points, the finite squeezing introduces a spreading and negativity characterised by the structure of $\Gamma$. We provide examples of the ZGW function compared with its Wigner function in Fig.~\ref{fig:wigners} for the $0$-logical GKP state $\gkpket[0][3]$, the GKP-encoded strange state from Eq.~(\ref{eq:strange}) --- both with 12 dB squeezing --- and the vacuum state. As can be seen visually, the negativity of the CV Wigner function is inversely proportional to the negativity of the ZGW function for the stabiliser $0$-logical GKP state with varying squeezing $\Delta$. However, for the magic state, both the CV Wigner and ZGW functions exhibit negativity.

\begin{figure}[h!]
     \centering
     \includegraphics[width=0.75\textwidth]{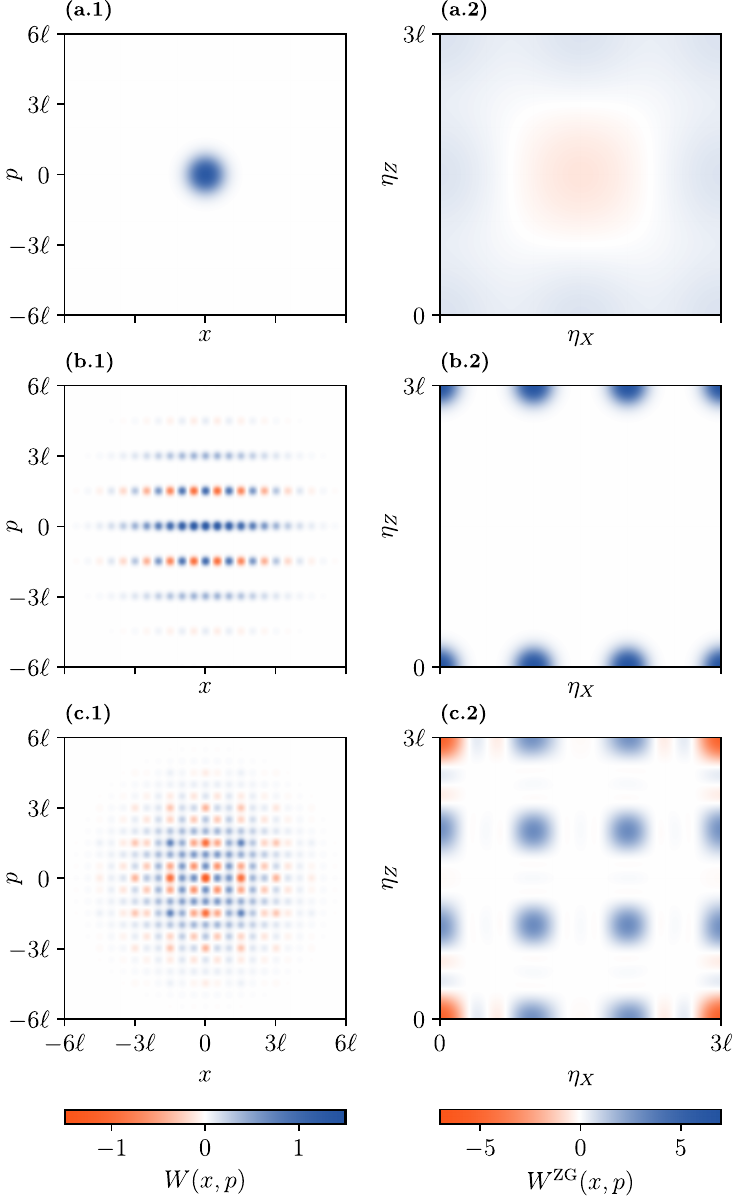}
     \caption{Here we demonstrate the difference between the Wigner and ZGW functions for three different states. \textbf{(a.1)} and \textbf{(a.2)} show the Wigner and ZGW function of the vacuum state, \textbf{(b.1)} and \textbf{(b.2)} show the Wigner and ZGW function of the qutrit zero-logical state, and \textbf{(c.1)} and \textbf{(c.2)} show the Wigner and ZGW function of the qutrit strange state. }
     \label{fig:wigners}
    \end{figure}

\section{Simulation algorithm}

\label{sec:simulation-algorithm}

Here, we describe the simulation algorithm introduced in Paper \pE.
First, we describe a method to simulate ideal stabiliser GKP states using the ZGW function. We then show that this can be generalised, using similar results to Ref.~\cite{pashayan2015}, to the case of states with ZGW negativity, such as realistic GKP states.

First, note that the ZGW function of ideal stabiliser GKP states is positive and that the GKP encoded Clifford operations can be understood as symplectic transformations and displacements of the phase space. Furthermore, if a measurement operator $\hat M$ is positive in this representation, and satisfies the condition that $\mathcal E(\hat M)=\hat M$, then we can use a phase space argument akin to the one presented in Subsection~\ref{sub:classical-simulation-gaussian} in order to prepare samples in a time that is polynomial with the number of modes. We will now see that the modular position measurement operators satisfy this condition. Hence, this provides an alternative proof to the techniques in Chapter~\ref{ch:simulation-of-ideal-gottesman-kitaev-preskill-states} that ideal GKP stabiliser states, followed by GKP encoded Clifford operations and modular measurements, are classically efficiently simulatable.

The proof that the modular position (equivalently, momentum) satisfies this condition when measuring a single mode is given in Appendix C3 of Paper~\pE. It relies on the fact that $\mathcal E$ leaves any $\hwop*[\mathbf a]$ that commutes with $\gkpproj$ invariant, i.e., if $[\hwop*[\mathbf a],\gkpproj]=0 \implies \mathcal E(\hwop*[\mathbf a])=\hwop*[\mathbf a]$. Modular position measurements can be expressed with the Kraus operator
\begin{align}
    \label{eq:mod-pos}
    \hat M_Z(\mathbf s)=\frac{1}{d\ell}\sum_{n\in\mathbb Z}e^{is\ell n} \lZ^n,
\end{align}
which by linearity means that it is invariant under the map $\mathcal E$. This means that we can evaluate the Born rule to measure $\hat M_Z(\mathbf s)$ on a state $\hat \rho$, which is expressed in terms of its ZGW function, as
\begin{align}
    \pdfunc[\mathbf s]=\Tr(\hat \rho \hat M_Z(\mathbf s))=\int \dd\mathbf r \, \zgw[\mathbf r][\hat \rho]\zgw[\mathbf r][\hat M_Z(\mathbf s)].
\end{align}
In addition, we can find the ZGW function of the Kraus measurement operator $\hat M_Z(\mathbf s)$ to be equal to the Dirac delta function, i.e.,
\begin{align}
    \zgw[\mathbf r][\hat M_Z(\mathbf s)]=\delta^{2n}(\mathbf r-\mathbf s).
\end{align}

Since the ZGW function of ideal stabiliser GKP states is positive, is mapped covariantly under GKP-encoded Clifford operations, and can be measured using modular measurements --- which are also positive in the ZGW function representation --- the entire circuit can be simulated efficiently. To see why, consider the same simulation algorithm presented in Subsection~\ref{sub:classical-simulation-gaussian}. All the same arguments apply; hence, the circuit is simulatable in polynomial time.

Furthermore, if the ZGW function has small amounts of negativity, we can use the same arguments as in Section~\ref{sec:classical-simulation-beyond} to prove that the following circuits have a PDF that can be estimated in a time with small exponential overhead. In particular, states that have ZGW negativity $\mathcal M_{\hat \rho}$, acted upon by GKP-encoded Clifford operations and continuous displacements $\paperclass[A]$ have a PDF which requires
\begin{align}
    N= \frac{2}{\epsilon^2}\mathcal M_{\hat \rho}^2 \log(2/\delta)
\end{align}
samples to approximate the true PDF with additive error $\epsilon$ and success probability $1-\delta$. Note that, practically, we must discretise the measurement outcomes to produce a probability function that can be estimated in finite time. This is discussed in more detail in Paper~\pE.

\clearpage{}
    \clearpage{}\chapter{A framework for universality in continuous-variable quantum computation}

\label{ch:a-framework-for-universality-in-continuous-variable-quantum-computation}

\DropCap{E}{fforts}{0.2} to develop quantum computers crucially rely on identifying the resources required for achieving computational advantage over classical devices. In the framework of resource theories~\cite{chitambar2019} as discussed in Section~\ref{sec:resource-theories}, this often involves identifying a restricted class of classically efficiently simulatable circuits and the ``resource'' states required to promote them to universality.

The identification of these enabling states holds both fundamental and practical significance.
First, from a fundamental point of view, it sheds light on the precise resources that enable QA, helping to refine our understanding of which non-classical features are required to perform calculations faster than classical computers. Second, from a practical point of view, it informs the design of future quantum devices, guiding experimental efforts toward architectures likely to outperform classical computation.

Resource theories for qubits have been studied in detail; however, the picture is much less clear in CV.
For qubits, it is known that to achieve universality, we must have the ability to perform Clifford circuits and access a set of magic gates, as discussed in Section~\ref{sec:universality}.
In the case of CVs, the same theorems do not directly apply. For instance, there is no all-Gaussian analogue to DV MSD for CV, which has remained an open problem for a long time~\cite{noh2020b}.
Instead, it is known that certain sets of states will lead to (either DV-encoded or true CV) universality when combined with Gaussian circuits. Here we focus on achieving universality in the DV sense but using CVQCs.

Since resource theories are often built upon the results of simulatability theorems, the results of the previous two Chapters naturally lead to the development of a framework for studying resources for the efficiently simulatable class of SGKP circuits. We use this framework to study the transition from classically simulatable CVQC regimes to UQC.
We explore the conditions required to break classical simulatability, focusing on the role of both Gaussian and non-Gaussian resources.
The central aim is to identify the minimal requirements for achieving a QA using CV systems, particularly those based on GKP states.

Through this lens, we find a highly counterintuitive result, which was initially reported in Paper \pC outside the context of a rigorous resource-theoretic framework. We found that within the classically efficiently simulatable class of SGKP circuits, UQC can be achieved by simply including the vacuum state in the otherwise simulatable set.
This leverages the result from Ref.~\cite{baragiola2019} that the same GKP architecture, combined with the vacuum state, is universal for QC.
This finding highlights the subtle and context-dependent nature of computational resources in CVQC, demonstrating how a state typically considered free or non-resourceful in standard Gaussian circuits acts as a resource in other contexts.

Inspired by this result, we later investigated a more general framework to assess the ability of general CV states to promote SGKP circuits to universality. Using this framework, detailed in Paper \pD, we were able to prove the first sufficient condition for universality using CV circuits. In addition, we assessed the ability of different circuit classes, including realistic GKP states, to promote circuits to universality in a quantitative way.

This Chapter begins with an explanation of our result that the vacuum provides QA to the otherwise simulatable class of SGKP circuits. We then introduce the framework designed to bridge the crucial gap between the classical simulatability of CV quantum circuits and the conditions required for achieving UQC.

\section{Vacuum provides quantum advantage}

\label{sec:vacuum-provides-quantum-advantage}

As we explained in Subsection~\ref{sub:gottesman-kitaev-preskill-gkp-codes}, the set of circuits containing GKP states along with Gaussian circuits was shown to be universal in Ref.~\cite{baragiola2019}. This result was an important realisation for the experimental feasibility of GKP circuits. Later, in Paper \pC, we proved that this result, in combination with our simulation algorithm described in Section~\ref{sec:stabiliser-formalism}, implies that the vacuum acts as a resource for SGKP circuits.
This is a highly counterintuitive result, because in the sense of Gaussian circuits, the vacuum state is considered completely resourceless. It is also one of the easiest states to prepare in quantum optics. How can it be that the simplest state provides a path to achieve UQC and the ability to unlock QA?

To understand this result, it is necessary to understand the result of Ref.~\cite{baragiola2019}. The main insight of their work is that when performing GKP error correction on the vacuum state, the resulting state is almost always magic. The quality of the magic state depends on the measurement outcomes; however, it can almost always be distilled to a perfect magic state through MSD (see Section~\ref{sec:resource-theories}).

More specifically, applying the error-correction gadget shown in Fig.~\ref{fig:error-correction} gives a state which has a Bloch vector $r_{\mu}(\mathbf t)$ as specified exactly in Eq.~(6) of Ref.~\cite{baragiola2019}. 

This state has a fidelity with the $H$ state, which has Bloch vector $\mathbf r^{H}=(1,0,1)^T/\sqrt 2$, that is above the MSD threshold for all values of $\mathbf t$, except for a zero measure set~\cite{baragiola2019}. Furthermore, the PDF of measuring each $\mathbf t$ are roughly equally likely, hence, in practice, the chance of measuring the values of $\mathbf t$ for which the state is non-magic is zero.
Following the production of a series of states, they can be distilled to a high-quality $H$ state. The MSD circuit consists only of Clifford operations~\cite{bravyi2005}.

The simulatability results of Paper \pC combined with the results of Ref.~\cite{baragiola2019}, are sufficient to prove the main claim of Paper, i.e., that the vacuum can promote an otherwise simulatable model to universality. Specifically, it was shown in Paper \pC that the set of SGKP circuits is efficiently simulatable. From Ref.~\cite{baragiola2019} we see that it is possible to achieve universality using only rational symplectic unitaries and displacements. Specifically, the GKP error correction gadget, the set of Clifford operations and the measurements are all contained within the SGKP set of circuits. Hence, we find that the only difference between circuits that are efficiently simulatable and universal is the addition of the vacuum state.

\section{A general sufficient condition for achieving universality}

\label{sec:a-general-sufficient-condition-for-achieving-universality}

Following the discovery that the vacuum could promote SGKP circuits to universality, the obvious question to ask was: what other states could do the same? A partial answer to this question was already given in Paper \pC, where we demonstrated that even realistic GKP states could unlock the ability to perform UQC. However, the majority of these results were presented in Paper \pD, which contains a resource-theoretic investigation.

Given the fact that SGKP circuits are efficiently simulatable, we can define the circuit class as a free set in a resource theory. In other words, we consider all stabiliser GKP states, rational Gaussian operations, and homodyne measurements to be resourceless. 
This choice of resourceless model is motivated by the fact that the set is simulatable. By understanding the resources that must be added to the set to make it universal, we expect to be able to identify whether newly invented quantum algorithms will be capable of achieving QA. Alternatively, by constructing circuits using resourceful states, this resource theory could guide the discovery of entirely new algorithms.

\subsection{General mapping from CV to DV}

\label{sub:general-mapping-from-cv-to-dv}

We consider the set of maps from CV to DV that are implementable using only components selected from the set of SGKP circuits. This set of maps can be expressed as~\cite{calcluth2024}
\begin{align}
\label{eq:general-map-def}
    M_{P}:\quad \hat \rho \to \int_{\intreg} \dd \mathbf s \sum_i \genkraus[i][\mathbf s] \hat \rho \genkraus*[i][\mathbf s],
\end{align}
where we have defined $\hat P_i(\mathbf s)$ to be some --- possibly measurement dependent --- choice of Kraus operators and $\intreg$ to be a choice of an integrable subregion of the range of the parameters $\mathbf s$. The Kraus operators are selected to satisfy $ \gkpproj\genkraus[i]=\genkraus[i]$. I.e., $\genkraus[i]$ must always contain a final projection onto the GKP subspace. We denote the set of Kraus operators as $\setgenkraus=\{\genkraus[1][\mathbf s],\dots,\genkraus[k][\mathbf s]\}$. Despite this additional restriction, the mapping still technically maps CV states to CV states. We solve this by introducing the state $\mappedstate$, which is the state encoded by the CV state in the GKP basis. Formally, we can consider the state as being defined in terms of its coefficients as
\begin{align}
    \label{eq:general-map}
    \bra{l}\hat \rho^{(P)}_L\ket{l'}=\bra{l_{\text{GKP}}}\int_\intreg \dd \mathbf s \sum_i \genkraus[i][\mathbf s] \hat \rho \genkraus*[i][\mathbf s]\ket{l'_{\text{GKP}}}.
\end{align}

Despite introducing the restriction that $ \gkpproj\genkraus[i]=\genkraus[i]$, it is important to note that the GKP projection operator is not a valid SGKP operation. In other words, there is no way to apply the operator directly using only SGKP circuits. Instead, we consider the error correction Kraus operator $\gkpec[\mathbf s]=\gkpproj \disp[-\mathbf s]$, which is implementable by the gadget shown in Fig.~\ref{fig:error-correction}.

Next, we construct a broad subclass of the possible mappings. These maps are characterised by the Kraus operators
\begin{align}
\label{eq:Kraus:projector-QEC}
    \genkraus[\mathbf s]=\genkrausprime[\mathbf s]\hat \Pi \disp[-\mathbf s] \hat U_i
\end{align}
where $\genkrausprime[\mathbf s]$ is a probabilistic GKP-encoded Clifford operation which occurs after the error-correction routine. Furthermore, $\hat U_i$ is any unitary Gaussian operation.

\subsection{Resourcefulness of states}
Using the tool developed in the previous Subsection, we now demonstrate that this framework can be applied to assess the resourcefulness of CV states to promote SGKP circuits to universality.

Given that the set of mappings defined in Eq.~(\ref{eq:general-map-def}) is resourceless by definition, we find that the ROM of the logical state obtained after the mapping is no more resourceful than the original CV state. In fact, in principle, we can define a well-grounded resource-theoretic function to quantify the resourcefulness of an arbitrary single-mode CV state to promote SGKP circuits to universality as
\begin{align}
    \label{eq:gkp-rom}
    \mathcal R_{\text{SGKP}}(\hat \rho)=\max_{P\in \mathcal{P}_{\text{SGKP}}}\mathcal R(M_P(\hat \rho)).
\end{align}
However, in practice, a full search over all possible mappings is not possible to calculate numerically.

Despite this limitation, the quantity is still useful for individual mappings, i.e.,
\begin{align}
    \mathcal R(M_P(\hat \rho)).
\end{align}
This is because any state projected onto a qubit state with a ROM above the threshold for distillation $\mathcal R^*$ will be sufficient to promote SGKP circuits to universality. This quantity also has a second useful feature. Since we know that the ROM of a CV state after a specific mapping will, by definition, be less than or equal to the quantity in Eq.~(\ref{eq:gkp-rom}), the value can be used to define a lower bound on the number of these states required in order to produce a specified number of ideal magic states. This is because the number of rounds of MSD required to produce an ideal magic state scales with the fidelity of the initial state to the target state~\cite{bravyi2005}.

\subsection{Mapping applied to stabiliser subsystem decomposition}

The stabiliser SSD is contained within the set of mappings defined in Eq.~(\ref{eq:general-map-def}). Specifically, we choose $P=\{\gkpec\}$ where $\gkpec$ is the GKP error correction Kraus operator defined in Eq.~(\ref{eq:gkp-ec-kraus}) and $R$ is chosen to be the interval $[-\sqrt{\pi}/2,\sqrt{\pi}/2)^{\times 2}$. The stabiliser SSD is then recovered as $\hat\rho_L^{(K)}=\hat \rho_{\Pi}$.

In practice, this map can be implemented using GKP error correction~\cite{shaw2024}. We also see that since GKP error correction is contained within the class of SGKP circuits, it is a resourceless mapping, as we would expect. 

Finally, after applying the stabiliser SSD, we note that the ROM of CV states provides a convenient and theoretically-grounded method for a lower bound for assessing the resourcefulness of CV states to promote SGKP circuits to universality. We explore this fact in the next Section.

\section{Resource analysis of realistic GKP states}
In this Section, we provide an example of the utility of the framework developed in the previous Sections. In particular, we apply the framework to a specific class of CV states using the stabiliser SSD. This is just one of the mappings explored in Paper \pD, where we also investigate the modular SSD and a new decomposition called the Gaussian-modular SSD. Furthermore, in that Paper, we apply our framework to a variety of states, including Gaussian, cat, and cubic phase states. Here we explore the class of states most relevant to this thesis: realistic GKP states.

\begin{figure}[h]
\centering
\includegraphics[width=0.8\textwidth]{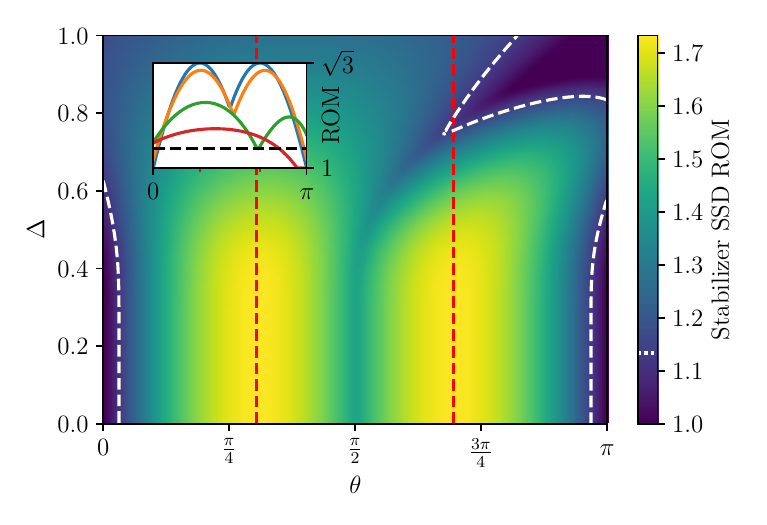}
\caption{The ROM of a the encoded qubit GKP state, as defined in Eq.~(\ref{eq:realistic-gkp-theta}), after performing the stabiliser SSD for different values of squeezing $\Delta$ and rotation angles $\theta$ and a fixed phase of $\phi=\pi/4$. The red dashed lines indicate the values of $\theta$ for which the state is an encoded $T$-state.  The large regions inside the dashed white boundaries in each plot indicate the regions of distillability. The inset plots show a subset of the same data plotted with $\theta$ on the $x$-axis and the value of ROM on the $y$-axis. The solid blue, red, green and purple lines correspond to $\Delta=0,1/2,3/4,1$, respectively. Equivalently, the lines for each increasing $\Delta$ have decreasing maxima. Note that the value of ROM in the main figures and the insets is always greater than or equal to $1$.}
\label{fig:stabssdrom}
\end{figure}
We investigate a set of realistic GKP states where the azimuthal angle is set to $\phi=\pi/4$ and the polar angle $\theta$ is a free parameter. This state is chosen as it represents a convenient cross-section of the possible GKP states. It can be expressed as
\begin{align}
    \label{eq:realistic-gkp-theta}
    \ket{\psi_{\text{GKP}}^\Delta(\theta)}=\frac{1}{\sqrt{\mathcal{N}_{\text{GKP}}}}\left(\cos(\theta/2)\ket{\bar 0_{\text{GKP}}^{\Delta}}+\sin(\theta/2)e^{i\pi/4}\ket{\bar 1_{\text{GKP}}^{\Delta}}\right)
\end{align}

For this state, $\theta=0$ and $\theta=\pi$ reproduce the logical $\gkpket[0]$ and $\gkpket[1]$ states, while the choices $\theta=\arccos(1/\sqrt{3})$ and $\theta=\pi-\arccos(1/\sqrt{3})$ yield an encoded $\gkpket[T]$ magic state and its orthogonal complement, respectively. Plotting the ROM $\mathcal{R}$ over the $(\theta,\Delta)$ plane in Fig.~\ref{fig:stabssdrom}, we find that the encoded $T$ states (vertical markers) achieve the highest ROM for any given $\Delta$, saturating the theoretical bound $\sqrt{3}$ as $\Delta\to0$. The shaded white contour shows the area where the ROM exceeds the threshold $\mathcal{R}^*$, revealing regions in which SGKP circuits are sufficient to promote SGKP circuits to UQC, generalising the results of Ref.~\cite{baragiola2019} and Paper \pC to a wider class of input states.

We see that at lower levels of squeezing $\Delta$, the ROM-level in the contour plot shows a pronounced skew, reflecting an asymmetry in how finite-squeezing affects the different logical states. This skew effect is detailed in Appendix G of Paper \pD.

\clearpage{}
    \clearpage{}\setcounter{chapter}{6}
\setcounter{section}{0}
\chapter*{Conclusion}
\fancyhead[LE]{Conclusion}
\hypertarget{toc}{}
\addcontentsline{toc}{chapter}{Conclusion} \label{ch:conclusion}
\DropCap{T}{his}{0.2} thesis has addressed the three key objectives introduced in Chapter \ref{ch:introduction}, which we list below. In summary, despite their high Wigner negativity, we have identified and characterised broad families of ideal-GKP circuits that remain efficiently simulatable with classical computers. We introduced three complementary simulation techniques for simulating ideal GKP states. We provided an alternative method to simulate realistic odd-dimensional GKP states when acted on by rational symplectic operations and continuous displacements by using the ZGW function. Furthermore, we provided a method to understand resources in CVQC that are sufficient for universality.

\section{Objectives}
Here, we discuss the contribution of each of the works in this thesis to its overarching goal. We provide a reminder of the objectives and a summary of how each work has contributed to addressing them.

\vspace{1.5em}
\begin{fancyquotes}
\textbf{Objective 1} (\textit{Classical simulation of ideal GKP states}): Building on a combination of techniques devised initially for qubits and new techniques devised from analytic number theory, we will show that circuits involving ideal GKP states can be efficiently simulated on a classical computer.
\end{fancyquotes}
\vspace{0.2em}

The first of these Objectives was addressed in Chapter~\ref{ch:simulation-of-ideal-gottesman-kitaev-preskill-states}, where we present the results of Papers~\pA, \pB, and \pC using three entirely different techniques. We established that broad classes of quantum circuits designed to utilise GKP states are, in fact, classically simulatable. The largest of these sets is dense in the set of Gaussian operations, and a visual overview of the relationships between each of these sets is presented in Fig.~\ref{fig:venn}. As a consequence of our results, we know that these circuits do not possess the capacity to achieve QA. Drawing this critical distinction, our work contributes to a clearer understanding of which quantum architectures can be considered computationally useful. 

\begin{fancyquotes}
\textbf{Objective 2} (\textit{Classical simulation of realistic GKP states}): We will demonstrate that it is possible to simulate realistic GKP states under certain conditions. Here, we expect that the algorithm will scale exponentially with respect to the number of modes. However, we aim to make this exponential overhead as small as possible, to the point that the simulation of practical circuits is tractable.
\end{fancyquotes}
\vspace{0.2em}

Addressing the second Objective, Chapter~\ref{ch:simulation-of-realistic-gottesman-kitaev-preskill-states} provides the first method to simulate realistic, i.e., finitely squeezed, GKP states in a reasonable time in the regime of high squeezing, which is relevant for fault-tolerance. 
Our algorithm is tailored to simulate circuits initialised with realistic odd-dimensional GKP qudits, undergoing encoded Clifford operations and followed by modular measurements.
While generic simulation methods face potential exponential overhead due to the significant negativity of realistic GKP states in standard representations, our approach leverages convenient properties of the ZGW function.
Specifically, the runtime of our algorithm scales with the negativity \textit{in the ZGW function representation}, which allows for efficient simulation for certain circuits despite the states' high negativity in the standard Wigner representation. 
This investigation directly contributes to the overall goal of the thesis of understanding the simulatability boundaries of CVQC.

\vspace{1.5em}
\begin{fancyquotes}
\textbf{Objective 3} (\textit{Framework to assess the resourcefulness of CV states}): We will develop a new resource theory for CV systems, introduce a method to measure the resourcefulness of states based on state decompositions, and establish a criterion for achieving universality.
\end{fancyquotes}
\vspace{0.1em}

In Chapter~\ref{ch:simulation-of-realistic-gottesman-kitaev-preskill-states}, we expanded on the topic of resources in CVQC. These results were inspired by the surprising result of Paper \pC that the vacuum state is a resource for QA in this model. Adding the vacuum state to the restricted architecture makes it a UQC that can perform calculations faster than classical computers. This is highly surprising because, in the context of Gaussian circuits, the vacuum is considered resourceless and is easy to produce experimentally.

To explore the idea of resources for SGKP circuits, we partitioned the components of general quantum circuits into two sets. The first resourceless set includes components handled efficiently by classical computers, according to our simulation algorithms demonstrated for the first Objective. The second, resourceful set comprises components that unlock abilities beyond what classical computers can do. We then addressed whether a smaller subset of the resourceful set exists that can achieve QA when combined with the resourceless set. While previous work has addressed this question for qubit architectures, our research provides the first sufficient condition for CV states to provide QA. To establish this, we created a mathematical framework that maps CV states onto qubit states using QEC techniques. This demonstrates a rigorous way to extract qubit quantum information from a CV state. We then leveraged known results in the resource theory of qubit QCs to establish a sufficient criterion for CV states.

\section{Implications}

Here we discuss the relevance of the contributions in this thesis for the pursuit of understanding QA. These contributions --- particularly those addressing Objectives $2$ and $3$ --- also provide insights for developing experimental devices. 

From a theoretical perspective, we provide a comprehensive set of results on the simulatability of quantum circuits involving GKP states, identifying extensive classes of circuits previously believed to require quantum resources. Notably, our investigation reveals that despite substantial Wigner negativity, certain broad families of ideal GKP circuits remain efficiently simulatable, fundamentally clarifying the nature of quantum resources required for achieving computational advantage.

Practically, our algorithm for simulating realistic GKP circuits is particularly valuable for state-of-the-art experiments attempting to implement GKP states in superconducting cavities, integrated photonics, and trapped-ion systems. This method directly addresses and overcomes previous limitations in simulation approaches such as Ref.~\cite{bourassa2021}, notably enabling efficient simulations in the critical high-squeezing regime necessary for fault-tolerant QC. As an illustrative example, our algorithm efficiently handles realistic stabiliser GKP states with 12 dB of squeezing for circuits containing up to one thousand modes, requiring less than twice the computational resources needed for single-mode cases --- a considerable improvement over existing simulators.

Moreover, this thesis enhances our conceptual understanding by identifying and rigorously characterising the role of the vacuum state as a surprising resource for universality in CVQC. Through a resource-theoretic framework, we establish the first sufficient criterion for CV states to yield QA. By explicitly connecting CV states to established qubit resource theories, we create a unified theoretical foundation that facilitates the precise quantification and exploitation of quantum resources across distinct quantum computing architectures.

\section{Limitations}

As seen in Chapter~\ref{ch:simulation-of-ideal-gottesman-kitaev-preskill-states}, our works addressing Objective~1 in Papers \pA, \pB, and \pC each prove different subsets of circuits to be efficiently simulatable. The third of these works proves the most extensive set to be simulatable, i.e., one that is dense in the complete set of Gaussian operations. While technically interesting, the significant and obvious limitation of these simulation algorithms is that they are incapable of simulating circuits that can be physically implemented.
However, this limitation was partially solved in Paper \pE with the first algorithm that can simulate realistic odd-dimensional GKP states in a time that scales with a small exponential overhead in the fault-tolerant regime.

At the same time, our results in Paper \pE are limited by the fact that only odd-dimensional circuits are simulatable with our technique.
However, qudit-GKP computation is also promising and is being pursued as an active field of research. For instance, Ref.~\cite{brock2025} experimentally demonstrates QEC with GKP qudits, including the odd-dimensional case of $d=3$, which is precisely the realm of application of our simulation framework.
More generally, QC based on qudits is an active and growing field offering several advantages over traditional qubit-based approaches. For example, qudits may unlock more efficient computational algorithms~\cite{wang2020}, and they provide a natural framework for simulating high-dimensional quantum systems, as recently demonstrated in Ref.~\cite{meth2025}. These benefits have led to increasing experimental interest in qudit-based architectures, which are now being actively pursued across a range of platforms~\cite{wang2020}.  
If qudits are to be useful in the long run, however, QEC will be necessary; hence, it is important to investigate circuits with GKP codes and consider the ability to simulate them classically.

We also note that the field of quantum information has a history of proving results for odd-dimensional qudits and then later resolving the same question for the often more complex case of qubits. For example, in 2014, Howard et al proved the connection between simulatability and contextuality in Ref.~\cite{howard2014}, and this was later shown to hold for qubits in Refs.~\cite{bermejo-vega2017,raussendorf2017}.

\section{Outlook}

The findings presented in this thesis have contributed to our understanding of the boundary between classical simulatability and QA, specifically in CVQC. In contributing to this understanding, we have uncovered many interesting open questions for future research. Here, we present some of those questions.

We begin with the most natural extensions of our results. First, it remains open whether CV circuits composed of only finitely squeezed GKP states, followed by all Gaussian operations and measurements, are necessarily universal, or even hard to simulate. Although this idea was explored and discussed in Paper~\pC, we were unable to provide a formal proof of this. However, it is expected from Ref.~\cite{baragiola2019} --- albeit, not formally proven --- that realistic GKP states in combination with Gaussian states will be universal. It would be interesting to explore these two regimes in future work. Similarly, an interesting avenue for future work would be exploring extending the resource-theoretic technique developed in Paper \pD to perform the stabiliser SSD using realistic GKP states.
Another natural extension of our work would be to investigate the ability to simulate realistic GKP states in even encoded dimensions, including for the encoded qubit case. This would be particularly relevant to cutting-edge experiments where the successful implementation of such states is becoming a reality.

Another interesting avenue to explore would be identifying a general sufficient condition for QA in CVQC. In particular, our results implicitly demonstrate that Wigner negativity of both the Wigner function and the ZGW function of some of the components in a circuit is necessary to achieve a QA. For qubit circuits, any state with a ROM above a certain threshold guarantees that a supply of those states will be sufficient for QA when combined with Clifford operations~\cite{bravyi2005}. However, it remains a fascinating open question whether some sufficient condition guarantees QA for CVQC with otherwise Gaussian circuits.

Finally, the connection between the computational power of CV and DVQC has been explored to some degree in Paper \pD. In particular, we demonstrated a universal method to map CV systems to DV systems. However, this mapping is neither unique nor injective and does not provide information about the fundamental connection between the two regimes. Understanding whether CV devices offer a fundamental advantage over DV devices, or whether they are ultimately computationally equivalent, is a pressing foundational problem.
On the one hand, if CVQC provides a computational edge, this insight will spur further advancements in CV-based hardware and algorithms. Alternatively, if a computational equivalence is demonstrated, it will enable scientists to convert between qubit-based and CV algorithms, optimising quantum algorithms for the most efficient hardware available. In either case, the answer to this question would have important implications for the field and beyond.\clearpage{}
    \fancyhead[LE]{BIBLIOGRAPHY}
    \fancyhead[RO]{BIBLIOGRAPHY}
    \makeatletter
        \bibliographystyle{mybibstyle}  
        \bibliography{mainbibtex}            
    
    \makeatother
    \cleardoublepage
    
    \fancyhead[LE]{\nouppercase{\leftmark}}
    \fancyhead[RO]{\nouppercase{\rightmark}}

    \begin{appendices}
    \renewcommand{\thechapter}{\Roman{chapter}}
\clearpage{}\chapter{Simulation with irrationals}
\label{appendix:irrationals}
In the following I explain why simulation of GKP states with irrational angles breaks down. I use the Wigner function to illustrate this, despite this not being the method we previously used to simulate GKP states. However, this is the most intuitive for single mode GKP states and gives the most insight into why it is not possible to write a meaningful PDF for the simulation of such a computation.

The Wigner function of a single mode $0$-logical GKP state is given by
\cite{garcia-alvarez2021}
\begin{align}
    W_{0,\text{{GKP}}}(q,p)=\frac{1}{4 \sqrt{ \pi }}\sum_{st}(-1)^{st}\delta\left( p-\frac{\sqrt{ \pi }s}{2} \right)\delta(q-\sqrt{ \pi }t).
\end{align}
A rotation of the state can be understood through a transformation of the quadrature operators  
\begin{align}
    \hat{q} \to \hat{q}\cos\theta+\hat{p}\sin\theta \nonumber\\
\hat{p} \to \hat{q}\sin\theta-\hat{p}\cos\theta.
\end{align}
As a result of this transformation, the Wigner function of the state is transformed to
\begin{align}
&W_{0,\text{{GKP}}}(q \cos \theta + p \sin \theta,q \sin \theta- p \cos \theta)=\frac{1}{4 \sqrt{ \pi }}\sum_{st}(-1)^{st}\delta\left(p -\frac{\sqrt{ \pi }s}{2} \right)\delta(q-\sqrt{ \pi }t) \nonumber\\
\implies &
W_{0,\text{{GKP}}}(q,p)=\frac{1}{4 \sqrt{ \pi }}\sum_{st}(-1)^{st}\delta\left(p\cos\theta - q \sin \theta -\frac{\sqrt{ \pi }s}{2} \right)\delta(q\cos \theta+p\sin\theta -\sqrt{ \pi }t).
\end{align}

The probability distribution of the measurement outcomes of a measurement of $\hat q$ can be found by integrating over $p$, while measurement outcomes of $\hat p$ can be found by integrating over $\hat q$. 

We find that the probability distribution over position $\hat q$ is given by
\begin{align}
\text{Pr}(\hat q=q)=&\int \, \text{d} p \, W_{0}(q,p)\\
=&\int \, \text{d} p\frac{1}{4 \sqrt{ \pi }}\sum_{st}(-1)^{st}\delta\left(p\cos\theta - q \sin \theta -\frac{\sqrt{ \pi }s}{2} \right)\delta(q\cos \theta+p\sin\theta -\sqrt{ \pi }t) \nonumber\\
=& \frac{1}{|\cos \theta|}\int \, \text{d} p\frac{1}{4 \sqrt{ \pi }}\sum_{st}(-1)^{st}\delta\left(p - q  \frac{\sin\theta}{\cos\theta} -\frac{\sqrt{ \pi }s}{2 \cos\theta} \right)\delta(q\cos \theta+p\sin\theta -\sqrt{ \pi }t)\nonumber\\
=& \frac{1}{|\cos \theta|} \frac{1}{4 \sqrt{ \pi }}\sum_{st}(-1)^{st}\delta\left(q\cos \theta+(q  \frac{\sin\theta}{\cos\theta} +\frac{\sqrt{ \pi }s}{2 \cos\theta})\sin\theta -\sqrt{ \pi }t\right)\nonumber\\
=& \frac{1}{|\cos \theta|} \frac{1}{4 \sqrt{ \pi }}\sum_{st}(-1)^{st}\delta\left(q\sec \theta+\frac{\sqrt{ \pi }s}{2 }\tan\theta -\sqrt{ \pi }t\right)\nonumber\\
=&\frac{1}{4 \sqrt{ \pi }}\sum_{st}(-1)^{st}\delta\left(q+\frac{\sqrt{ \pi }s}{2 } \sin\theta -\sqrt{ \pi }t \cos\theta\right),
\end{align}
while the PDF over $\hat p$ is given by
\begin{align}
    &\text{PDF}(\hat p=p)\nonumber\\
    =&\int \, \text{d} q \, W_{0}(q,p)\nonumber\\
    =&\int \, \text{d} q\frac{1}{4 \sqrt{ \pi }}\sum_{st}(-1)^{st}\delta\left(p\cos\theta - q \sin \theta -\frac{\sqrt{ \pi }s}{2} \right)\delta(q\cos \theta+p\sin\theta -\sqrt{ \pi }t) \nonumber\\
=& \frac{1}{|\cos \theta|}\int \, \text{d} q\frac{1}{4 \sqrt{ \pi }}\sum_{st}(-1)^{st}\delta\left(p\cos\theta - q \sin \theta -\frac{\sqrt{ \pi }s}{2} \right)\delta\left( q+p\tan\theta -\frac{1}{\cos\theta}\sqrt{ \pi }t \right) \nonumber\\
=& \frac{1}{|\cos \theta|}\frac{1}{4 \sqrt{ \pi }}\sum_{st}(-1)^{st}\delta\left(p\cos\theta -  \sin \theta \left( \frac{1}{\cos\theta}\sqrt{ \pi }t -p \tan\theta \right)  -\frac{\sqrt{ \pi }s}{2} \right) \nonumber\\
=& \frac{1}{|\cos \theta|}\frac{1}{4 \sqrt{ \pi }}\sum_{st}(-1)^{st}\delta\left(p\sec\theta - \sqrt{ \pi }t \tan \theta  -\frac{\sqrt{ \pi }s}{2} \right)\nonumber \\
=& \frac{1}{4 \sqrt{ \pi }}\sum_{st}(-1)^{st}\delta\left(p - \sqrt{ \pi }t \sin \theta  -\frac{\sqrt{ \pi }s}{2} \cos\theta\right) .
\end{align}

First, analyzing the PDF over $\hat q$, 
we write $\sin \theta / \cos\theta= \frac{u}{v}$. The delta function is non-zero whenever $q$ satisfies the following equation
\begin{align}
&q+\frac{\sqrt{ \pi }s}{2} \sin\theta- \sqrt{ \pi }t \cos\theta=0 \nonumber\\
\implies&\frac{q}{\cos\theta}+\frac{\sqrt{ \pi }s}{2} \frac{u}{v}- \sqrt{ \pi }t =0\nonumber\\
\implies&\frac{2qv}{\sqrt{  \pi }\cos\theta}+s u-2 t v=0\nonumber\\
\implies& c+s u - 2 t v =0 \label{eq:dio-for-q},
\end{align}
for $s,t \in \mathbb{Z}$, where we have implicitly defined $c$ in terms of $q,v$ and $\theta$ in the last two lines.
Note that for the PDF over $\hat p$ we also have a similar equation,
\begin{align}
&p-\sqrt{ \pi }t \sin\theta - \sqrt{ \pi } \frac{s}{2} \cos\theta =0\nonumber \\
\implies & \frac{p}{\cos\theta}-\sqrt{ \pi }t \frac{u}{v} - \sqrt{ \pi } \frac{s}{2} =0\nonumber\\
\implies & \frac{2pv}{\sqrt{ \pi }\cos\theta}-2t  u -  s  v =0\nonumber\\
\implies & c-2t  u -  s  v =0\label{eq:dio-for-p},
\end{align}
where we have now defined $c$ in terms of $p,v$ and $\theta$ in the last two lines.

We now try to find solutions for these equations. We start with the first equation, i.e., Eq.~(\ref{eq:dio-for-q}). We have two distinct cases.

\textbf{Case 1.} ($u,v \in \mathbb Q$) In this case, we first multiply the equation by the denominators of the fractions of $u,v$ such that the equation becomes of the form of Eq.~(\ref{eq:dio-for-q}) with $u,v \in \mathbb Z$. Therefore, in this case we assume that both numbers are integers.

If $u,v$ are integers and hence $\sin\theta/\cos\theta$ is rational then the equation is a linear Diophantine equation. The equation will have a solution iff $c$ is a multiple of $\gcd(u,2v)$. Hence, assuming we have expressed $\frac{u}{v}$ in the simplest form, the greatest common divisor will be $1$. In which case $c \in \mathbb{Z}$. 

Note also that $(-1)^{(s)(t)}$ is $1$ for odd and even $s,t$ and when both $s,t$ are even. It is only negative when $s,t$ are both odd. We can consider this to contribute a factor of two in the following expression. Therefore, we can write
\begin{align}
\text{Pr}(\hat{q}=q)=\int \, \text{d} p \, W_{0}(q,p)=&\frac{1}{2 \sqrt{ \pi }}\sum_{m}\delta\left(q -  \frac{\sqrt{ \pi }\cos\theta}{2 v}m\right)
\end{align}
In this case, we can provide a sensible answer to the question of how the probability distribution of the simulation will look, i.e., it has a periodic structure.

\textbf{Case 2.} (Either or both $u,v \notin \mathbb Q$.) In the case that either $u$ or $v$, or both $u,v$ are irrational we no longer have a Diophantine equation. Instead we make use of the following theorem.

\begin{theorem}
\label{theorem:kronecker}
    (Theorem 438 of Ref. \cite{hardy1960}) If $\alpha$ is irrational, $c$ is arbitrary and $N,\varepsilon$ are positive then there exist integers $n,m$ such that $n>N$ and
\begin{align}
|n\alpha- m -c | < \varepsilon.
\end{align}
\end{theorem}

\begin{corollary}
    If $\alpha,\beta$ are irrational, $c$ is arbitrary and $N,\varepsilon$ are positive, the set $\{n \alpha + m\beta -c : n,m \in \mathbb Z\}$ is dense in the reals.
\end{corollary}
\begin{proof}
    Following from Theorem \ref{theorem:kronecker}, we have that  
    \begin{align}
        & |n\alpha- m -c | < \varepsilon\\
        \implies &|n\alpha\beta- m\beta -c \beta| < \varepsilon\beta\\
        \implies &|n\alpha'- m\beta -c '| < \varepsilon'
    \end{align}
    whereby we have simply reparameterized $\alpha'=\alpha\beta$, $c'=c\beta$ and $\varepsilon'=\varepsilon\beta$. We can also rescale $\varepsilon'$ and since $m\in \mathbb Z$ we can change the sign of $m \to -m$.
\end{proof}

Therefore, $\varepsilon$ can be made arbitrarily small and hence we see that the combinations of $su-2tv$ is dense in the reals. At first, this may not appear to be a problem. (E.g., if we have a flat distribution.) However, as we will see, the PDF is also dense in the reals when measuring in the momentum basis.

This can be seen by noticing Eq.~(\ref{eq:dio-for-q}) has the same form as Eq.~(\ref{eq:dio-for-p}). Both have discrete solutions (i.e., that are not dense in the reals) if and only if $u/v \in \mathbb Q$. Therefore, the PDF for $q$ is non-zero at points that are dense in the reals whenever the PDF for $p$ has the same property.

This means we have a Wigner function which, when measured gives a set of outcomes which are dense in the reals in both the position basis and the momentum basis at the same time. However, there is no periodicity and there is no way to characterise the distribution of the peaks that occur in the probability distributions.
\clearpage{}
\end{appendices}

\end{document}